\documentclass[3p]{elsarticle} 
\usepackage{graphicx}

\date{}

\usepackage{fancyhdr}
\usepackage[dvipsnames]{xcolor}
\usepackage{commandes_EP}
\usepackage{graphicx}
\usepackage{dcolumn}
\usepackage{bm}
\usepackage{commandes_EP}
\usepackage{eufrak}
\usepackage{euscript}
\usepackage{cancel}
\usepackage{color}
\usepackage{mathrsfs}
\usepackage{adjustbox}
\usepackage{amsmath}
\usepackage{amsthm}
\usepackage{amssymb}

\usepackage{tikz}
\usepackage{float} 

\usepackage[normalem]{ulem}
\usepackage{graphicx}
\usepackage{wrapfig}
 
\usepackage[font={footnotesize,bf}]{caption}

\usepackage{empheq} 

\newtheorem{proposition}{Proposition}
\newtheorem{property}{Property}
\newtheorem{numScheme}{Numerical Scheme}

\begin{document}

\begin{frontmatter}
\title{An asymptotic preserving well-balanced scheme for the isothermal fluid equations in low-temperature plasma applications}


\author[cmap,lpp]{A.~Alvarez Laguna} 
\ead{alejandro.alvarez-laguna@polytechnique.edu}
\author[cmap]{T.~Pichard}
\ead{teddy.pichard@polytechnique.edu}
\author[vki]{T.~Magin}
\ead{magin@vki.ac.be}
\author[lpp]{P.~Chabert}
\ead{pascal.chabert@lpp.polytechnique.fr}
\author[lpp]{A.~Bourdon}
\ead{anne.bourdon@lpp.polytechnique.fr}
\author[cmap]{M.~Massot}
\ead{marc.massot@polytechnique.edu}

\address[cmap]{Centre de Math\'ematiques Appliqu\'ees,Ecole Polytechnique, Route de Saclay, 91128 Palaiseau Cedex, France}
\address[lpp]{Laboratoire de Physique des Plasmas, CNRS, Sorbonne Universit\'e, Univ.~Paris Sud, Ecole Polytechnique, F-91128, Palaiseau, France}
\address[vki]{von Karman Institute for Fluid Dynamics, Waterloosesteenweg 72, 1640 Sint Genesius Rode, Belgium}

\begin{abstract}

 We present a novel numerical scheme for the efficient and accurate solution of the isothermal two-fluid (electron + ion) equations coupled to Poisson's equation for low-temperature plasmas. The model considers electrons and ions as separate fluids, comprising the electron inertia and charge separation. The discretization of this system with standard explicit schemes is constrained by very restrictive time steps and cell sizes related to the resolution of the Debye length, electron plasma frequency, and electron sound waves. Both sheath and  electron inertia are fundamental to fully explain the physics in low-pressure and low-temperature plasmas. However, most of the phenomena of interest for fluid models occur at speeds much slower than the electron thermal speed and are quasi-neutral, except in small charged regions. A numerical method that is able to simulate efficiently and accurately all these regimes is a challenge due to the multiscale character of the problem. In this work, we present a scheme based on the Lagrange-projection operator splitting that preserves the asymptotic regime where the plasma is quasi-neutral with massless electrons. As a result, the quasi-neutral regime is treated without the need of an implicit solver nor the resolution of the Debye length and electron plasma frequency. Additionally, the scheme proves to accurately represent the dynamics of the electrons both at low speeds and when the electron speed is comparable to the thermal speed. In addition, a well-balanced treatment of the ion source terms is proposed in order to tackle problems where the ion temperature is very low compared to the electron temperature. The scheme significantly improves the accuracy both in the quasi-neutral limit and in the presence of plasma sheaths when the Debye length is resolved. In order to assess the performance of the scheme in low-temperature plasmas conditions, we propose two specifically designed test-cases: a quasi-neutral two-stream periodic perturbation with analytical solution and a low-temperature discharge that includes sheaths. The numerical strategy, its accuracy, and computational efficiency are assessed on these two discriminating configurations. 
\end{abstract}
\end{frontmatter}

\section{Introduction}

Low-temperature plasmas are fundamental to a wide range of technological applications -- from lighting to semiconductor manufacturing and electric propulsion. These plasmas are created by an electrical discharge in a gas. The physical conditions where they occur are particularly diverse: pressures varying from below a millitorr to few hundred atmospheres, different excitation methods (e.g., inductively coupled plasma (ICP), capacitively coupled plasma (CCP), dielectric barrier discharge (DBD), magnetron), diverse geometries, and power ranges (see, e.g., \cite{Chabert11,Lieberman94}). For that reason, a number of different descriptions are used for their study, from kinetic approaches, e.g., Particle-In-Cell/Monte Carlo Collision (PIC-MCC) and Direct Simulation Monte Carlo (DSMC) method, to fluid and global models (see \cite{Bogaerts17,Alves17} for a review on the modeling and numerical approaches of low-temperature plasmas). 

In atmospheric pressure discharges, the drift-diffusion approximation (e.g., \cite{Colella99,Becker13,Markosyan15,Hagelaar00, Duarte12, Bourdon16,Viegas18}) is a valid simplification. However, under low-pressure conditions, the momentum transfer between species is smaller since the number of collisions decreases for decreasing pressures. This results in non-equilibrium conditions that may lead to a differential motion between the heavy species and electrons. This differential motion is responsible for plasma instabilities such as the two-stream instability \cite{Chen84} or the electron drift instability that is observed to cause anomalous transport in Hall thrusters \cite{Adam04,Adam08,Ducrocq06, Heron13,Boeuf14,Lafleur16a,Boeuf18}. Even though the electron inertial term is small in the electron momentum equations, these instabilities cannot be explained without it. Similarly, 
electron inertia cannot be completely neglected to fully explain the physics of sheaths \cite{Duarte11}, in electron sheaths  and presheaths \cite{Scheiner15, Borgohain17}, in plasma sheath instabilities \cite{Stenzel11}, Langmuir probes \cite{Guittienne18}, in magnetized plasmas \cite{Andres14, Jana17,Srinivasan18}, and Hall thrusters \cite{Ahedo02,Ahedo05}.

For that reason, in low-pressure discharges, alternative models to the drift-diffusion approach such as kinetic and hybrid models, the plasma moment equations or the multicomponent non-equilibrium one velocity fluid model \cite{Magin09,Wargnier18c} are needed. Kinetic simulations provide a very accurate description of the state of the plasma, but are computationally very expensive. On the other hand, hybrid models \cite{Kushner09, Hara12, Kolobov18} are a cheaper alternative that combines the fluid and the kinetic description. However, both methods need to resolve the Debye length and the electron plasma frequency if the time integration is explicit and the quasi-neutrality hypothesis is not assumed.

In Table \ref{tableScales}, we present characteristic conditions of a typical low-temperature low-pressure Argon plasma RF discharge (from \cite{Chabert11}). We first note that under these pressure conditions, the Knudsen number is of order one for electrons and ten times smaller for ions. The kinetic phenomena is important and transport models are needed to provide a closure for the plasma fluid equations. The second important feature of Table \ref{tableScales} is related to the smallness of the electron-to-ion mass ratio and the normalized Debye length. Most of the explicit methods (PIC methods, hybrid methods or deterministic discretization of the multi-fluid equations) require to resolve these scales in order to guarantee the stability of the scheme. However, if we consider an explicit 1D simulation of the discharge of Table \ref{tableScales} that resolves the Debye length and the plasma frequency with ten spatial and temporal points, the numerical set-up would need around $10^4$ spatial points (in 1D) and $10^6$ time steps to simulate one transit time of an ion acoustic wave. Therefore, this set-up involves a significant computational cost for a 1D simulation. Implicit methods can guarantee the numerical stability with larger time-steps. However, the Newton method to solve the linear system can become also very expensive.


In the present work, we consider the isothermal plasma fluid equations under collisionless conditions. This system of equations contains the numerical difficulties associated to the low-temperature moment plasma equations in electrostatic conditions. The model considers the first two moments of Vlasov's equation for electrons and ions while assuming both temperatures to be constant. The dynamics of the two fluids are coupled through the electric potential, which satisfies a Poisson equation. These equations are widely used in the sheath theory \cite{Tonks29,Bohm49,Riemann05} and in the study of plasma waves \cite{Chen84}. Although kinetic phenomena play an important role in the plasma sheath \cite{Scheiner15}, a fluid model that does not assume quasi-neutrality can potentially capture the interaction between the macroscopic scales and the sheaths.



The most general form of the multi-fluid plasma equations considers the electron inertial term in the electron momentum equation and charge separation effects. This is equivalent to considering a finite electron mass and a finite Debye length. As mentioned before, this allows for representing plasma instabilities such as the two-stream or the drift waves and charged regions of the plasma such as the sheaths. Nevertheless, both the electron mass and the Debye length are two very small parameters as compared to the ion and macroscopic scales. Consequently, these small scales impose very restrictive numerical constraints related to the resolution of the Debye length, electron plasma waves, and electron acoustic waves. Consequently, due to these requirements, the time-steps and mesh sizes are not significantly more advantageous than in the kinetic or hybrid approach. Moreover, the multiscale character of the problem can lead to very large discretization errors when the small scales are not properly resolved (see~\cite{AlvarezLaguna18}).

\begin{table}[!htb] 
\caption{Characteristic values of an Argon RF discharge at $1 $ Pa \cite{Chabert11}.}\label{tableScales} 
\begin{center}
\resizebox{\columnwidth}{!}{
\begin{tabular}{l | c | c c || l | c | c}
  \multicolumn{4}{c ||}{Dimensional quantities} & \multicolumn{3}{c}{Dimensionless quantities} \\
  \hline\\[-0.34cm]
  \hline\\[-0.37cm]
  Neutral density &$\nneut$ & $1.25\times 10^{20}$ & $\text{m}^{-3}$ & Electron-to-ion mass ratio & $\varepsilon = m_\elec/m_\ion$ & $1.36\times10^{-5}$  \\
  Electron density &$n_{\elec,\ion}$ & $10^{16}$ & $\text{m}^{-3}$ & Ion-to-electron temperature ration & $\kappa = T_\ion/T_\elec$ & $0.025$\\
   Neutral and ion temperature &$T_{\neut,\ion}$ & $0.05$ & $\text{eV}$ & Normalized Debye length & $\lambda = \debye/l$ & $3.5\times10^{-3}$  \\
  Electron temperature & $T_{\elec}$ & $2$ & $\text{eV}$ & Ionization level & $n_{\elec, \ion}/\nneut$& $8\times10^{-5}$\\
  Distance between plates & $l$ & $3$ & $\text{cm}$ & Electron Knudsen number& $\text{Kn}_\elec$ & $1.7$ \\
  Ion-neutral collisional cross section & $\sigma_{\ion\neut}$ & $10^{-18}$ & $\text{m}^2$ & Ion Knudsen number& $\text{Kn}_\ion$ & $0.17$\\
  Electron-neutral collisional cross section & $\sigma_{\elec\neut}$ & $10^{-19}$& $\text{m}^2$ & Normalized ionization rate & $\Damk$  & $0.0139$\\
  Ionization constant & $K_{ion}$ & $8.16\times10^{-18}$& $\text{m}^3\text{s}^{-1}$ & Normalized electron collision rate &$\nueb $ & 153.8 \\
  Ionization potential & $\varepsilon_{ion}$ & $17.44$ & $\text{eV}$ &Normalized ion collision rate & $\nuib $& $0.94$ \\
  Electron plasma period & $\omega^{-1}_{p\elec}$ & $1.77\times 10^{-10}$ & s & Normalized  plasma period &  $\bar\omega^{-1}_{p\elec} =  \omega^{-1}_{p\elec}u_B/l$& $1.29\times 10^{-5}$
  \end{tabular}
  }
\end{center}
\end{table}

Numerical methods for the ideal multi-fluid equations coupled to Maxwell's equations have been proposed by a number of authors \cite{Sousa16,Shumlak11,Shumlak03,Loverich11,Hakim06,AlvarezLaguna17c,AlvarezLaguna17b, Alonso19}, for the study of plasma sheaths \cite{Cagas17}, and for the study of plasma expansion in vacuum with the isentropic electrostatic approximation \cite{Crispel04,Crispel05,Crispel07,Degond08,Degond11}. The main difficulties of the low-temperature multi-fluid plasma model and the solutions proposed in the current literature are summarized in the following. First, the stability of explicit schemes for the multi-fluid equations solving for the Poisson equation requires the resolution of the electron plasma wave frequency \cite{Degond11}. Alvarez Laguna et al.~\cite{AlvarezLaguna17c,AlvarezLaguna17b} have proposed implicit time integration for the previous set of equations coupled to full Maxwell's equations. However, the inversion of the matrix is computationally costly, which can be improved by the use of GPUs \cite{Alonso19}. Alternatively, a more advantageous approach is the asymptotic-preserving (AP) scheme proposed such as the one proposed by Degond et al.~\cite{Crispel07,Degond08,Degond11}. AP schemes \cite{Jin99} preserve the quasi-neutral asymptotic limit without the resolution of the Debye length nor the plasma frequency. These schemes do not need an implicit solver, which results in a significant advantage from the computational point of view. However, the AP scheme from \cite{Crispel07,Degond08} is not able to tackle the asymptotic limit of the small electron-to-ion mass ratio, finding large errors in the velocity of the electrons even with high order schemes. Additionally, the AP scheme proposed in \cite{Crispel07} demands the solution of a second order equation for the electric potential that requires the storage of the solution in two different time-steps, which increases the computational stencil.

The second difficulty of the multi-fluid equations is due to the small mass of the electrons, which results in a very large electron speed of sound. However, in most of the cases, the electron fluid travels at much smaller speeds than the electron thermal speed. This corresponds to a low-Mach regime for electrons, which is known to give numerical problems due to an excessive numerical dissipation that restricts the time step of compressible solvers \cite{VKILS99}. Different strategies are proposed in order to tackle the low-Mach regime in compressible solvers. One method are the so-called all-regime flux-splitting methods such as the AUSM$^+$-up \cite{Liou06} (applied to the multi-fluid equations in \cite{AlvarezLaguna17c,AlvarezLaguna17b, Alonso19}) or the preconditioning methods to remove the stiffness of the low-Mach regime \cite{Turkel87}. Similarly, the operator-splitting Lagrange-projection scheme is combined to the preconditioning method by Chalons et al.~\cite{Chalons16}. Additionally, the AP schemes are also used to tackle the incompressible asymptotic behaviour in the Euler equations \cite{Dimarco18}.

The third difficulty is related to the regime where the temperature of the ions is much lower than the electron temperature. In low-temperature plasmas, the electric potential, in general, scales as the electron temperature (in eV). This results in a Lorentz force that is much larger than the ion pressure flux. The main problem arises when the ion convective flux is treated with an upwind scheme and the Lorentz force (which involves the gradient of the potential) with a cell-centered scheme. The upwind scheme provides accurate non-oscillatory solutions in homogeneous problems, but they can lose accuracy in the presence of stiff source terms \cite{Roe87}. A solution to this problem is to use well-balanced schemes that upwind the source terms in a consistent manner \cite{Bermudez94, Leveque90a, Gosse13}. 

In the present work, we propose a numerical scheme that addresses these three problems and design three dedicated discriminating test-cases in order to benchmark such a method. Regarding the stiffness introduced by the smallness of the Debye length and the electron mass, we propose a novel operator splitting method based on the all-regime Lagrange-projection \cite{Chalons16} coupled to the Poisson equation. We retrieve a numerical scheme that has the AP property of the quasi-neutral limit with massless electrons. Therefore, neither the Debye length nor the electron plasma frequency need to be resolved to be stable and accurate, without the need of an implicit solver. As compared to previous AP schemes for the multi-fluid equations \cite{Crispel07, Degond08}, our approach does not need to solve an equivalent second order equation for the electric potential, but it solves the standard Poisson equation. This results in a much simpler algorithm for the potential. Additionally, the main advantage is that the numerical scheme can simultaneously tackle the problem of the small mass of electrons. The solution of these two problems with a unique AP scheme is an original contribution of this work. This can be considered a step forward that could help to reduce the numerical cost of solving the electron fluid dynamics coupled to the one of ions. Finally, a well-balanced discretization of the Lorentz force in the ion momentum equation is proposed. In the results we show that an incorrect discretization of this term can lead to spurious numerical instabilities, which can be avoided with well-balanced schemes.

The numerical scheme is then benchmarked against three numerical set-ups that allows us to assess our strategy in terms of accuracy and computational cost. The first one simulates a quasi-neutral periodic perturbation in a thermal plasma. This case is used to test the asymptotic preserving property of the discretization for a small Debye length and a small electron-to-ion mass ratio. The same case is reproduced in a low-temperature plasma in order to assess the well-balanced discretization of the ions. We demonstrate that the proposed numerical strategy allows for a dramatic reduction of the computational time in cases with quasi-neutral plasma, as compared to a standard discretization. Finally, a low-temperature plasma discharge between two floating walls is simulated. This realistic set-up is able to capture the physics of the electrically charged plasma sheath coexisting with a quasi-neutral bulk and encompasses most of the difficulties of more realistic configurations, while amenable to detailed analysis of the proposed methods.


\paragraph{Outline of the paper} First, we present the set of equations, its normalization, and the asymptotic study when both the mass of electrons and the Debye length are very small as compared to the ion scales. Second, we present a standard discretization of the system in order to illustrate the difficulties of the numerical discretization. Third, we describe the operator splitting scheme and the well-balanced treatment of the ion source-terms while proving that the acoustic step preserves the asymptotic behaviour. Fourth, we simulate a two-stream perturbation in thermal and low-temperature plasmas. Finally, we self-consistently simulate a low-temperature discharge between two floating plates. The numerical scheme proves to be accurate both in the quasi-neutral limit and when the sheaths are included in the computational domain.

\section{Set of equations and asymptotic behaviour}

We consider the isothermal plasma (electron +  ion) equations. The equations are obtained by taking the first two moments of the kinetic equation for electrons and ions (see, e.g., \cite{Braginskii65} for a derivation). These moments correspond to the mass and momentum balance laws for the two charged species. The system considers ionization reactions and, in the present paper, we neglect the effect of recombination and the elastic collisions. As explained in \cite{Chabert11}, in most discharges these scales are much smaller than the ionization scale. Nevertheless, the elastic collisions in low-pressure conditions do not impose a major difficulty from the numerical point of view and for that reason are not the focus of the present work. Finally, the system is closed with the Poisson equation for the electric potential. 


The equations in dimensional form read
\begin{subequations}
\begin{eqnarray}
    \dt \nelec + \dx  (\nelec \ue) &=&\nelec \nuionization,\label{eq1}\\
    \dt \nion + \dx (\nion \uion) &=&\nelec \nuionization,\label{eq2}\\
     \me \dt (\nelec \ue) +\dx  \left(\me\nelec\ue^2 + \pelec \right)&=&\nelec \ee \dx \phi ,\label{eq3}\\
      \mion \dt (\nion \uion) +\dx \left(\mion\nion\uion^2 + \pion \right)&=&-\nion \ee \dx \phi ,\label{eq4}\\
    \partial^2_{xx} \phi &=& \frac{\nelec-\nion}{\permvac} \ee, \label{eq5}
\end{eqnarray}
\end{subequations}
where $\nelec $ and  $ \nion $ stand for the electron and ion number density respectively, and $\ue$ and $\uion$ the electron  and ion velocities.  The electron-impact ionization rate coefficient $\nuionization$ is a function of the electron temperature $\nuionization = \nneut\Ki(\Te)$, for instance through Arrhenius' law, $\Ki(\Te)=A\exp[{{- \Eion}/{(\boltz\Te)}}]$, where the quantity $\Eion$ is the ionization energy. The neutral number density $\nneut$ is assumed to be constant. The partial pressures of the electron and ion fluids are assumed to obey the perfect gas law, $\pelec=\nelec \boltz \Te $ and $\pion=\nion \boltz \Tion $, where $\boltz$ is Boltzmann's constant, and $\Te$ and $\Tion$, the  electron and ion constant temperatures, respectively. The plasma is considered to be in thermal non-equilibrium, i.e., $\Te\neq \Tion$, at constant temperatures.

Despite the isothermal assumption, the considered set of equations contains the main difficulties of the low-temperature moment plasma models. These difficulties, as described in the introduction, are related to the mass disparity between electrons and ions, the small Debye length, and the low-temperature of the ions. 

\subsection{Normalized equations}

The set of equations \eqref{eq1}-\eqref{eq4} and \eqref{eq5} are normalized by introducing some reference quantities: $\no$, the characteristic number density common to electrons and ions,  $\Lo=l$, the reference length, $\Te$, the electron temperature, and $\Tion$, the ion temperature. The rest of the reference variables are calculated as a combination of the previous ones: the reference velocity common to electrons and ions is based on the Bohm velocity $\uo \equiv u_B = \sqrt{\boltz \Te / \mion}$, the charateristic time  $\timeo= \Lo/ \uo$ is obtained from the reference velocity and reference distance, whereas the reference potential  $\phio = \boltz \Te/\ee$ is based on the thermal energy of electrons. The normalized set of equations reads
\begin{subequations}\label{eq:sysND}
\begin{eqnarray}
    \dtb \rhoeb + \dxb (\rhoeb \ueb) &=&  \rhoeb \Damk ,\label{elMass_ND}\\
    \dtb  \rhoionb + \dxb ( \rhoionb \uionb) &=& \rhoeb \Damk,\label{IonMass_ND}\\
    \dtb (\rhoeb\ueb) + \dxb  \left[\rhoeb \left(\ueb^2+  \varepsilon^{-1}\right)\right]&=&\varepsilon^{-1}\rhoeb \dxb\phib ,\label{elMom_ND}\\
    \dtb (\rhoionb \uib) + \dxb \left[\rhoionb \left(\uionb^2 + \kappaT\right) \right]&=&-\rhoionb \dxb\phib ,\label{IonMom_ND}\\
    \dxxb \phib &=& \lambdaSq^{-1}\left(\rhoeb - \rhoionb\right), \label{Pois_ND}
\end{eqnarray}
\end{subequations}
where the non-dimensional parameters are defined as:
\begin{subequations}
\begin{eqnarray}
  \text{Electron-ion mass ratio: }\varepsilon = \me/\mion,&~~&\text{ion-to-electron temperature: }\kappaT= \Tion/\Te,\nonumber \\ \text{squared non-dimensional Debye length: }\lambdaSq = \frac{\permvac\boltz\Te}{\nelec\ee^2L^2_0},&~~&\text{dimensionless frequency: }\Damk= \nuionization\timeo.\nonumber
\end{eqnarray}
\end{subequations}

Note that, for the sake of simplicity of notation, we use the square of the Debye length $\lambdaSq$ as a non-dimensional parameter. Additionally, we define the non-dimensional Debye length $\lambda = \sqrt{\lambdaSq}$ as it will be used to describe characteristic lengths of the problem.

We highlight the importance of the three nondimensional parameters $\varepsilon$, $\lambdaSq$, and $\kappaT$ in order to build a numerical scheme. In Table \ref{tableScales}, we present the typical values in an Argon RF discharge. As discussed in the introduction, the smallness of the $\varepsilon$ and $\lambdaSq$ impose very restrictive numerical constraints. In the following section, we study the asymptotic behaviour when these two parameters tend to zero. 

\subsection{Asymptotic behavior}\label{sec:AsymBehavior}

We study the multiscale asymptotic behavior \cite{Klainerman81} with respect to $\varepsilon$ and $\lambdaSq$. Previous work \cite{Degond11,Crispel07} performed this study only with respect to the Debye length, $\lambda$. However, the inclusion of $\varepsilon$ in the analysis is fundamental as the electron velocity is generally much smaller than the thermal speed of electrons. Since  $\varepsilon$ and $\lambdaSq$ are the smallest parameters of the system, we do not include $\kappa$ or $\Damk$ in our study. 

From a physics point of view, we can consider three different asymptotic behaviours that correspond to different plasma phenomena, as illustrated in Fig.~\ref{Fig:asymLimits}. The complete problem $^\lambdaSq\espace F^\varepsilon$ corresponds to the system of eqs.~\eqref{eq:sysND}. This system considers finite Debye length and electron inertia. This problem resolves all the possible scales, being the fastest one corresponding to the electron plasma waves. The main problem of designing a numerical scheme for these small scales is that it might be very inefficient to represent the macroscopic scales and leading to consistency problems due to an imbalanced numerical dissipation when the two parameters are small. 

The regime where the Debye length tends to zero for arbitrarily small $\varepsilon$ corresponds to a quasineutral plasma with the electrons that can move at bulk speed closer to the electron thermal velocity. This problem is of interest specially in the presence of a magnetic field that can produce drift motions at very high speed, such as in Hall effect thrusters. This regime allows for the representation of plasma instabilities such as the two-stream instability or the electron-drift instability. We denote this problem as $^0\espace F^\varepsilon$.

Alternatively, we can consider the asymptotic behaviour of $\varepsilon$ tending to zero for a finite Debye length. In this regime, the electrons move at speeds comparable to the ion sound velocity (Bohm's velocity) and the Debye length is arbitrarily small. This regime is important, for instance, in the plasma-sheath transition. We denote this problem as $^\lambdaSq\espace F^0$.

Finally, we consider the asymptotic limit where the electrons travel at speeds comparable to the ion velocity in a quasi-neutral plasma, i.e., $\varepsilon\rightarrow0$ and $\lambdaSq\rightarrow0$ This behaviour is present in most of the phenomena occuring at ion scales in cold and thermal plasmas. For that reason, we focus in this problem in the rest of the paper. We denote the problem as $^0\espace F^0$. In the following, we show the set of equations corresponding to the problem $^0\espace F^0$ and we demonstrate that $\lim\limits_{\lambdaSq\rightarrow0}~^\lambdaSq\espace F^0 =\lim\limits_{\varepsilon\rightarrow0}~^0\espace F^\varepsilon \equiv~^0\espace F^0$ in the case of periodic boundary conditions.

\begin{figure} [!htb] 
	\centering
		  \includegraphics[trim=0cm 0cm 0cm 0cm, clip=true,width=0.7\textwidth]{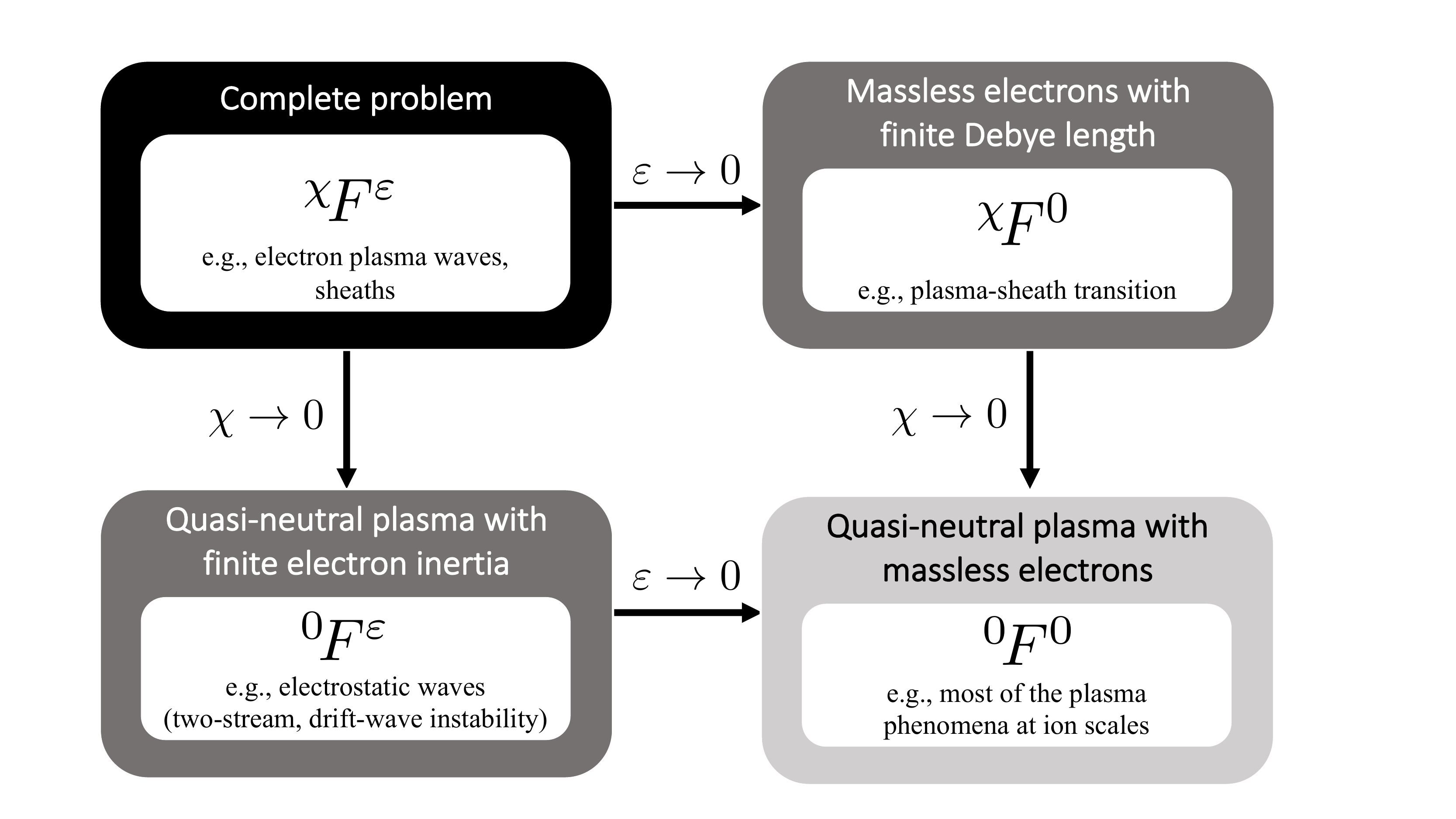} 
	\caption{Asymptotic behaviour of the electrostatic isothermal multi-fluid plasma equations.}
	\label{Fig:asymLimits}
\end{figure}


Let us consider the system of equations \eqref{elMass_ND}-\eqref{Pois_ND} with periodic boundary conditions, which we denote as the problem $^\lambdaSq\espace F^\varepsilon$. We consider two different asymptotic expansions for the problem $^\lambdaSq\espace F^\varepsilon$: (1) In terms of the small parameter $\varepsilon$ and (2) in terms of the small parameter $\lambdaSq$. We define the problem $^\lambdaSq\espace F^0$ as the system of equations for the zero-th order terms of the expansion in terms of $\varepsilon$. Similarly, we define $^0\espace F^\varepsilon$ as the  the system of equations for the zero-th order terms of the expansion in terms of $\lambdaSq$. 

\begin{proposition}\label{prop}
The system of equations that corresponds to $\lim\limits_{\lambdaSq\rightarrow0}~^\lambdaSq\espace F^0$ with periodic boundary conditions is formally the same as $\lim\limits_{\varepsilon\rightarrow0}~^0\espace F^\varepsilon$ and it reads
\begin{subequations}\label{eq:sysNN}
\begin{empheq}[left=~^0\espace F^0:~\empheqlbrace]{align}
    \rhoeb &= \rhoionb ,\label{quasineutrality}\\
    \dtb \rhoionb + \dxb (\rhoionb \uionb) &=  \rhoeb \Damk,\\
    \dxb  (\rhoeb \ueb) &= \dxb  ( \rhoionb \uionb), \\ \label{currentConserv}
    \dtb(\rhoionb \uionb) + \dxb (\rhoionb \uionb^2 + \kappaT\rhoionb) &= \rhoionb\dxb \phib,\\
    \dxb \phib &= \frac{1}{\rhoeb}\dxb\rhoeb.\label{BoltzmannElec}
\end{empheq}
\end{subequations}
\end{proposition}

\begin{proof} In order to prove Proposition \ref{prop}, we first propose the following expansion for the variables of the problem $^\lambdaSq\espace F^\varepsilon$
\begin{equation}\label{eq:expasionEps}
  f(x,t) = \AsymExpR{f}{0}{\lambdaSq} +  \varepsilon  \AsymExpR{f}{1}{\lambdaSq} + \mathcal{O}(\varepsilon^2).
\end{equation}
By injecting this expansion in the system of equations \eqref{elMass_ND}-\eqref{Pois_ND}, we find that the system of equations for the zero-th order terms reads
\begin{subequations}\label{eq:G0mu}
\begin{empheq}[left=~^\lambdaSq\espace F^0:~\empheqlbrace]{align}
    \dtb\rhoebNmu + \dxb(\rhoebNmu\uebNmu)&= \rhoebNmu \Damk ,\label{eq:gElMass}\\
    \dtb\rhoionbNmu + \dxb(\rhoionbNmu\uionbNmu)&= \rhoebNmu \Damk ,\label{eq:gIonMass}\\
    \frac{1}{\rhoebNmu}\dxb\rhoebNmu &= \dxb\phibNmu ,\label{eq:gElMom}\\
    \dtb\left(\rhoionbNmu\uionbNmu\right) + \dxb\left[\rhoionbNmu\left(\uionbNmuSquare + \kappa \right)\right]&= -\rhoionbNmu \dx \phibNmu ,\label{eq:gIonMom}\\
    \dxxb\phibNmu &= \lambdaSq^{-1}\left(\rhoebNmu - \rhoionbNmu \right).\label{eq:gPois}
\end{empheq}
\end{subequations}
We denote this system as $~^\lambdaSq\espace F^0$.

Alternatively, we propose an expansion for the variables of the problem $^\lambdaSq\espace F^\varepsilon$ in terms of $\lambdaSq$ of the form 
\begin{equation}\label{eq:expasionMu}
  f(x,t) =\AsymExpL{f}{\varepsilon}{0} +\lambdaSq~\AsymExpL{f}{\varepsilon}{1} + \mathcal{O}(\lambdaSq^2).
\end{equation}
The system of equations corresponding to the zero-th order terms reads
\begin{subequations}\label{eq:H0eps}
\begin{empheq}[left=~^0\espace F^\varepsilon:~\empheqlbrace]{align}
    \dtb\rhoebNeps + \dxb\left(\rhoebNeps\uebNeps\right)&= \rhoebNeps \Damk ,\label{eq:gElMass}\\
    \dtb\rhoionbNeps + \dxb\left(\rhoionbNeps\uionbNeps\right)&= \rhoebNeps \Damk ,\label{eq:gElMass}\\
        \dtb\left(\rhoebNeps\uebNeps\right) + \dxb\left[\rhoebNeps\left(\uebNeps^2 + \varepsilon^{-1} \right)\right]&= \varepsilon^{-1}~\rhoebNeps \dx \phibNeps ,\label{eq:gElMom}\\
    \dtb\left(\rhoionbNeps\uionbNeps\right) +  \dxb\left[\rhoionbNeps\left(\uionbNeps^2 + \kappa \right)\right]&= -\rhoionbNeps \dx \phibNeps ,\label{eq:gIonMass}\\
   \rhoebNeps &= \rhoionbNeps.
\end{empheq}
\end{subequations}
We denote this system as $^0\espace F^\varepsilon$.

We can use the expansion of eq.~\eqref{eq:expasionMu} for the problem $^\lambdaSq\espace F^0$ in order to study the problem $\AsymExpR{F}{0}{0}\equiv\lim\limits_{\lambdaSq\rightarrow0}~^\lambdaSq\espace F^0$. The system of equations for the zero-th order terms is very similar to system $^\lambdaSq\espace F^0$, with the difference in the Poisson equation, as shown below
\begin{subequations}\label{eq:G0mu}
\begin{empheq}[left=\AsymExpR{F}{0}{0}\equiv\lim\limits_{\lambdaSq\rightarrow0}~^\lambdaSq\espace F^0:~\empheqlbrace]{align}
    \dtb\rhoebNNmu + \dxb(\rhoebNNmu\uebNNmu)&= \rhoebNNmu \Damk ,\label{eq:gNElMass}\\
    \dtb\rhoionbNNmu + \dxb(\rhoionbNNmu\uionbNNmu)&= \rhoebNNmu \Damk ,\label{eq:gNIonMass}\\
    \frac{1}{\rhoebNNmu}\dxb\rhoebNNmu &= \dxb\phibNNmu ,\label{eq:gNElMom}\\
    \dtb\left(\rhoionbNNmu\uionbNNmu\right) + \dxb\left[\rhoionbNNmu\left(\uionbNNmuSquare + \kappa \right)\right]&= -\rhoionbNNmu \dx \phibNNmu ,\label{eq:gNIonMom}\\
    \rhoebNNmu &=\rhoionbNNmu .\label{eq:gNPois}
\end{empheq}
\end{subequations}
By using the condition \eqref{eq:gNPois} in eqs \eqref{eq:gNElMass} and \eqref{eq:gNIonMass}, we obtain that the fluxes are conserved, i.e., $\dxb(\rhoebNNmu\uebNNmu) = \dxb(\rhoionbNNmu\uionbNNmu)$. This proves that the system $\AsymExpR{F}{0}{0}$ is the system proposed as $^0\espace F^0$ in eqs.~\eqref{eq:sysNN}.

Alternatively, we use the expansion of eq.~\eqref{eq:expasionEps} for the problem $^0\espace F^\varepsilon$. The equations for the zero-th order terms are denoted as $\AsymExpL{F}{0}{0}\equiv\lim\limits_{\varepsilon\rightarrow0}~^0\espace F^\varepsilon$. In this new system of equations, the electron momentum equation is modified as compared to the one in $^0\espace F^\varepsilon$.
\begin{subequations}\label{eq:F00}
\begin{empheq}[left=\AsymExpL{F}{0}{0}\equiv\lim\limits_{\varepsilon\rightarrow0}~^0\espace F^\varepsilon:~\empheqlbrace]{align}
    \dtb\rhoebNNeps + \dxb\left(\rhoebNNeps\uebNNeps\right)&= \rhoeb \Damk ,\label{eq:gElMass}\\
    \dtb\rhoionbNNeps + \dxb\left(\rhoionbNNeps\uionbNNeps\right)&= \rhoeb \Damk ,\label{eq:gElMass}\\
        \frac{1}{\rhoebNNeps}\dxb\rhoebNNeps &=  \dx \phibNNeps ,\label{eq:gIonMass}\\
    \dtb\left(\rhoionbNNeps\uionbNNeps\right) +  \dxb\left[\rhoionbNNeps\left(\uionbNNeps^2 + \kappa \right)\right]&= -\rhoionbNNeps \dx \phibNNeps ,\label{eq:gIonMass}\\
   \rhoebNNeps &= \rhoionbNNeps.
\end{empheq}
\end{subequations}
This system of equations is the same as the problem $\AsymExpR{F}{0}{0}$. Consequently, in the case of periodic boundary conditions, both limits result in the same system of equations \eqref{quasineutrality}-\eqref{BoltzmannElec}, i.e., $\AsymExpL{F}{0}{0}=\AsymExpR{F}{0}{0}=~^0\espace F^0$.

\end{proof}

We highlight that the system $~^0\espace F^\varepsilon$, which corresponds to the limit of the Debye length tending to zero, is the one that was previously considered by the works  \cite{Degond11,Crispel07}. In this system, the electric potential becomes a Lagrange multiplier that imposes the charge neutrality. In system \eqref{eq:H0eps}, we decided to write the set of equations in this form for clarity. Nevertheless, a differential equation for the electric potential can be retrieved from the conservation of electric charge, as done in  \cite{Crispel07}.

The asymptotic behaviour proposed in this paper considers massless electrons as in \cite{Varet13,Slemrod01}. It is therefore very different to the one derived in \cite{Degond11,Crispel07}. By including in our analysis the electron-to-ion mass ratio, we find that the quasi-neutrality is found in eq.~\eqref{quasineutrality} and the electron density follows the Boltzmann distribution in eq.~\eqref{BoltzmannElec}. Furthermore, the electric current is conserved by equation \eqref{currentConserv} and the ion mass and momentum equations are unchanged. The main difference as compared to \cite{Crispel07} is the limit for the zero-th order of the potential in eq.~\eqref{BoltzmannElec} and the electron momentum equation.

\subsection{Study of a bounded low-temperature plasma through the fluid plasma equations}\label{sec:problemdefinition}

In low-temperature plasma industrial applications, plasmas are confined. As a result, the charged particles that are produced inside the reactor through ionization, are lost through the boundaries when they strike the wall. Since the thermal motion of electrons is larger than that of ions, the surface will charge negatively with respect to the plasma (in electropositive plasmas), forming a charged boundary layer called the plasma sheath. 


The analytical models for the sheath and presheath rely on the isothermal multi-fluid equations \citep{Riemann91}, while assuming the inertia of electrons and temperature of ions to be negligible. Nevertheless, important kinetic phenomena taking place are not included in this model \cite{Badsi15}.

Let us consider a 1D domain of length $l$ filled with a plasma between two floating walls, with no secondary electron emission, the distribution function of electrons is a Maxwellian and all the electrons that touch the wall are absorbed by the wall. With these assumptions, the flux of electrons collected by the wall (see, e.g., \cite{Chabert11}) both in dimensional and dimensionless units read:
\begin{equation}
  \text{Dimensional: } \nelec\ue|_{wall} = \nelec\sqrt{\frac{\boltz\Te}{2\pi\me}}~~~\text{and}~~\text{dimensionless: } \rhoeb\ueb|_{wall} = \frac{\rhoeb}{\sqrt{2\pi\varepsilon}}.\label{fluxElectrons}
\end{equation}
A steady solution is found when the ionization inside the bulk of the plasma balances the particle loss as follows
\begin{equation}
 2\nelec\ue|_{wall} = \int_0^l \nelec \nuionization dx.\label{fluxElectrons}
\end{equation}

As mentioned by Riemann \cite{Riemann05}, the ionization frequency is an eigenvalue of the problem. Consequently, there is only one ionization frequency that finds a steady state solution for a given distance between plates. In this paper, we propose a numerical methodology that proves to be convergent to find this eigenvalue.

With the previous assumptions, the potential at the pre-sheath $\phi_p$ and the wall $\phi_W$, in dimensional units \cite{Liebermann95}, as follows
\begin{equation} \label{eq:SheathTheory}
  \phi_p = -\frac{\boltz \Te}{2\ee} ~~~\text{and}~~~ \phi_W = \frac{\boltz \Te}{\ee}\ln\left(\frac{\me}{2\pi\mi}\right)^{1/2},
\end{equation}
where $\phi_p$ is the potential drop needed to accelerate the ions to  Bohm's speed (neglecting the ion pressure gradient) and $\phi_W$ is the potential drop in the sheath. In Fig.~\ref{0_SheathExample}, we illustrate the steady state solution of a bounded plasma between two floating plates. 

\begin{figure} [H] 
	\centering
		  \includegraphics[trim=1.5cm 1.2cm 1.5cm 0cm, clip=true,width=0.325\textwidth]{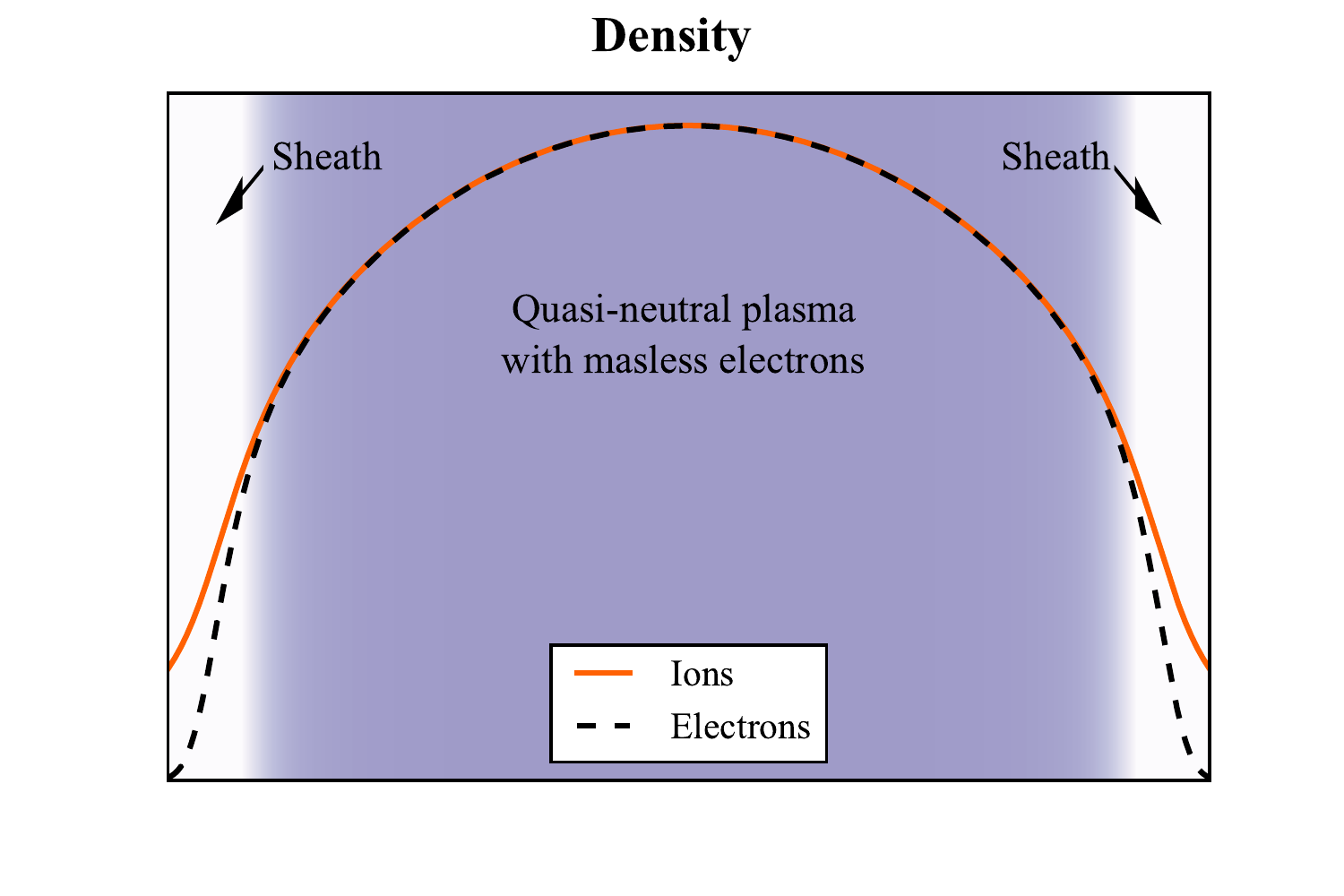} \hfill
		  \includegraphics[trim=1.5cm 1.2cm 1.5cm 0cm, clip=true,width=0.325\textwidth]{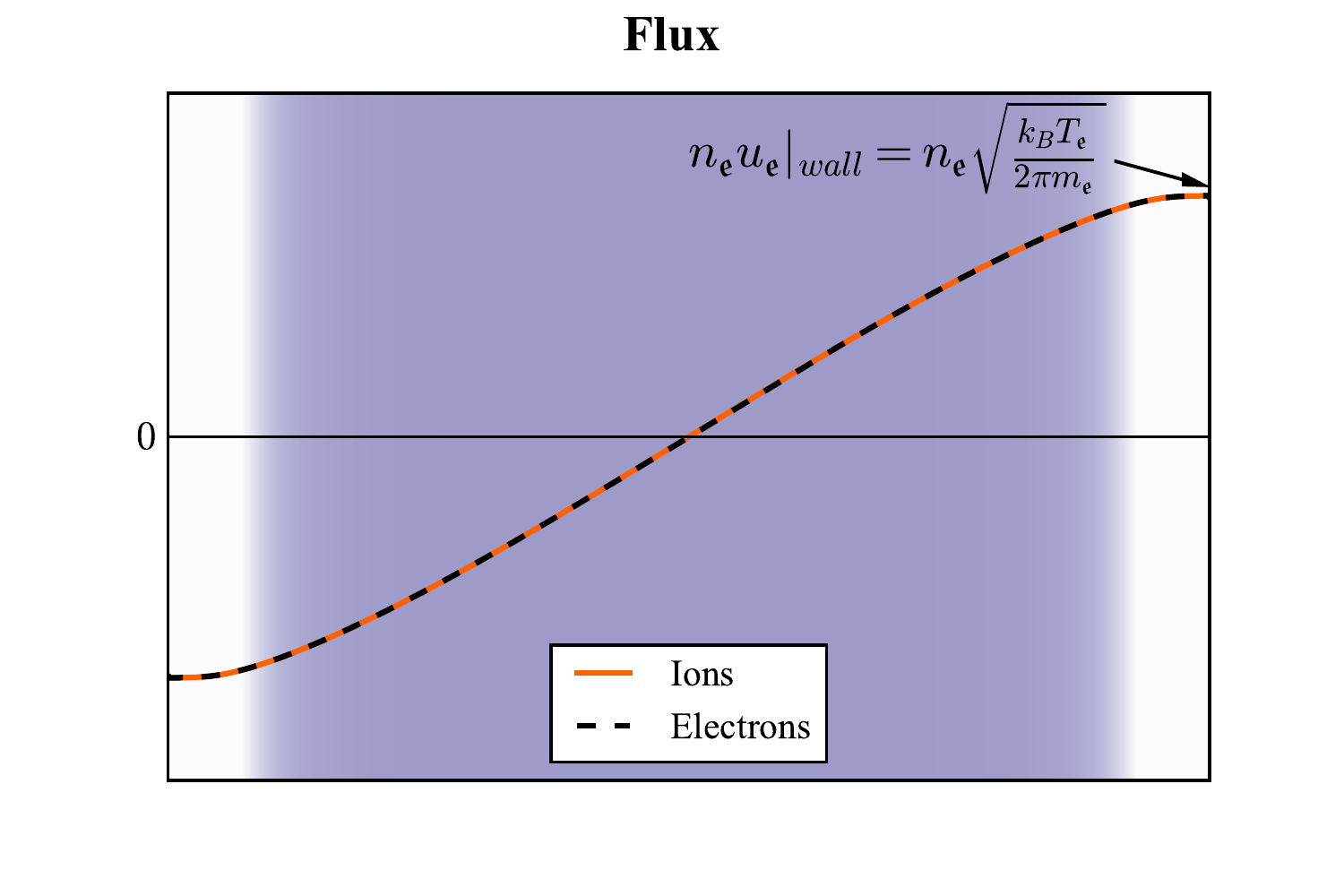}
		  \includegraphics[trim=1.5cm 1.2cm 1.5cm 0cm, clip=true,width=0.325\textwidth]{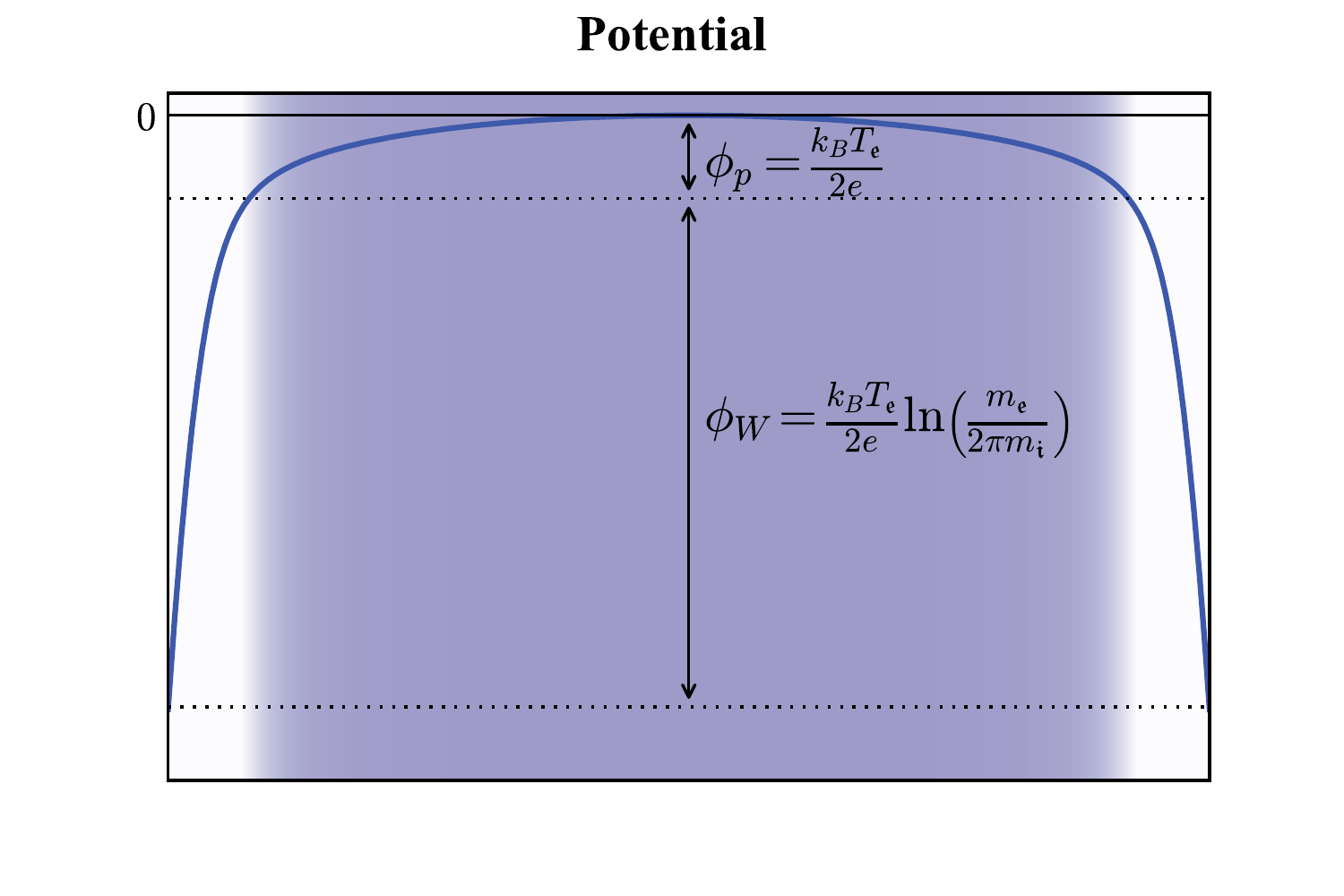}
	\caption{Solution of a plasma between two floating walls with a fluid model. The proposed numerical set-up is able to captures the physics as predicted by the theory \cite{Riemann05}.}
	\label{0_SheathExample}
\end{figure}

\section{Standard upwind finite volume discretization}\label{sec:standardDiscretization}

We present a standard discretization of the system \eqref{eq:sysND} in order to illustrate the associated numerical difficulties. An example of a simulation of a low-temperature discharge with this discretization can be found in Alvarez Laguna et al.~\cite{AlvarezLaguna18}. Alternatively, a similar discretization is described in \cite{Crispel07} in order to illustrate a standard solver of the Euler-Poisson system. 

We use a finite volume discretization where the domain $x \in [0, l]$ is divided into $N$ elements of equal length $\Delta x$. We approximate the value of the unknowns as a piecewise function inside the volume $\Omega_j$. The flux at the interfaces is approximated by a numerical flux function that is in general a function of the values on the right and left of the cell interface. The source is approximated by a piecewise constant value. After making these assumptions, the first order 1D finite volume discretization for the cell $j$ reads

\begin{multline}
\left(\begin{array}{c}
  \rhoeb\\
  \rhoionb\\
  \rhoeb\ueb \\
  \rhoionb\uionb \\  
\end{array}\right)_j^{\nPlusOne}
=
\left(\begin{array}{c}
  \rhoeb\\
  \rhoionb\\
  \rhoeb\ueb \\
  \rhoionb\uionb \\  
\end{array}\right)_{j}^{n}
-
\frac{\Delta t}{\Delta x}
\left[
\left(\begin{array}{c}
  \rhoeb\ueb \\
  \rhoionb \uib\\
  \rhoeb (\ueb^2 +  \varepsilon^{-1}) \\
  \rhoionb (\uib^2+\kappaT) \\  
\end{array}\right)_{j+1/2}^{n}
-
\left(\begin{array}{c}
  \rhoeb\ueb \\
  \rhoionb \uib\\
  \rhoeb (\ueb^2 +  \varepsilon^{-1}) \\
  \rhoionb (\uib^2+\kappaT) \\  
\end{array}\right)_{j-1/2}^{n}
\right]
+
\Delta t
\left(\begin{array}{c}
  (\rhoeb \Damk) \\
  \rhoeb\Damk\\
  \rhoeb\varepsilon^{-1} \dxb\phib \\
  -\rhoionb \dxb\phib\\  
\end{array}\right)_{j}^{n}
\end{multline}
with 
\begin{equation}\label{eq:StandardPoisson}
\phib_{j+1}^{\nPlusOne} - 2\phib_{j}^{\nPlusOne} +  \phib_{j-1}^{\nPlusOne} = \frac{\Delta x^2}{\lambdaSq}\left(\rhoebj^{n} - \rhoionbj^n\right)~~~\text{and}~~~\dxb\phib_{j}^{n} = \frac{1}{2\Delta x}\left( \phib_{j+1}^{n} - \phib_{j-1}^{n}\right).
\end{equation}
The numerical fluxes at the interfaces can be calculated with different Riemann solvers, e.g., Roe as in \cite{AlvarezLaguna18}, Lax-Friedrich as in \cite{Crispel07} or HLL as in the results of this paper using the standard discretization. We note here that the election of the Riemann solver for this problem has a small impact in the results as the numerical dissipation is dominated by the low-Mach regime of electrons, as it will be shown in the results. TVD reconstruction can improve the results of the standard discretization as shown in \cite{Crispel07,AlvarezLaguna18}.

As explained in \cite{Crispel07,Degond11}, the stability of the time discretization is restricted by a CFL condition that takes into account the convective scales of both fluids and the characteristic time scales of the source terms. The convective CFL reads
\begin{equation}
  \text{CFL}^{conv} = \frac{\Delta t |\lambda^{\elec,\ion}_{max}|}{\Delta x}~~~\text{with}~~~|\lambda^{\elec,\ion}_{max}| = \max\left(|\bar{u}_{\elec,\ion} + \bar c_{\elec,\ion}|,~|\bar{u}_{\elec,\ion} - \bar c_{\elec,\ion}|\right).
\end{equation}
Where $ \bar c_{\elec,\ion}$ is the dimensionless speed of sound of electrons and ions, i.e.,  $\varepsilon^{-1/2}$ and $\kappaT^{1/2}$, respectively. Note that due to the mass disparity between ions and electrons, the CFL condition of the electrons is typically more restrictive than this of the ions. Therefore, the convective CFL condition of electrons uses the maximum eigenvalues of electrons that are $|\lambda^{\elec}_{max}| = \max\left(|\bar{u}_{\elec} + \varepsilon^{-1/2}|,~|\bar{u}_{\elec} - \varepsilon^{-1/2}|\right)$.
Similarly, the source terms impose a constraint in the time step. The stability condition for the electrostatic force is related to the resolution of the electron plasma wave \cite{Degond08}
\begin{equation}\label{eq:plasmaFreq}
  \Delta t\ \omega_{p\elec}<1~~~\text{with}~~~\omega_{p\elec} = \sqrt{\frac{\rhoeb}{\lambdaSq\varepsilon}}.
\end{equation}
Finally, the ionization term has the stability constraint as follows
\begin{equation}
  \Delta t\ \nuionization < 1.
\end{equation}
The stability condition reads
\begin{equation}
  \max\left(\text{CFL}^{conv},~\Delta t\ \omega_{p\elec},~\Delta t\ \nuionization \right)_{i\in N}<1.
\end{equation}
As it can be seen in Table \ref{tableScales}, if the cell size is larger than the Debye length, the most restrictive constraint is the resolution of the electron plasma waves. If the Debye length is resolved, then the convective condition is sufficient to fulfil the condition $\Delta t< \omega^{-1}_{pe}$. As shown in \cite{AlvarezLaguna18,Joncquieres18}, when this scheme does not spatially resolve the Debye length, the simulation leads to large spurious charge separation errors that can excite plasma modes, leading to an erroneous solution. Similarly, the truncation error of the upwind discretization of the fluxes of the electrons leads to a large error in the flux of the electrons. As it will be shown in the results and previously noted in \cite{Crispel04,Crispel07}, this is due to the low-Mach regime of the electron when the bulk speed is much smaller than the thermal speed.
\begin{figure}[h]
	\centering
		  \includegraphics[trim=0cm 0cm 0cm 0cm, clip=true,width=0.6\textwidth]{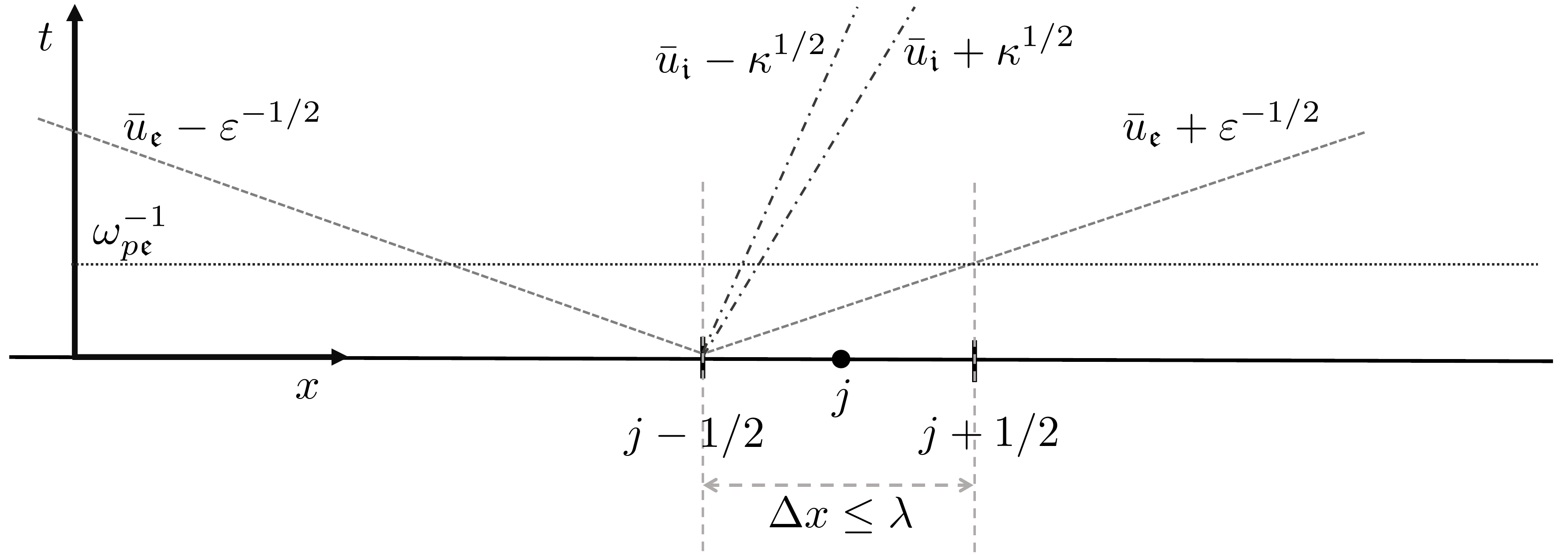} 
	\caption{Illustration of the characteristic scales for the discretization of the two-fluid electrostatic problem. We assume conditions where the ions and electrons travel at a positive speed of similar magnitude, ($\uionb \sim \ueb$) and $\uionb\gg\kappa^{1/2}$ and $\ueb\ll\varepsilon^{-1/2}$. The ionization frequency is supposed to be such that $\Damk \ll \omega_{p\elec}$ and thus is not in the figure.}
	\label{2_NoAP}
\end{figure}

\section{Acoustic/transport operator splitting strategy}

We present an alternative to the standard discretization of section \ref{sec:standardDiscretization}. We present a novel operator splitting strategy that decouples the acoustic and transport phenomena of electrons. In 1D, this method is analogous to an explicit Lagrange-Projection \cite{Chalons16} method. Nevertheless, the present splitting does not need a moving Lagrangian mesh and can be naturally expressed for multi-dimensional problems (see Chalons et al.~\cite{Chalons16} for the extension of the operator splitting method to multiple dimensions). 

We propose to approximate the system of equations \eqref{elMass_ND}-\eqref{Pois_ND} in the successive solution of the following systems. The first step solves for the electron acoustic system together with the Poisson equation and the second solves for the electron transport (advection) and the ion equations. It should be noted that the first system contains the small scales related to the parameters $\varepsilon$ and $\lambda$, whereas the second system solves for the slow dynamics, of order $\mathcal{O}(1)$, as explained in Section \ref{sec:AsymBehavior}.

\paragraph{Electron acoustic and electrostatic system}
\begin{subequations}
\begin{eqnarray}\label{eq:acoustic}
    \dtb \rhoeb + \rhoeb\dxb \ueb &=&  0,\label{elMass_ND_acoustic}\\
    \dtb (\rhoeb\ueb) + \rhoeb\ueb \dxb \ueb + \dxb \pelecb&=&\frac{\rhoeb}{\varepsilon} \dxb\phib,\label{elMom_ND_acous}\\
    \dxxb \phib &=& \frac{\rhoeb - \rhoionb}{\lambdaSq}, \label{Pois_ND_acoustic}
\end{eqnarray}\label{acoustic}
\end{subequations}
\paragraph{Electron transport and ion dynamics}
\begin{subequations}
\begin{eqnarray}
    \dtb \rhoeb + \ueb \dxb \rhoeb &=&  \rhoeb \Damk ,\label{elMass_ND_transport}\\
    \dtb (\rhoeb\ueb) + \ueb  \dxb (\rhoeb \ueb)&=& 0,\label{elMom_ND_transport}\\
    \dtb  \rhoionb + \dxb ( \rhoionb \uionb) &=& \rhoeb \Damk,\label{IonMass_ND_transport}\\
    \dtb (\rhoionb \uib) + \dxb \left[\rhoionb \left(\uionb^2 + \kappaT\right) \right]&=&-\rhoionb \dxb\phib ,\label{IonMom_ND_transport}
    \end{eqnarray}\label{transport}
\end{subequations}
where the electron pressure $\pelecb = \varepsilon^{-1}\rhoeb$, in eq.~\eqref{elMom_ND_acous}. 

\subsubsection*{Strategy to solve the equations}

Given a state at the time $(\rhoeb, \rhoeb\ueb, \rhoionb, \rhoionb\uionb, \phib)^n_j$ at the time $t^n$ and the cell center $x_j$. The scheme is split into

\begin{enumerate}
  \item By numerically solving the system \eqref{acoustic}, we update the state $(\rhoeb, \rhoeb\ueb, \phib)^n_j$ to the value at $t^{\nPlusOneM}$, i.e., $(\rhoeb, \rhoeb\ueb, \phib)^{\nPlusOneM}_j$.
  \item By numerically solving the system \eqref{transport}, we update the state $(\rhoeb, \rhoeb\ueb, \rhoionb, \rhoionb\uionb, \phib)^{\nPlusOneM}_j$ to the value $t^{n+1}$, i.e., $(\rhoeb, \rhoeb\ueb, \rhoionb, \rhoionb\uionb, \phib)^{n+1}_j$.
\end{enumerate}

In the following, we present the properties and description of the two steps of the numerical scheme.

\begin{figure}[h]
	\centering
		  \includegraphics[trim=0cm 1cm 3cm 0cm, clip=true,width=0.49\textwidth]{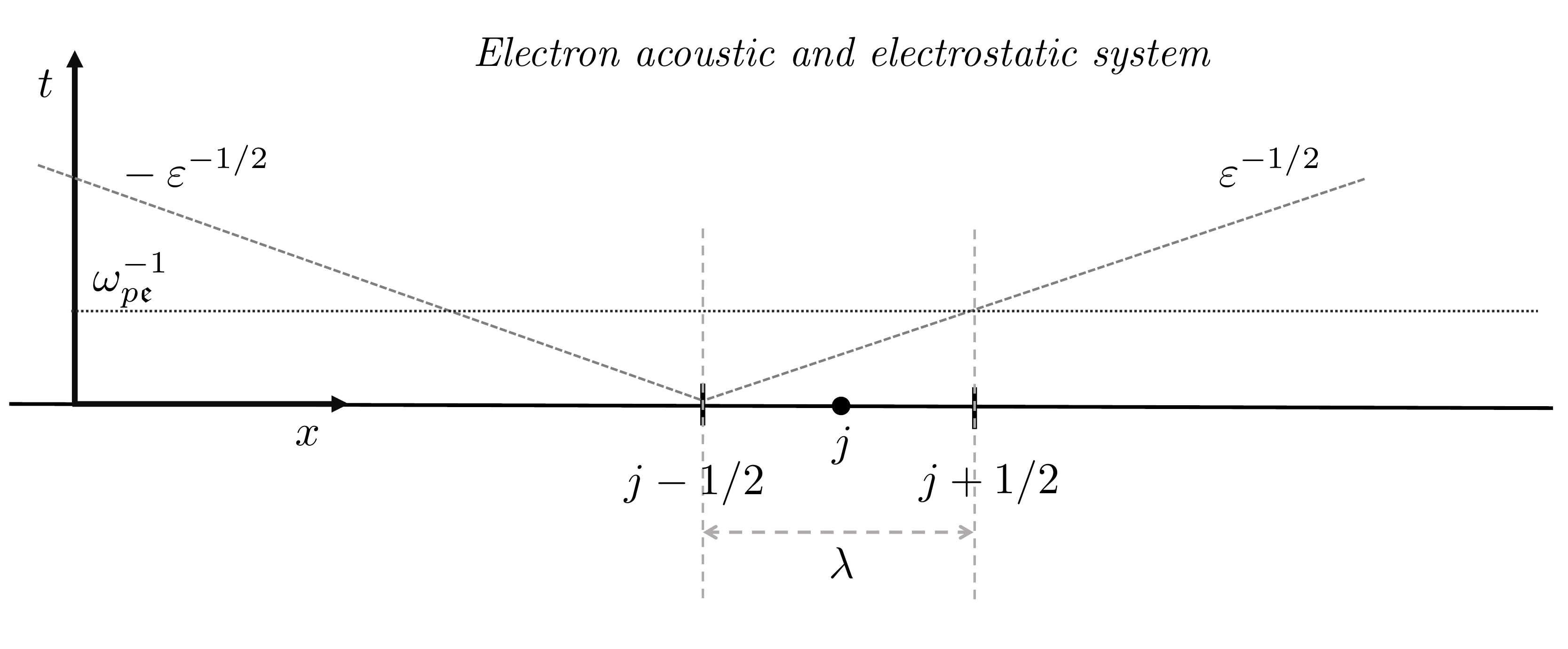} 
 		  \includegraphics[trim=0cm 1cm 3cm 0cm, clip=true,width=0.49\textwidth]{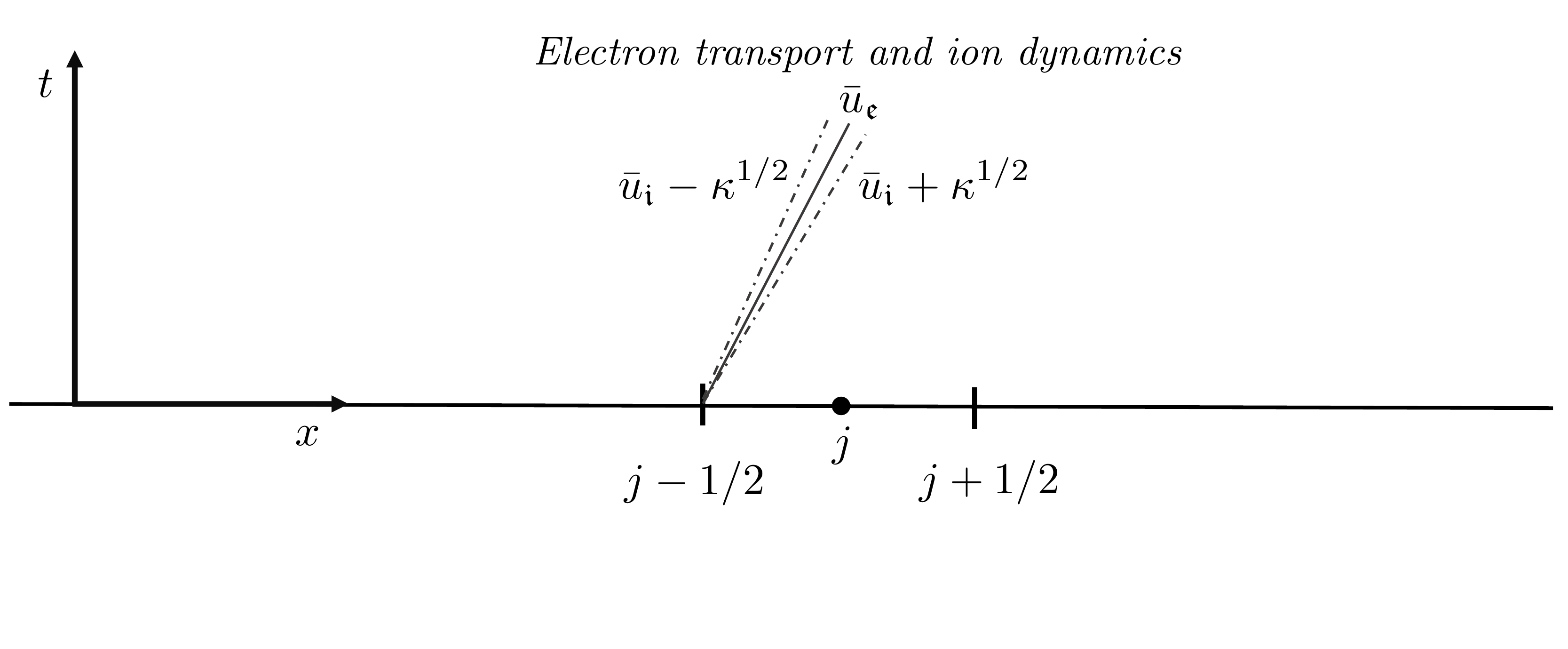} 
	\caption{Illustration of the characteristic scales for the operator splitting. }
	\label{2_NoAP}
\end{figure}

\subsection{Properties and discretization of the electron acoustic and electrostatic system}\label{sec:ElectronAcoustic}

The system eq.~\eqref{acoustic} can be written in conservative form, i.e., the compression terms written as fluxes. To do this, we define the variable $\tauelec \equiv 1/\nelec$. In this variable, the problem \eqref{acoustic} reads
\begin{subequations}
\begin{eqnarray}
    \dtb \tauelec - \tauelec\dxb \ueb &=&  0,\label{elMass_ND_acoustic}\\
    \dtb \ueb + \tauelec \dxb \pelecb&=&\frac{1}{\varepsilon} \dxb\phib,\label{elMom_ND_acoustic}\\
    \lambdaSq\dxxb \phib &=& \frac{1}{\tauelec} - \rhoionb. \label{Pois_ND_acoustic}
\end{eqnarray}\label{acoustic_NewVar}
\end{subequations}

We approximate $\tauelec(x,t)\dxb$ by the solution at time $t^n$, i.e., $\tauelec(x,t^n)\dxb$. We define $dm = \nelec(x,t^n)dx$. By using this new variable, we obtain the following system of equations for the electrons
\begin{subequations}
\begin{eqnarray}
    \dtb \tauelec - \dm \ueb &=&  0,\label{elMass_ND_acoustic_m}\\
    \dtb \ueb +  \dm \pelecb&=&\frac{1}{\varepsilon} \dxb\phib.\label{elMom_ND_acoustic_m}
\end{eqnarray}\label{acoustic_NewVar_m}
\end{subequations}

The left-hand-side of the system is conservative in this new variable. The eigenvalues of the homogeneous system \label{acoustic_NewVar} are $\lambda^{acous} = (-\varepsilon^{-1/2}, \varepsilon^{-1/2})$. For that reason, the system is called acoustic. Note that the source term in equation \eqref{elMom_ND_acoustic_m} is still with the previous space variable $x$. As this term cannot be written in conservative form in the variable $m$, it will be treated in the discretization as a source term.

\subsubsection{Discretization of the electron acoustic and electrostatic system}
The system \eqref{acoustic_NewVar_m} with Poisson's eq.~\eqref{Pois_ND_acoustic} will be discretized in time by using a semi-implicit discretization. In this approach, we want to discretize the Lorentz force in the electrons and the electron density in Poisson's equation implicitly, as follows

\begin{subequations}
\begin{eqnarray}
    \frac{\tauelec^{\nPlusOneM} - \tauelec^n}{\Delta t} - \left[\dm \ueb\right]^{\nPlusOneM} &=&  0,\label{elMass_ND_acoustic_m_semiDiscrete}\\
    \frac{\ueb^{\nPlusOneM} - \ueb^n}{\Delta t} + \left[ \dm \pelecb\right]^n&=&\frac{1}{\varepsilon} \left[\dxb\phib\right]^{\nPlusOneM}.\label{elMom_ND_acoustic_m_semiDiscrete}\\
     \lambdaSq\left[\dxxb \phib\right]^{\nPlusOneM} &=& \frac{1}{\tauelec^{\nPlusOneM}} - \rhoionb^n. \label{Pois_ND_acoustic_semiDiscrete}
\end{eqnarray}\label{acoustic_NewVar_m_fullDiscrete}
\end{subequations}

The full discretization of system \eqref{acoustic_NewVar_m_fullDiscrete} at $t^{\nPlusOneM}$ is given by numerical scheme \ref{pr:scheme}. 

\begin{numScheme} \label{pr:scheme} The discretization of the electron acoustic and electrostatic system \eqref{acoustic} reads as follows. The discretized electron density is 
\begin{equation}\label{eq:ElectronNPlusOneM}
  \rhoebj^{\nPlusOneM} = \rhoebj^n\frac{\left(1 + \frac{\rhoionbj^n\Delta t^2}{\varepsilon\lambdaSq}\right)}{\left[1 + \frac{\Delta t}{\Delta x}\left(\uebjphalf^n - \uebjmhalf^n \right) - \Delta t^2\rhoebj^n \left[ \dmm \pelecb\right]^n_j + \frac{\Delta t^2 \rhoebj^n}{\varepsilon \lambdaSq}\right]}.
\end{equation}
with the discretization of the pressure Laplacian as
\begin{equation}\label{PressureLaplacianDiscrete}
  \left[ \dmm \pelecb\right]^n_j= \frac{\tauelec^n}{\varepsilon\Delta x^2}\left(\ln \rhoebjpone^n + \ln \rhoebjmone^n - 2\ln \rhoebj^n \right)
\end{equation}
and the velocity at the cell interface
\begin{equation}
  \uebjphalf^n= \frac{\uebR + \uebL}{2} - \frac{ \nelecbjphalf\varepsilon^{-1/2}}{2}f_a(M_\elec)\left(\taueR - \taueL\right) ~~~\text{with}~~~ \nelecbjphalf = \frac{\rhoebR+ \rhoebL}{2},
  \label{velocityInterface}
\end{equation}
where $f_a(M_\elec) = \mathcal{O}(M_\elec)$. For this function, we choose the formula proposed by Liou \cite{liou3}, as follows
\begin{equation}\label{eq:Liou}
  f_a(\bar{M}) = \frac{\sqrt{\left(1 - M_o\right)^2\bar{M}^2 + 4M_o^2}}{1 - M_o^2}
\end{equation}
where 
\begin{equation}
  \bar{M} = \frac{\sqrt{\rhoebL}\uebL + \sqrt{\rhoebR}\uebR }{\sqrt{\rhoebL} + \sqrt{\rhoebR}}\varepsilon^{1/2}~~~\text{and}~~~ M_o = \min\left(1,~\max\left(M_{co}, \bar{M}\right)\right).
\end{equation}
The cut-off Mach $M_{co}$ is chosen to avoid having a null numerical dissipation when the electrons have zero velocity. 

The discretized potential is solved by the following expression
\begin{equation}\label{eq:PoissonDiscrete}
\phib_{j+1}^{\nPlusOneM} - 2\phib_{j}^{\nPlusOneM} +  \phib_{j-1}^{\nPlusOneM} = \frac{\Delta x^2}{\lambdaSq}\left(\rhoebj^{\nPlusOneM} - \rhoionbj^n\right).
\end{equation}

The discretized electron velocity is computed as follows
\begin{equation}
  \uebj^{\nPlusOneM} = \uebj^n - \frac{\Delta t }{\rhoebj^n\Delta x} \left( \pelecbjphalf^n - \pelecbjmhalf^n \right) + \frac{\Delta t }{\varepsilon}\dxb\phib_{j}^{\nPlusOneM},\label{velocityElecDiscrete}
\end{equation}
where the gradient of the electric potencial is discretized as
\begin{equation}
\left[\dxb\phib^{\nPlusOneM}\right]_{j} = \frac{1}{2\Delta x}\left( \phib_{j+1}^{\nPlusOneM} - \phib_{j-1}^{\nPlusOneM}\right),
\end{equation}
and the pressure flux from the expression
\begin{equation}\label{eq:pressureFlux}
  \pelecbjphalf = \varepsilon^{-1}\frac{\left(\rhoebR + \rhoebL\right)}{2} - \frac{\nelecbjphalf\varepsilon^{-1/2}}{2}f_a(M_\elec)\left(\uebR - \uebL\right).
\end{equation}
\end{numScheme}

The previous numerical scheme is obtained as follows. In order to build a scheme with an asymptotic-preserving property, we follow a similar approach to this of Dimarco et al.~\cite{Dimarco18}. The divergence of the velocity at time $t^{\nPlusOneM}$ in eq.~\eqref{elMass_ND_acoustic_m_semiDiscrete} is computed by taking the divergence of eq.~\eqref{elMom_ND_acoustic_m_semiDiscrete}. This divergence of the velocity reads
\begin{equation}
    \left[\dm \ueb\right]^{\nPlusOneM} =   \left[\dm \ueb\right]^n - \Delta t \left[ \dmm \pelecb\right]^n + \frac{\Delta t}{\varepsilon} \left[\dm\dxb\phib\right]^{\nPlusOneM}.\label{divergenceOfSpeed}
\end{equation}

We discretize the term due to the Lorentz force at the time $t^{\nPlusOneM}$ by using Poisson's eq.~\eqref{Pois_ND_acoustic_semiDiscrete}, as follows
\begin{equation}
  \left[\dm\dxb\phib\right]^{\nPlusOneM} = \frac{1}{\rhoeb^{\nPlusOneM}}\left[\dxxb\phib\right]^{\nPlusOneM} = \frac{1}{\lambdaSq}\left( 1 - \nion^n\tauelec^{\nPlusOneM}\right).
\end{equation}
Note that the ion density is computed at time $t^n$ in the Poisson equation. This approximation is justified by the fact that the ions move slower than the electrons in the scales for the acoustic step. The divergence of the velocity is then introduced in the equation for the conservation of mass in the acoustic step. This reads
\begin{equation}
  \frac{\tauelec^{\nPlusOneM} - \tauelec^n}{\Delta t} - \left[\dm \ueb\right]^{n} = \frac{\Delta t}{\varepsilon\lambdaSq}\left[1 - \nion^n\tauelec^{\nPlusOneM}\right] - \Delta t \left[ \dmm \pelecb\right]^n.\label{elMass_ND_acoustic_m_semiDiscrete_2}
\end{equation}
We note that the LHS of the equation is the same as in an explicit discretization, whereas the RHS are terms that appear due to the implicitation of the scheme. The first term of the RHS is a diffusion term that is due to the Lorentz force and the second is a diffusion term due to the pressure gradient. We highlight the fact that the diffusion due to the Lorentz force, which will stabilize the scheme, appears without the need of solving Poisson's equation and additionally, it is linear in $\tauelec^{\nPlusOneM}$.

In order to discretize in space eq.~\eqref{elMass_ND_acoustic_m_semiDiscrete_2}, we follow the formalism of Chalons et al.~\cite{Chalons16}.
\begin{equation}
  \left(1 + \frac{\rhoionbj^n\Delta t^2}{\varepsilon\lambdaSq}\right)\tauelecj^{\nPlusOneM} =  \tauelecj^n\left[1 + \frac{\Delta t}{\Delta x}\left(\uebjphalf^n - \uebjmhalf^n \right) \right] + \frac{\Delta t^2}{\varepsilon \lambdaSq} - \Delta t^2 \left[ \dmm \pelecb\right]^n_j.
\end{equation}
The pressure diffusion term is discretized as follows. By changing from the variable $m$ to the space variable, the pressure gradient reads $\dm \pelecb = \rhoeb^{-1} \dxb \pelecb$. By assuming the electrons to be isothermal, the pressure is $\pelecb = \rhoeb\varepsilon^{-1}$ and the pressure gradient in the mass variable $\dm \pelecb = \varepsilon^{-1} \dxb \ln \rhoeb$. Therefore, the Laplacian in the mass variable reads $\dmm  \pelecb = (\varepsilon \rhoeb)^{-1}\dxxb \ln \rhoeb$. The space discretization of this term is based on centered difference as follows
\begin{equation}\label{PressureLaplacianDiscrete}
  \left[ \dmm \pelecb\right]^n_j= \frac{\tauelec^n}{\varepsilon\Delta x^2}\left(\ln \rhoebjpone^n + \ln \rhoebjmone^n - 2\ln \rhoebj^n \right).
\end{equation}
Similarly, the velocity at the interface is approximated by a Riemann solver that is described in Appendix A. The original velocity at the interface reads
\begin{equation}
  \uebjphalf^n= \frac{\uebR + \uebL}{2} - \frac{ \nelecbjphalf\varepsilon^{-1/2}}{2}\left(\taueR - \taueL\right)~~~\text{with}~~~ \nelecbjphalf = \frac{\rhoebR+ \rhoebL}{2},
\label{velocityInterface}
\end{equation}
where the subscripts $R$ and $L$ refer to the values on the right and left of the interface $j + 1/2$. As explained in Appendix B, the truncation error is proportional to $\varepsilon^{-1/2} \dxxb \tauelec \Delta x$. As suggested in~\cite{Chalons16}, in order to control this error when $\varepsilon$ is small, we rescale the numerical error produced by an imbalanced numerical dissipation in the low Mach regime. This yields 
\begin{equation}
  \uebjphalf^n= \frac{\uebR + \uebL}{2} - \frac{ \nelecbjphalf\varepsilon^{-1/2}}{2}f_a(M_\elec)\left(\taueR - \taueL\right),
  \label{velocityInterface}
\end{equation}
where $f_a(M_\elec) = \mathcal{O}(M_\elec)$. For this function, we choose the formula proposed by Liou \cite{liou3}, as follows
\begin{equation}\label{eq:Liou}
  f_a(\bar{M}) = \frac{\sqrt{\left(1 - M_o\right)^2\bar{M}^2 + 4M_o^2}}{1 - M_o^2},
\end{equation}
where 
\begin{equation}
  \bar{M} = \frac{\sqrt{\rhoebL}\uebL + \sqrt{\rhoebR}\uebR }{\sqrt{\rhoebL} + \sqrt{\rhoebR}}\varepsilon^{1/2}~~~\text{and}~~~ M_o = \min\left(1,~\max\left(M_{co}, \bar{M}\right)\right).
\end{equation}
The cut-off Mach $M_{co}$ is chosen to avoid having a null numerical dissipation when the electrons have zero velocity. 

With this equation, the electron density at the time $t^{\nPlusOneM}$ is calculated as follows
\begin{equation}\label{eq:ElectronNPlusOneM}
  \rhoebj^{\nPlusOneM} = \rhoebj^n\frac{\left(1 + \frac{\rhoionbj^n\Delta t^2}{\varepsilon\lambdaSq}\right)}{\left[1 + \frac{\Delta t}{\Delta x}\left(\uebjphalf^n - \uebjmhalf^n \right) - \Delta t^2\rhoebj^n \left[ \dmm \pelecb\right]^n_j + \frac{\Delta t^2 \rhoebj^n}{\varepsilon \lambdaSq}\right]}.
\end{equation}

After solving this equation, the electric potential is computed by solving Poisson's by the linear system described in eq.~\eqref{eq:PoissonDiscrete}.

The last part of the acoustic step is to solve the electron velocity equation. For this, we use an approximate Riemann solver in order to discretize the pressure flux in eq.~\eqref{elMom_ND_acoustic_m_semiDiscrete}, as follows
\begin{equation}
  \uebj^{\nPlusOneM} = \uebj^n - \frac{\Delta t }{\rhoebj^n\Delta x} \left( \pelecbjphalf^n - \pelecbjmhalf^n \right) + \frac{\Delta t }{\varepsilon}\dxb\phib_{j}^{\nPlusOneM}.\label{velocityElecDiscrete}
\end{equation}
The gradient of the electric potencial is discretized as
\begin{equation}
\left[\dxb\phib^{\nPlusOneM}\right]_{j} = \frac{1}{2\Delta x}\left( \phib_{j+1}^{\nPlusOneM} - \phib_{j-1}^{\nPlusOneM}\right)
\end{equation}
The pressure flux from the Riemann solver in Appendix A reads
\begin{equation}\label{eq:pressureOriginal}
  \pelecbjphalf = \varepsilon^{-1}\frac{\left(\rhoebR + \rhoebL\right)}{2} - \frac{\nelecbjphalf\varepsilon^{-1/2}}{2}\left(\uebR - \uebL\right).
\end{equation}
As done previously, the truncation error of eq.~\eqref{velocityElecDiscrete} is proportional to $\left(\varepsilon^{-1/2}\dxxb \ueb \Delta x\right)$ (see Appendix B). In this case, we choose to balance the numerical dissipation with the technique proposed by Liou \cite{liou3}.
\begin{equation}\label{eq:pressureFlux}
  \pelecbjphalf = \varepsilon^{-1}\frac{\left(\rhoebR + \rhoebL\right)}{2} - \frac{\nelecbjphalf\varepsilon^{-1/2}}{2}f_a(M_\elec)\left(\uebR - \uebL\right).
\end{equation}
As in Chalons et al.~\cite{Chalons16}, by preconditioning the numerical dissipation with the factor $f_a$, the truncation error of the numerical system is of order $\mathcal{O}(\Delta x)$.

As the Lorentz force is treated implicitly, the scheme is shown to be unconditionally stable for the plasma wave frequency. Additionally, with the preconditioning of the low-Mach regime, the electrons can have a larger time step, not limited by the sound waves, as it will be shown in the results. The linear stability analysis to determine the CFL condition is not trivial due to the non-linear term $\left[ \dmm \pelecb\right]_j^n$ and hence is left for a future work.

\subsubsection{Asymptotic-preserving property of the acoustic and electrostatic system}

In this section, we prove the consistency of the acoustic and electrostatic system when $\lambdaSq$ and $\varepsilon$ tend to zero with respect to the asymptotic limit of system \eqref{eq:sysNN} explained in Section \ref{sec:AsymBehavior}.

\begin{property}\label{prop_quasiNeutralityExpansion}
   In the limits $\varepsilon\rightarrow0$ and $\lambdaSq\rightarrow0$, the discretized electron density at $t^{\nPlusOneM}$ reads
   \begin{equation}\label{quasiNeutralityExpansion}
     \underset{\lambdaSq \rightarrow 0}{\lim\limits_{\varepsilon \rightarrow 0}} \rhoebj^{\nPlusOneM} - \rhoionbj^n = 0.
   \end{equation}
\end{property}
\begin{proof}
First, we may rewrite $f_a(M_\elec)$ at the interface $j+1/2$ in the limit $\varepsilon\rightarrow 0$ as
   \begin{align*}
     f_a(M_\elec) = \sqrt{C_{j+1/2}^1 + \varepsilon C_{j+1/2}^2} ~~~\text{with}~~~ C_{j+1/2}^1 &= \frac{\left(1 - M_o\right)^2}{\left(1 - M_o^2\right)^2}\left(\frac{\sqrt{\rhoebj}\uebj + \sqrt{\rhoebjpone}\uebjpone }{\sqrt{\rhoebj} + \sqrt{\rhoebjpone}}\right)^2, \\
     C_{j+1/2}^2 &= \frac{4M_o^2}{\left(1 - M_o^2\right)^2}\left(\frac{\sqrt{\rhoebj}\uebj + \sqrt{\rhoebjpone}\uebjpone }{\sqrt{\rhoebj} + \sqrt{\rhoebjpone}}\right)^2.
   \end{align*}
   Then, injecting~\eqref{velocityInterface} and~\eqref{eq:pressureFlux} into~\eqref{eq:ElectronNPlusOneM} leads to
   \begin{subequations}
     \begin{align}
       \rhoebj^{\nPlusOneM} &= \rhoebj^n
       \frac{
         \left({\varepsilon\lambdaSq} + {\rhoionbj^n\Delta t^2}\right)
       }{
         \lambdaSq \left[\varepsilon K_j^1 - \varepsilon^{1/2} K_j^{1/2} - K_j^0\right] + \Delta t^2 \rhoebj^n
       } \\
       K_j^1    &= 1 + \frac{\Delta t}{\Delta x}\frac{\uebjpone - \uebjmone}{2}, \\
       K_j^{1/2} &= \frac{\Delta t}{\Delta x}\left[\frac{ \nelecbjphalf}{2}\sqrt{C_{j+1/2}^1 + \varepsilon C_{j+1/2}^2}\left(\tauelecjpone - \tauelecj\right)
         - \frac{ \nelecbjmhalf}{2}\sqrt{C_{j-1/2}^1 + \varepsilon C_{j-1/2}^2}\left(\tauelecj - \tauelecjmone\right)\right],\\
       K_j^0 & = \frac{\Delta t^2}{\Delta x^2}\left(\ln \rhoebjpone^n + \ln \rhoebjmone^n - 2\ln \rhoebj^n \right).
     \end{align}
     \label{eqt_proof:ElectronNPlusOneM}
   \end{subequations}
   One observes that $\rhoebj^{\nPlusOneM}$ is a function of $\varepsilon$ and $\lambdaSq$ continuous in $(0,0)$, then the two limits commute and \eqref{quasiNeutralityExpansion} holds.
\end{proof}
We highlight that eq.~\eqref{quasiNeutralityExpansion} is the quasi-neutral limit that occurs at scales $\mathcal{O}(\lambdaSq^{-1})$, as shown in eq.~\eqref{quasineutrality}.

\begin{property}
   Assuming $\rhoebj^n = \rhoionbj^n \neq 0$, then, in the limits $\varepsilon\rightarrow0$ and $\lambdaSq\rightarrow0$, the discretized electric potential at $t^{\nPlusOneM}$ satisfies
   \begin{equation}
     \label{eqt:dxxphi=dxxlnrhoe}
     \underset{\lambdaSq \rightarrow 0}{\lim\limits_{\varepsilon \rightarrow 0}} [\dxxb \phib]^{\nPlusOneM}_j - \left[\dxxb \ln \rhoeb\right]^{\nPlusOneM}_j = 0, 
   \end{equation}
   where the operator $[\dxxb ]_j $ stands for the central discretization in space at the cell $j$. This expression is the discretization of the divergence of eq.~\eqref{BoltzmannElec}. 
\end{property}

 \begin{proof}
   By injecting eq.~\eqref{eq:ElectronNPlusOneM} into eq.~\eqref{eq:PoissonDiscrete}, we obtain the following expression for Poisson's equation discretization
   \begin{align}\label{PoissonExpansion1}
     [\dxxb \phib^{\nPlusOneM}]_j &= 
     \frac{
       \varepsilon\rhoebj^n - \rhoionbj^n\left[\varepsilon K_j^1 - \varepsilon^{1/2} K_j^{1/2} - K_j^0\right]
     }{
       \lambdaSq \left[\varepsilon K_j^1 - \varepsilon^{1/2} K_j^{1/2} - K_j^0\right] + \Delta t^2 \rhoebj^n
     },
   \end{align}
   where the coefficients $K_j^i$ are defined in \eqref{eqt_proof:ElectronNPlusOneM}.
     
   This is again continuous in $(\varepsilon,\lambdaSq) = (0,0)$, thus the two limits commute and provide
   \[
   \underset{\lambdaSq \rightarrow 0}{\lim\limits_{\varepsilon \rightarrow 0}} [\dxxb \phib]^{\nPlusOneM}_j = 
   \frac{\rhoionbj^n}{\rhoebj^n\Delta x^2}\left(\ln \rhoebjpone^n + \ln \rhoebjmone^n - 2\ln \rhoebj^n \right).
   \]
   Using property~\ref{prop_quasiNeutralityExpansion}, we may replace in this expression $\rhoionbj^n$ by $\rhoebj^{\nPlusOneM}$ in this expression.
   
   And using the continuity of the logarithm, the hypothesis $\rhoebj^n = \rhoionbj^n \neq 0$ and property~\ref{prop_quasiNeutralityExpansion}, we may replace $\rhoebj^n$ in this expression to obtain \eqref{eqt:dxxphi=dxxlnrhoe}.

 \end{proof}

As seen above, the advantage of this method is that it is AP by solving a standard Poisson equation, i.e., eq.~\eqref{eq:PoissonDiscrete}. However, when this equation is discretized by a standard explicit discretization, i.e., eq.~\eqref{eq:StandardPoisson}, the electric potential tends to infinity when $\lambdaSq\rightarrow0$ if the $\Delta x$ is not small enough, resulting in a large numerical error. 

%

In the discrete electron momentum equation of the acoustic step in eq.~\eqref{velocityElecDiscrete}, we can find problems of consistency due to the small Mach number ($\varepsilon\to0$). In the following property we explain how these problems are solved by preconditioning the numerical dissipation.

\begin{property} If $f_a=\mathcal{O}(\varepsilon^{1/2})$, then, the truncation error of the electron momentum equation in the acoustic step in eq.~\eqref{elMom_ND_acoustic} is uniform with respect to $\varepsilon < 1$.
\end{property} 
\begin{proof}
The proof is analogous to the one presented in \cite{Chalons16,Chalons17b}. We assume that the variables describe a smooth flow so we can expand them in Taylor series to study the truncation error (see Appendix B for the detailed derivation of the truncation error). With this development, we use the numerical pressure flux of \eqref{eq:pressureFlux}, and we neglect second order terms, i.e., $\mathcal{O}\left(\Delta x^2\right)$, $\mathcal{O}\left(\Delta t^2\right)$. Consequently, the discretization is consistent with 
\begin{equation}
  \dtb \ueb + \tauelec \varepsilon^{-1}\dxb \rhoeb - \varepsilon^{-1} \dxb \phib = \mathcal{O}\left(\Delta t\right) + \mathcal{O}\left(f_a\varepsilon^{-1/2}\Delta x\right) + \mathcal{O}\left(\Delta x\right).
\end{equation}
The function $f_a$ in eq.~\eqref{eq:Liou} is a continuous function of $\varepsilon$, as follows,
\begin{equation}
f_a(\varepsilon) = \sqrt{C_{j+1/2}^1 + \varepsilon C_{j+1/2}^2}.
\end{equation}
Assuming that $C_{j+1/2}^1$ and $C_{j+1/2}^2$ are of order one, then, $f^a = \mathcal{O}(\varepsilon^{1/2})$ and therefore the truncation error is uniform with respect to $\varepsilon < 1$.
\end{proof}


\subsection{Discretization of the electron transport and ion dynamics system}\label{sec:motivation}

In this subsection we present the discretization for the system \eqref{transport}. As explained before, this system represents the larger scales related to the transport of electrons and the ion scales.

\subsubsection{Discretization of the electron transport system}
The transport system for the electrons (eqs.\eqref{elMass_ND_transport}-\eqref{elMom_ND_transport}) is a hyperbolic system. By using the upwind scheme proposed in Chalons et al.~\cite{Chalons16}, we can discretize this system by the following numerical scheme.

\begin{numScheme}
The discretized transport system for the electrons (eqs.\eqref{elMass_ND_transport}-\eqref{elMom_ND_transport}) reads
\begin{multline}
\left(\begin{array}{c}
  \rhoeb\\
  \rhoeb\ueb
\end{array}\right)_j^{\nPlusOne}
=
\left(\begin{array}{c}
  \rhoeb\\
  \rhoeb\ueb
\end{array}\right)_j^{\nPlusOneM}
-
\frac{\Delta t}{\Delta x}
\left[
\uebjphalf^n 
\left(\begin{array}{c}
  \rhoeb\\
  \rhoeb\ueb
\end{array}\right)_{L/R}^{\nPlusOneM}
-
\uebjmhalf^n 
\left(\begin{array}{c}
  \rhoeb\\
  \rhoeb\ueb
\end{array}\right)_{L/R}^{\nPlusOneM}
\right]
\\+
\frac{\Delta t}{\Delta x}
\left(\uebjphalf^n -  \uebjmhalf^n\right)
\left(\begin{array}{c}
  \rhoeb\\
  \rhoeb\ueb
\end{array}\right)_j^{\nPlusOneM}
+
\left(
\begin{array}{c}
 \left(\rhoeb \Damk\right)\\
0
\end{array}\right)^{\nPlusOneM}_j,
\end{multline}
where the velocity in the interface is defined in eq.~\eqref{velocityInterface} and
\begin{equation}
\left(\begin{array}{c}
  \rhoeb\\
  \rhoeb\ueb
\end{array}\right)_{L/R}
 = \left\{
 \begin{array}{r l}
\left(\begin{array}{c}
  \rhoeb\\
  \rhoeb\ueb
\end{array}\right)_{L} & \text{if~~} \uebjphalf^n \geq 0\\
\left(\begin{array}{c}
  \rhoeb\\
  \rhoeb\ueb
\end{array}\right)_{R} & \text{else.}
 \end{array}\right .
\end{equation}
\end{numScheme}
As in Chalons et al.~\cite{Chalons16}, the CFL condition associated to the transport system is
\begin{equation}
\max_{\ {j\in N}\ }\left(\frac{\Delta t}{\Delta x} \left( (\uebjmhalf^n)^+ - (\uebjphalf^n)^-\right),~\Delta t\Damk_j\right) \leq 1,
\end{equation}
where he superscript stands for $b^\pm = \frac{b\pm|b|}{2}$. 


\subsubsection{Discretization of the ion dynamics system}

The low-temperature plasmas have an additional difficulty in the integration of the ion dynamics. As the plasma potential scales with the electron temperature, the Lorentz force on the ions is in general much larger than the ions pressure gradient. Naive discretizations of the source term may create numerical instabilities (see, e.g., Section~\ref{subsec_WB_test} for the present problem or \cite{Bermudez94, Leveque90a, Gosse13, Leveque_book, Despres05} in a general framework). In order to tackle this problem, we propose an upwind discretization \cite{Bermudez94} of the source term in eqs.~\eqref{IonMass_ND_transport}-\eqref{IonMom_ND_transport},as presented in the following numerical scheme. 

\begin{numScheme}

The discretization of the ion dynamics system \eqref{IonMass_ND_transport}-\eqref{IonMom_ND_transport} reads

\begin{equation}
  {\Uion}_j^{\nPlusOne} = {\Uion}_j^n - \frac{\Delta t}{\Delta x}\left( {\Fion}_{j+1/2}^n -  {\Fion}_{j-1/2}^n  \right) + \Delta t \left({{\Sion}_j^{\phi}}^{\nPlusOneM} + {{\Sion}_j^{\text{Colls}}}^{n} \right)
\end{equation}
where 
\begin{equation}
  {\Uion} = \left( \rhoionb, ~\rhoionb \uib\right)^T,~~{\Fion} = \left( \rhoionb \uib,~\rhoionb \uib^2 + \rhoionb\kappa \right)^T, ~~{\Sion}^{\phi} = \left(0,~-\rhoionb\dxb\phib\right)^T,~\text{and}~~~{\Sion}^{\text{Colls}} = (\rhoeb\Damk,~0)^T
\end{equation}

The discretization of the convective fluxes ${\Fion}$ is performed with an approximate Riemann HLL scheme \cite{HLL83}. The discretization of the collisional source term is at the cell center, i.e.,
\begin{equation}
{\Sion}_j^{\text{Colls}} = \left( (\rhoeb{\Damk})_j^n,~0\right)^T.
\end{equation}
Alternatively, as the Lorentz force is dominant in low pressure discharges, we propose a well-balanced scheme as in \cite{Bermudez94}. We decompose the source term into the contribution at the interfaces. For the sake of simplicity of notation, we do not include the subscript corresponding the time discretization. The discretization reads
\begin{eqnarray}\label{eq:WBSource}
  {\Sion}^{\phi}_j &=& \mathcal{G}\left({\Uion}_{j-1},{\Uion}_j,{\Uion}_{j+1}, \phib_{j-1},\phib_j,\phib_{j+1}\right)\\
   & = & \mathcal{G}^L\left({\Uion}_{j-1},{\Uion}_j,\phib_{j-1},\phib_j\right) + \mathcal{G}^R\left({\Uion}_j,{\Uion}_{j+1},\phib_j,\phib_{j+1}\right). \nonumber
\end{eqnarray}
where the functions $\mathcal{G}^{L/R}$ read
\begin{subequations}
\begin{align}
  \mathcal{G}^{L}({\Uion}_{j-1},{\Uion}_j, \phib_{j-1},\phib_j) &= \frac{1}{2}\left[ \mathbf{I} + \left| \mathbf{A}({\Uion}_{j-1},{\Uion}_j)\right| \mathbf{A}({\Uion}_{j-1},{\Uion}_j)^{-1} \right] \hat{\mathcal{G}}({\Uion}_{j-1},{\Uion}_j,\phib_{j-1},\phib_j), \nonumber \\ \\
    \mathcal{G}^{R}({\Uion}_j,{\Uion}_{j+1}, \phib_j,\phib_{j+1}) &= \frac{1}{2}\left[ \mathbf{I} - \left| \mathbf{A}({\Uion}_{j},{\Uion}_{j+1})\right| \mathbf{A}({\Uion}_{j},{\Uion}_{j+1})^{-1} \right] \hat{\mathcal{G}}({\Uion}_{j},{\Uion}_{j+1},\phib_{j},\phib_{j+1}). \nonumber \\
\end{align}
\end{subequations}
The matrix $\mathbf{A}$ is the Jacobian matrix of the convective flux and $\hat{\mathcal{G}}$ is the source term computed in the interface. They are defined as 
\begin{align}
\mathbf{A}\left({\Uion}_j,{\Uion}_{j+1}\right) &= \left[
\begin{array}{cc}
0 & 1\\
-\uionbjphalf^2 + \soundionb^2 & 2\uionbjphalf
\end{array}
\right]
~~~
\text{and}
~~~\\
\hat{\mathcal{G}}\left({\Uion}_j,{\Uion}_{j+1},\phib_{j},\phib_{j+1}\right) &= \left(
\begin{array}{c}
\hat{g}_{j+1/2}^{mass}\\
\hat{g}_{j+1/2}^{mom}
\end{array}
\right)
=  \left(
\begin{array}{c}
0\\
\bar{n}_{\ion_{j+1/2}}\frac{\phib_{j+1} - \phib_{j}}{\Delta x}
\end{array}
\right).
\end{align}
where 
\begin{equation*}
\uionbjphalf = \frac{\sqrt{\rhoionbj}\uionbj + \sqrt{\rhoionbjpone}\uionbjpone}{\sqrt{\rhoionbj} + \sqrt{\rhoionbjpone}}~~~\text{and}~~~\bar{n}_{\ion_{j+1/2}} = \frac{\rhoionbj +\rhoionbjpone}{2}.
\end{equation*}

In the particular case of isothermal flow with no mass source term, $\mathcal{G}^{L/R}$ simplify to
\begin{subequations}
\begin{align}
  \mathcal{G}^{L}(&{\Uion}_{j-1},{\Uion}_j, \phib_{j-1},\phib_j) = \hat{g}_{j-1/2}^{mom} \\ &\left(
  \begin{array}{c}
  \frac{1}{2\soundionb}\sign(\soundionb + \uionbjmhalf) - \frac{1}{2\soundionb}\sign(\soundionb - \uionbjmhalf)\\
    1 + \frac{\uionbjmhalf}{2\soundionb}\left[\sign(\soundionb + \uionbjmhalf) - \sign(\uionbjmhalf - \soundionb)\right] + \frac{1}{2}\left[\sign(\soundionb + \uionbjmhalf) + \sign(\uionbjmhalf - \soundionb)\right]
  \end{array}\right), \nonumber \\
  \mathcal{G}^{R}(&{\Uion}_j,{\Uion}_{j+1}, \phib_j,\phib_{j+1})= \hat{g}_{j+1/2}^{mom} \\ &\left(
  \begin{array}{c}
  -\frac{1}{2\soundionb}\sign(\soundionb + \uionbjphalf) + \frac{1}{2\soundionb}\sign(\soundionb - \uionbjphalf)\\
    1 - \frac{\uionbjphalf}{2\soundionb}\left[\sign(\soundionb + \uionbjphalf) - \sign(\uionbjphalf - \soundionb)\right] - \frac{1}{2}\left[\sign(\soundionb + \uionbjphalf) + \sign(\uionbjphalf - \soundionb)\right]
  \end{array}\right).\nonumber
\end{align}
\end{subequations}
Note that the velocity of the ions is evaluated at the interface between the cells. As $\sign(\delta u_\ion)$ is a discontinuous function, in numerical experiments, we observe that the $\sign(\delta u_\ion)$ function introduces discontinuities in the solution that are advected by the fluids, as the model has no physical dissipation. For that reason, we use the following approximation for the $\sign(\delta u_\ion)$ function
\begin{equation}\label{eq:signFnct}
  \text{sign}(\delta u_\ion) \approx \tanh\left(\frac{\delta u_\ion}{u_{\ion_\infty}} \right)
\end{equation}
where $\delta u_\ion = \soundionb + \uion$ or $\soundionb - \uion$ and $u_{\ion_\infty}$ is a characteristic speed of the problem. 
\end{numScheme}

\subsection{Summary of the numerical scheme}

In Fig.~\ref{Fig:algorithm}, we summarize the steps followed in order to obtain the numerical results corresponding to the AP scheme. We highlight that the steps in the electron acoustic system need to be done sequentially. The order of steps 4) and 5) can be inverted obtaining the same result.

\begin{figure} [H] 
	\centering
		  \includegraphics[trim=0cm 0cm 0cm 0cm, clip=true,width=0.9\textwidth]{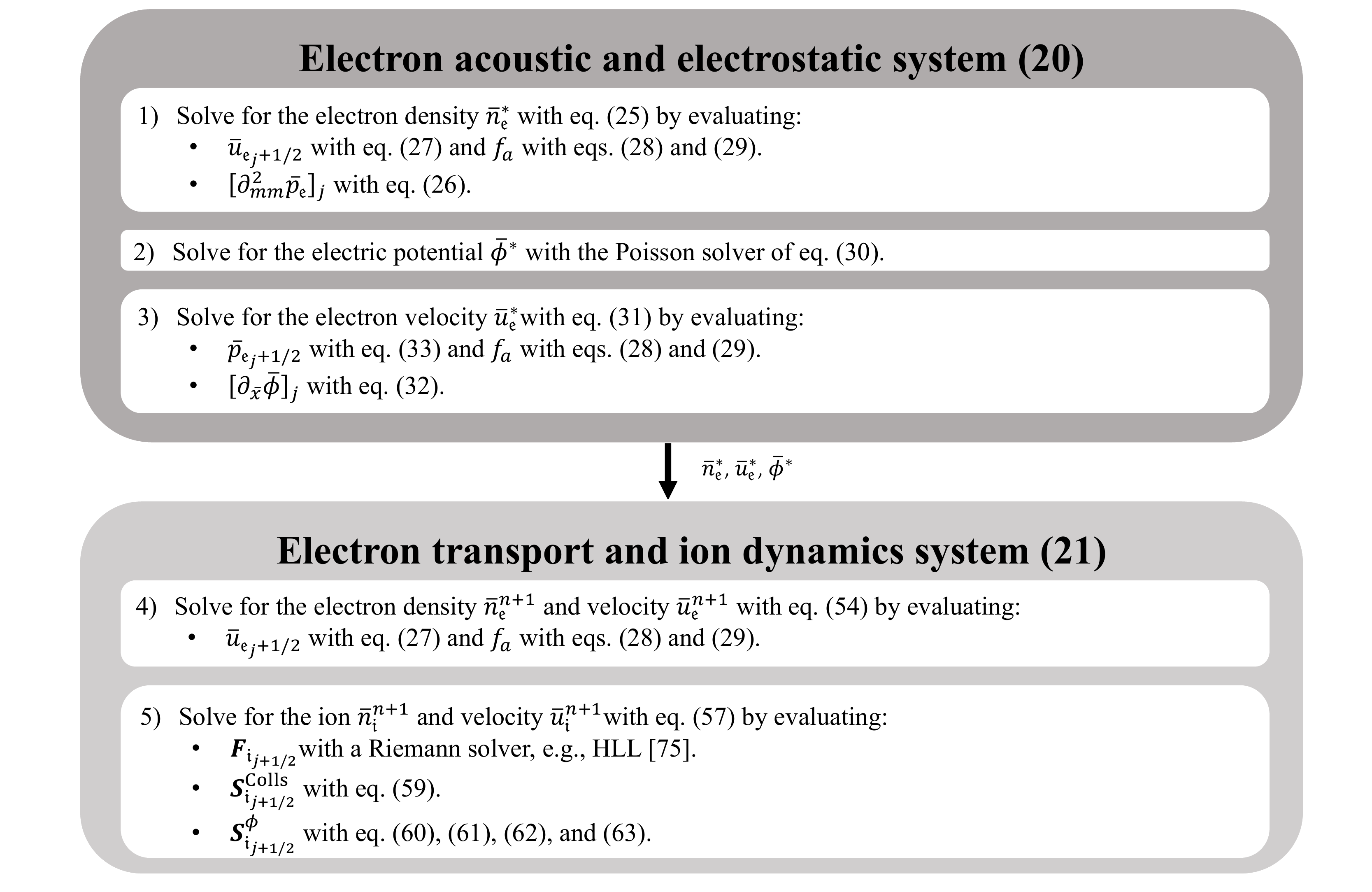} 
	\caption{Algorithm for the implementation of the AP scheme}
	\label{Fig:algorithm}
\end{figure}

\section{Numerical results}
\label{sec:results}

In the present section, we present the numerical results of the proposed scheme. We first simulate a periodic two-stream perturbation in a thermal plasma ($\Te = \Tion$). In this problem we study: (1) the convergence of our numerical scheme with the time step and the mesh size, (2) the asymptotic behaviour and consistency of the scheme when $\lambdaSq\rightarrow0$, and (3) to assess the low-Mach correction (limit when $\varepsilon\rightarrow0$) in the electron fluxes. 

Second, we simulate the same instability in a low-temperature plasma order to study: (1) the performance of the well-balanced scheme in the ions, (2) assess the spatial convergence in a low-temperature plasma, and (3) to evaluate the computational performance as compared to a standard discretization. 

Finally, we propose a numerical set-up to simulate a plasma sheath between two floating conducting walls. In this case, we simulate a physical domain that contains charge separation regions. We show that our numerical set-up captures the physics of the classical plasma sheath \cite{Riemann05}, computing the potential drop as predicted by the theory. As the numerical domain contains both a quasi-neutral and charged regions, in this test we assess the ability of the scheme to represent both limits.

\subsection{Propagation of a two-stream periodic perturbation in a collisionless thermal plasma}

The two-stream instability occurs in a uniform plasma of density $n_0$ where the electrons have a relative velocity $u_0$ with respect to the ions \cite{Chen84}. Perturbations of the equilibrium state can lead to instabilities for sufficiently long wavelengths $k$. In our case, we will restrict our work to study the propagation of a periodic perturbation in a stable plasma. This case has been previously used for the convergence study of the two-fluid model in a number of works \cite{Crispel05,Crispel04,Crispel07}. Our set of equations is slightly different to these previous works as we study the isothermal case whereas in Crispel et al.~\cite{Crispel05,Crispel04,Crispel07} the fluid model assumes an isentropic law. As in \cite{Chabert11}, in our case, the dispersion relation is found to be
\begin{equation}
  F(k, \omega) = \alpha_4\omega^4 + \alpha_3\omega^3 + \alpha_2\omega^2 + \alpha_1\omega + \alpha_0 = 0 \label{dispersionRelation}
\end{equation}
where the coefficients $\alpha_i(k;\varepsilon,\kappa, n_0, v_0, \lambdaSq)$ read
\begin{subequations}
\begin{eqnarray}
  \alpha_4 = \varepsilon\lambdaSq, ~~~\alpha_3 = -2\varepsilon\lambdaSq k u_0,&~~~&\alpha_2 = \left(\varepsilon k^{2} \lambdaSq \left( v_{0}^{2} - \kappa \right) - \varepsilon n_{0} - k^{2} \lambdaSq - n_{0}\right)\\
  \alpha_1 = 2 \left(k^{2} \kappa \lambdaSq + n_{0}\right), ~~~\text{and}~~~ &\alpha_0& = -\varepsilon k^2 v_{0} n_{0} v_{0} - \varepsilon \lambdaSq \left( k^{4} \kappa v_{0}^{2} + 2 k \omega^{3} v_{0}\right) + k^{4} \kappa \lambdaSq + k^{2} n_{0} \left(\kappa + 1\right).
\end{eqnarray}
\end{subequations}
Firstly, we consider a perturbation in a thermal plasma, i.e., where the ions have the same temperature as the electrons. We consider a mass ratio of $\varepsilon = 10^{-4}$ and a domain of length $L = 1$ with $10^4$ Debye lengths, i.e., $\lambda = 10^{-4}$. The background plasma has a normalized density of $n_0 = 1$ and the electrons travel at a velocity of $u_0 = 1$ that corresponds to Bohm's speed at these conditions. We fix the wavelength of the perturbation to $k = 2\pi$. Under these conditions, the characteristic polynomial, eq.~\eqref{dispersionRelation}, has four real solutions, i.e., the perturbation is stable. There are two high frequency solutions, which correspond to the electron plasma waves, and two low frequency that are quasi-neutral. For our study, we use one of the low frequency solutions, i.e., $\omega = 8.8857268$. 

The initial field is the analytical solution for a small perturbation at the given frequency, wavelength, and plasma conditions. In the present case, this solution is
\begin{equation}
\Ub(\xb,t=0) =\left(\begin{array}{c}
\rhoeb\\
\ueb\\
\rhoionb\\
\uionb\\
\phib
\end{array}\right)
= 
\left(\begin{array}{c}
1 + 2.41425\times10^{-2}\sin(2\pi \xb)\\
1 + 10^{-2}\sin(2\pi \xb)\\
1 + 2.41425\times10^{-2}\sin(2\pi \xb)\\
3.41425\times10^{-2}\sin(2\pi \xb)\\
2.41421\times10^{-2}\sin(2\pi \xb)
\end{array}\right).
\end{equation}

In Fig.~\ref{1_NoLowMachCorrection}, we present the solution of the Lagrange-projection AP scheme without the low-Mach correction for the electrons at $t=0.7071$ which corresponds to one period of the wave. We present the results with four different mesh resolutions $N = 50$, $100$, $200$, and $400$ points. We highlight that the mesh size is not adapted to the Debye length as there are $200$ Debye lengths inside one cell in the case of lowest resolution. The simulations use a CFL$^{conv}=0.7$, which corresponds to a time step of $\Delta t = 1.4141\times10^{-4}$, $0.7071\times10^{-4}$, $0.3535\times10^{-4}$, and $0.176775\times10^{-4}$. We also notice that this time step is much larger than the period of a plasma wave, i.e., $\omega_{p\elec}^{-1} = 10^{-6}$. A standard simulation that resolved the Debye length with 10 mesh points and the plasma frequency with $10$ time-steps, would have $2000$ more mesh points and a time step $10^3$ times smaller.

\begin{figure} [h] 
	\centering
		  \includegraphics[trim=0cm 0cm 0cm 0cm, clip=true,width=0.495\textwidth]{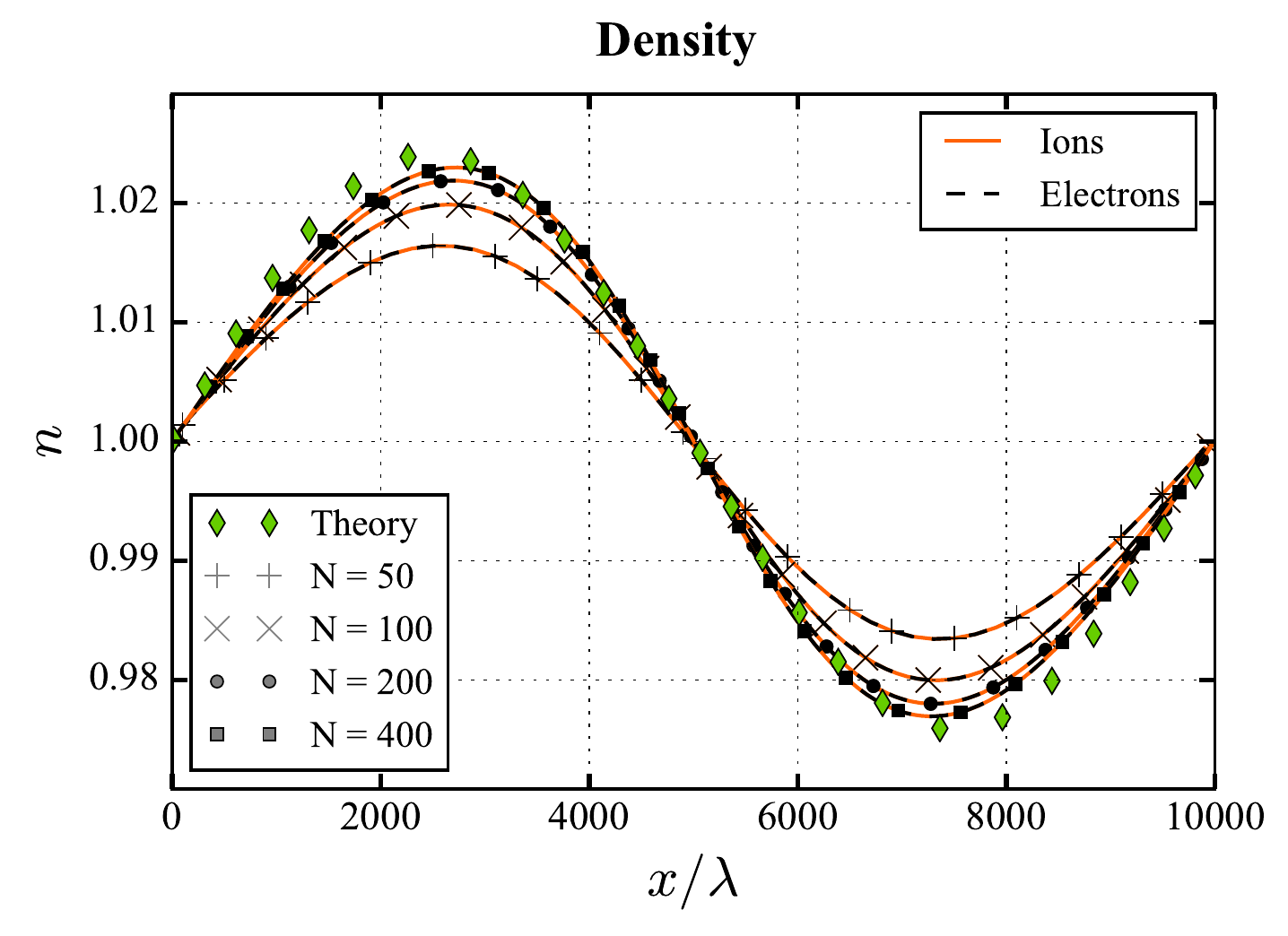} \hfill
		  \includegraphics[trim=0cm 0cm 0cm 0cm, clip=true,width=0.495\textwidth]{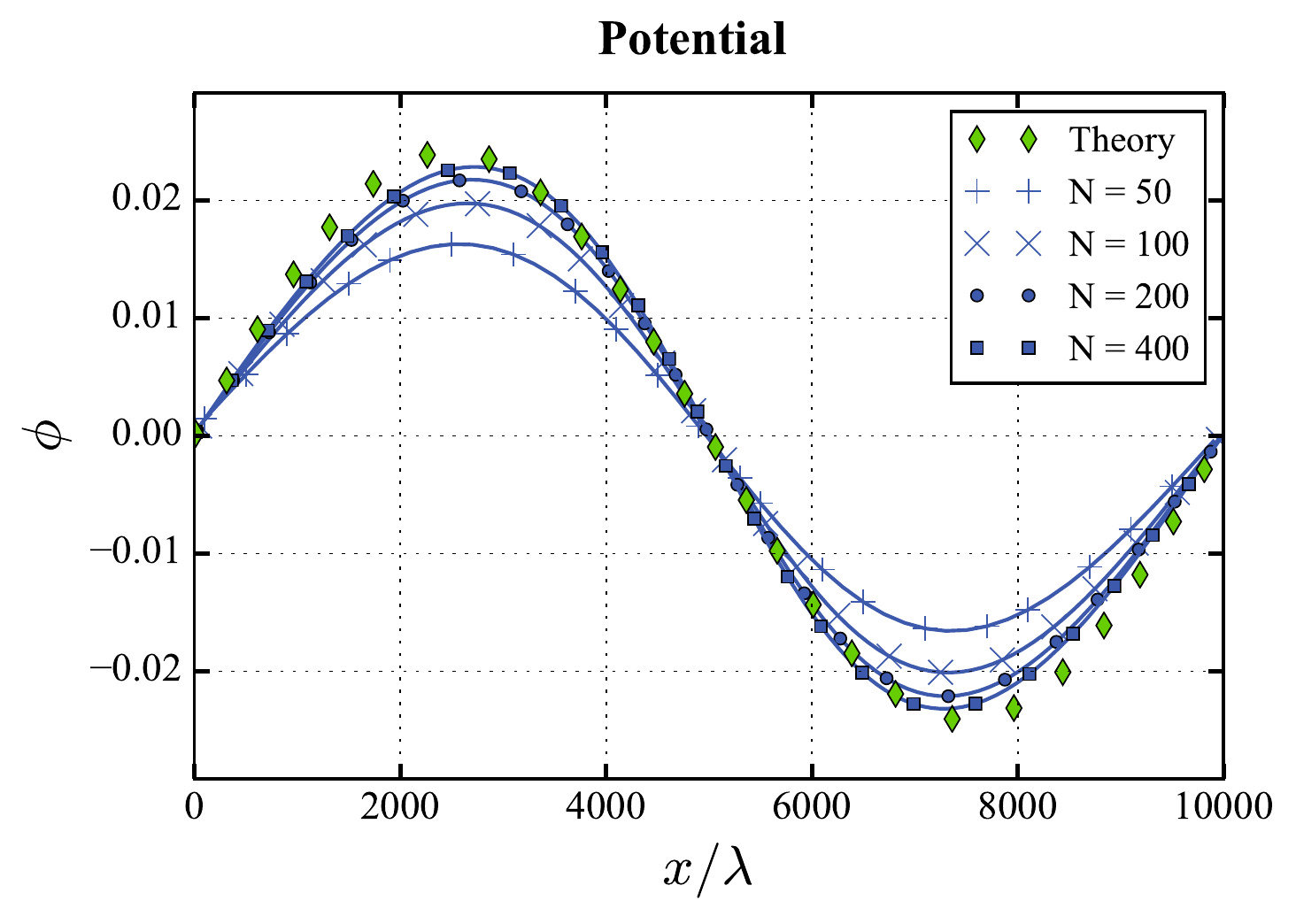}
		  \includegraphics[trim=0cm 0cm 0cm 0cm, clip=true,width=0.495\textwidth]{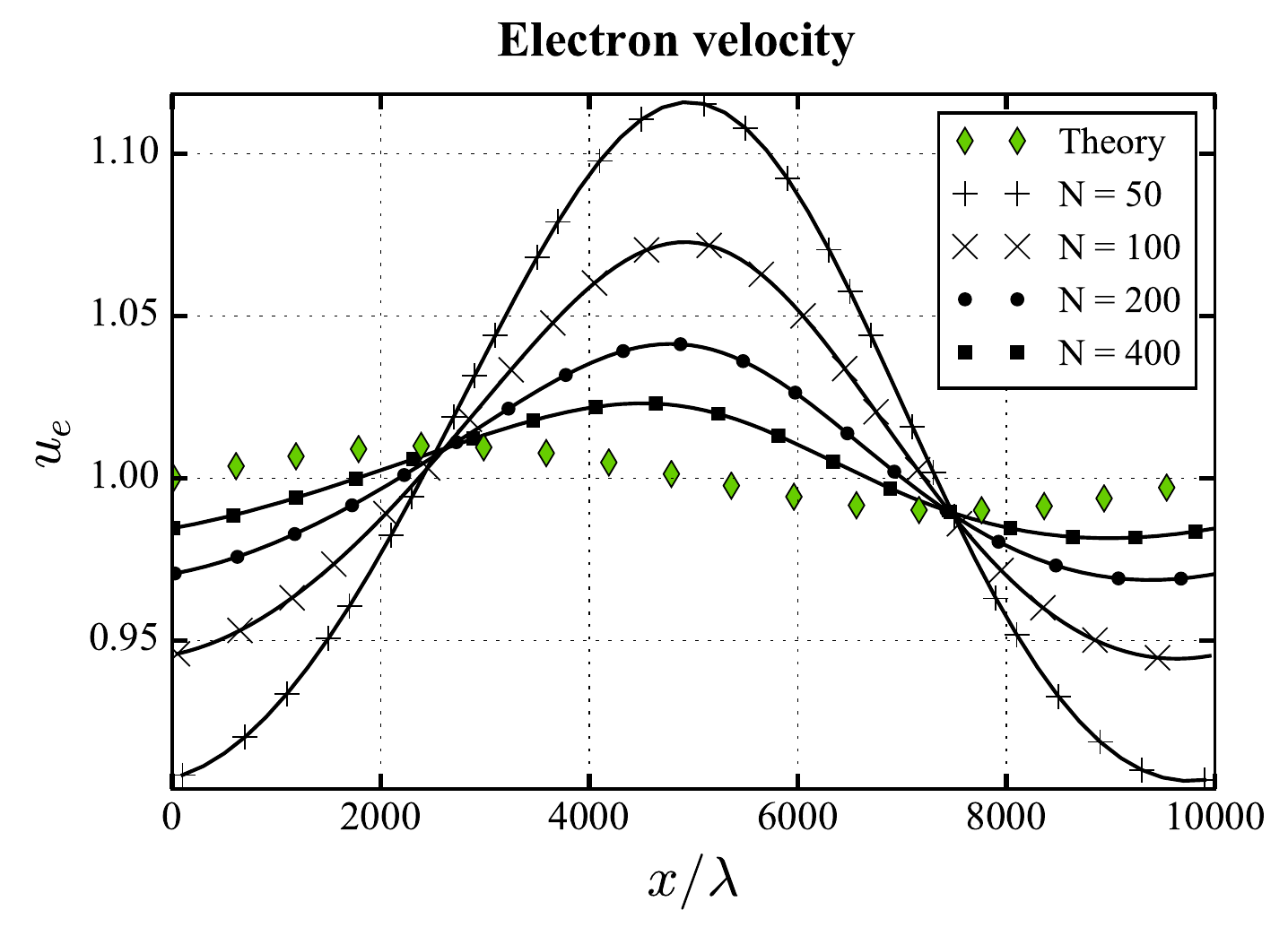}\hfill
		  \includegraphics[trim=0cm 0cm 0cm 0cm, clip=true,width=0.495\textwidth]{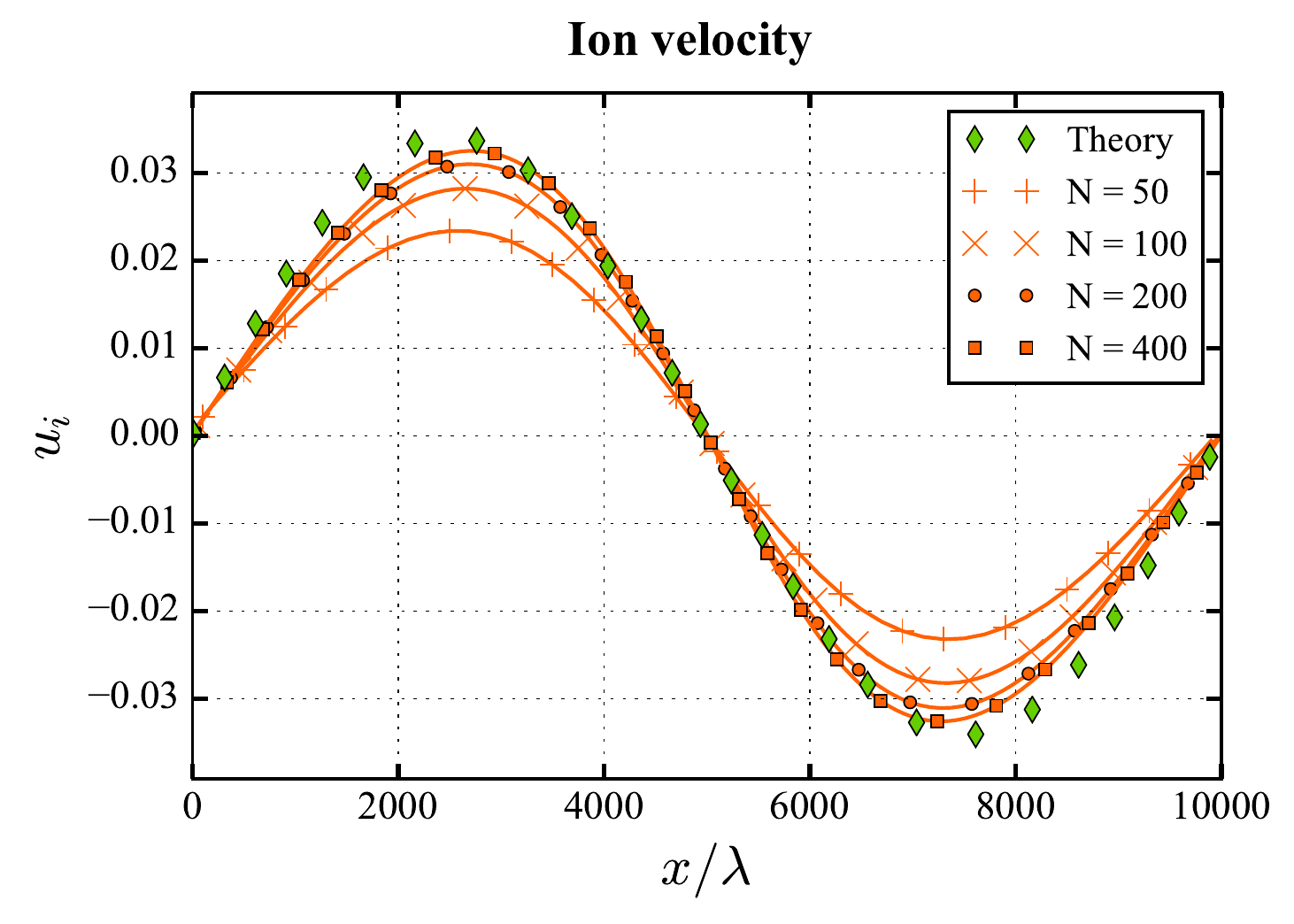} 
	\caption{Solution of the two-stream periodic perturbation in a collisionless thermal plasma at $t = 0.7071$ for different mesh resolutions with the asymptotic preserving scheme without the low-Mach correction. The simulations use CFL$^{conv}=0.7$, corresponding to a $\Delta t = 1.4142\times10^{-4},~0.7071\times10^{-4},~0.3535\times10^{-4},$ and $0.176775\times10^{-4}$. Note that $\Delta t > \omega_{pe}^{-1} =  10^{-6}$. The absence of correction for the low-Mach regime of the electrons results in a large error on the electron velocity. }
	\label{1_NoLowMachCorrection}
\end{figure}

In Fig.~\ref{2_NoAP}, we present the solution for $\rhoeb$, $\rhoionb$, $\ueb$, and $\phi$ at $t= 0.0626$, i.e., after $900$ iterations of the standard HLL scheme described in Sec.~\ref{sec:standardDiscretization} that uses $100$ mesh points and a CFL$^{conv}=0.7$. These results illustrate that indeed, using the same conditions than in the Lagrange-projection AP scheme with $N=100$, the standard discretization is unstable. It is also important to note that the scheme becomes unstable very fastly due to the rapid unstable response of the electrons.

\begin{figure}[h]
	\centering
		  \includegraphics[trim=0cm 0cm 0cm 0cm, clip=true,width=0.325\textwidth]{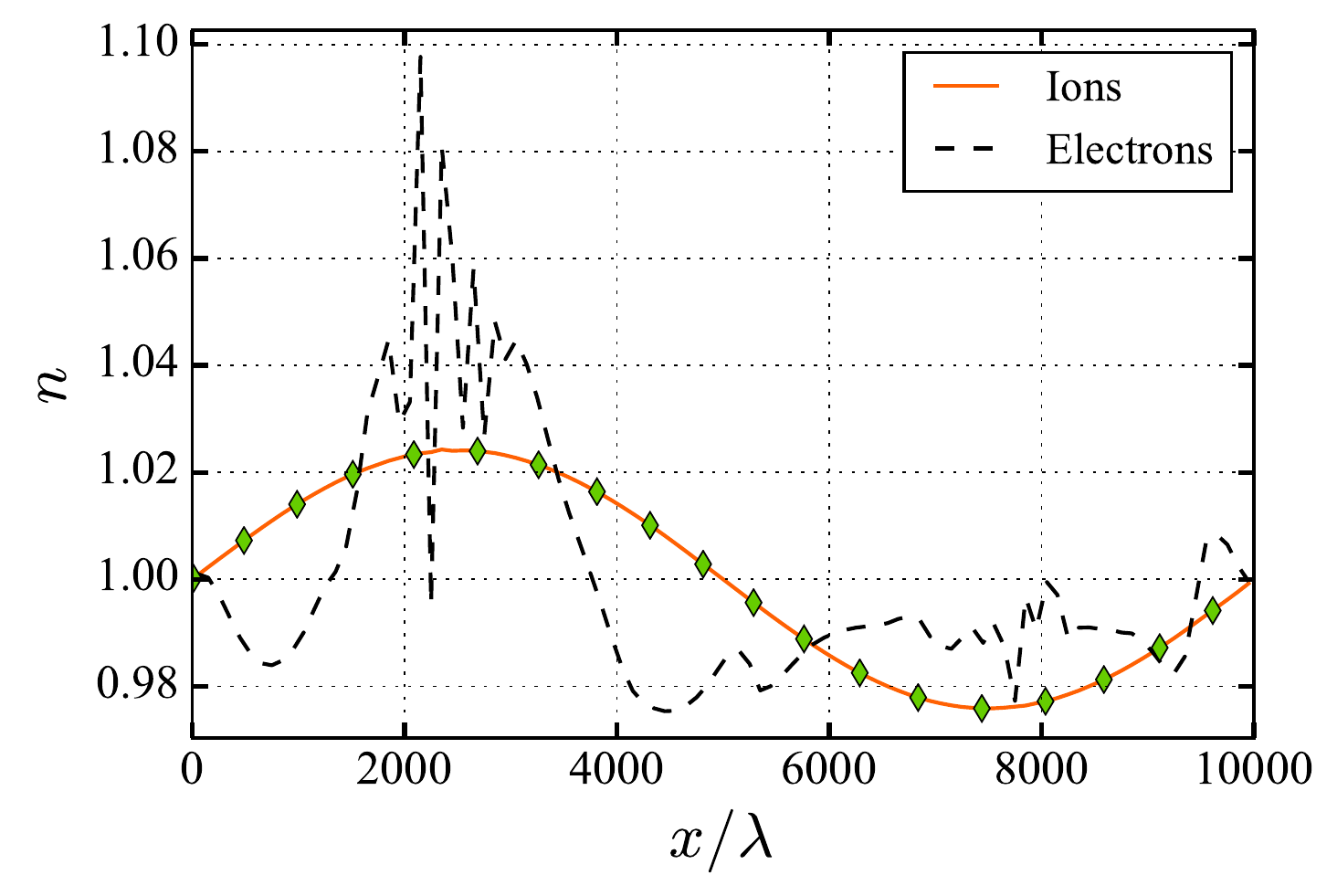} \hfill
		  \includegraphics[trim=0cm 0cm 0cm 0cm, clip=true,width=0.325\textwidth]{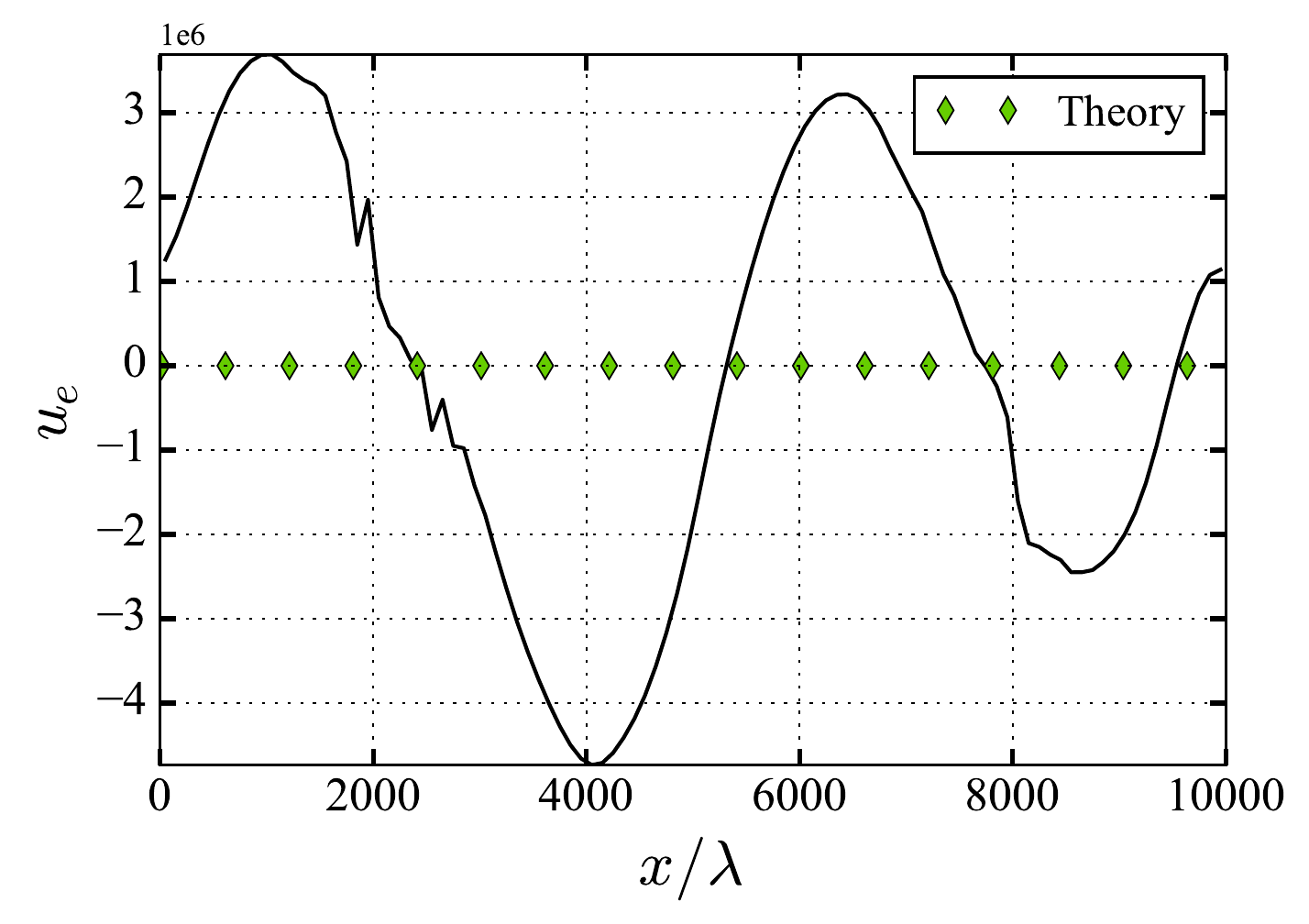}
		  \includegraphics[trim=0cm 0cm 0cm 0cm, clip=true,width=0.325\textwidth]{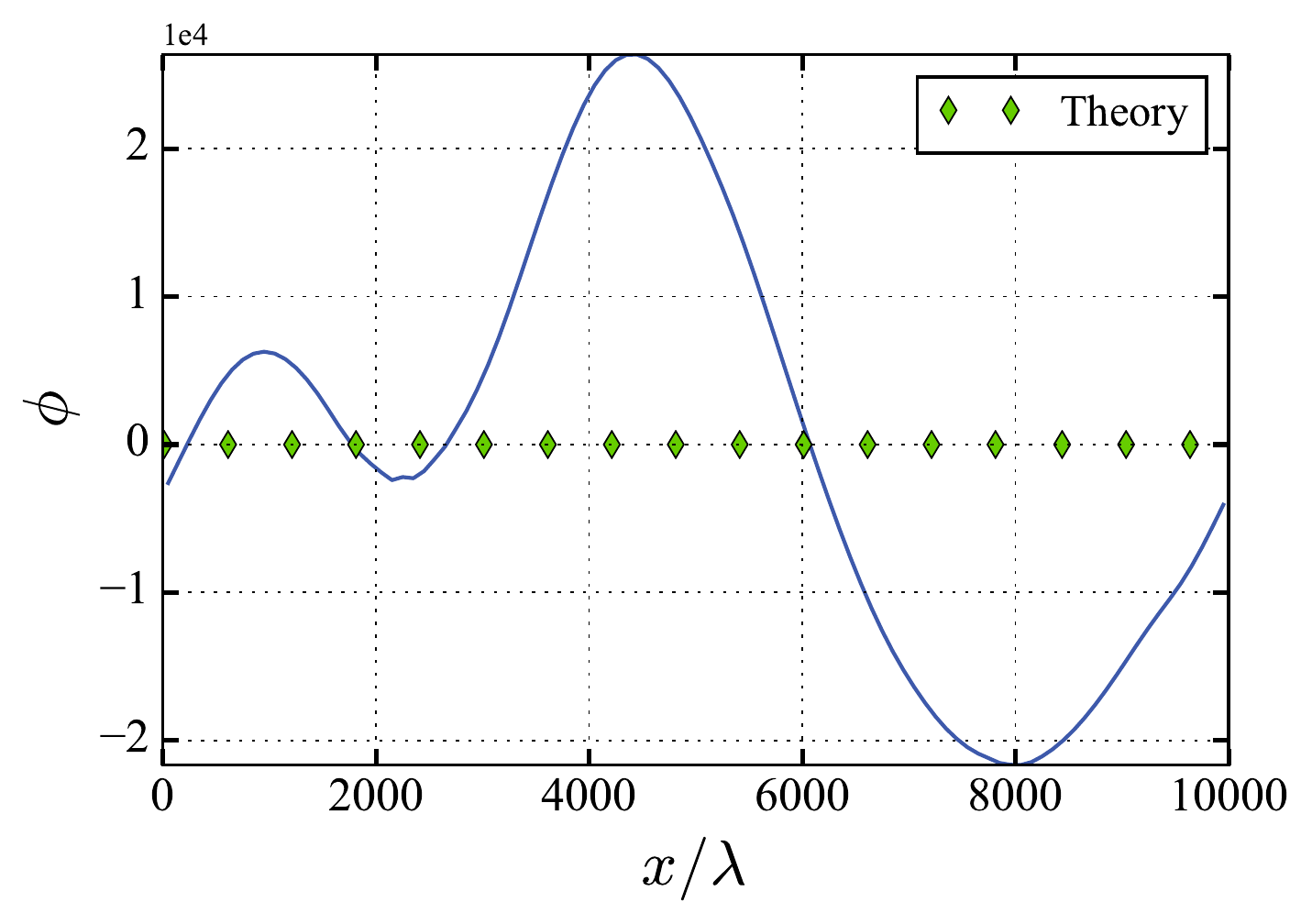} \hfill
	\caption{Solution of the two-stream periodic perturbation in a collisionless thermal plasma at $t = 0.0636$ with a standard HLL scheme discretization on $100$ cells with CFL$^{conv}=0.7$. Under these conditions the discretization is unstable and diverges very rapidly due to the fast dynamics of the electrons.}
	\label{2_NoAP}
\end{figure}

The same simulation is presented with using the Lagrange-projection AP scheme with the low-Mach correction for the electrons in Fig.~\ref{3_APMachCorrection}. We show the previous four discretization resolutions with CFL$^{conv}=0.7$. With the low-Mach correction, the solution for the electron velocity converges to the analytical solution when the mesh size is decreased, showing a much smaller error. 

\begin{figure} [!htb] 
	\centering
		  \includegraphics[trim=0cm 0cm 0cm 0cm, clip=true,width=0.495\textwidth]{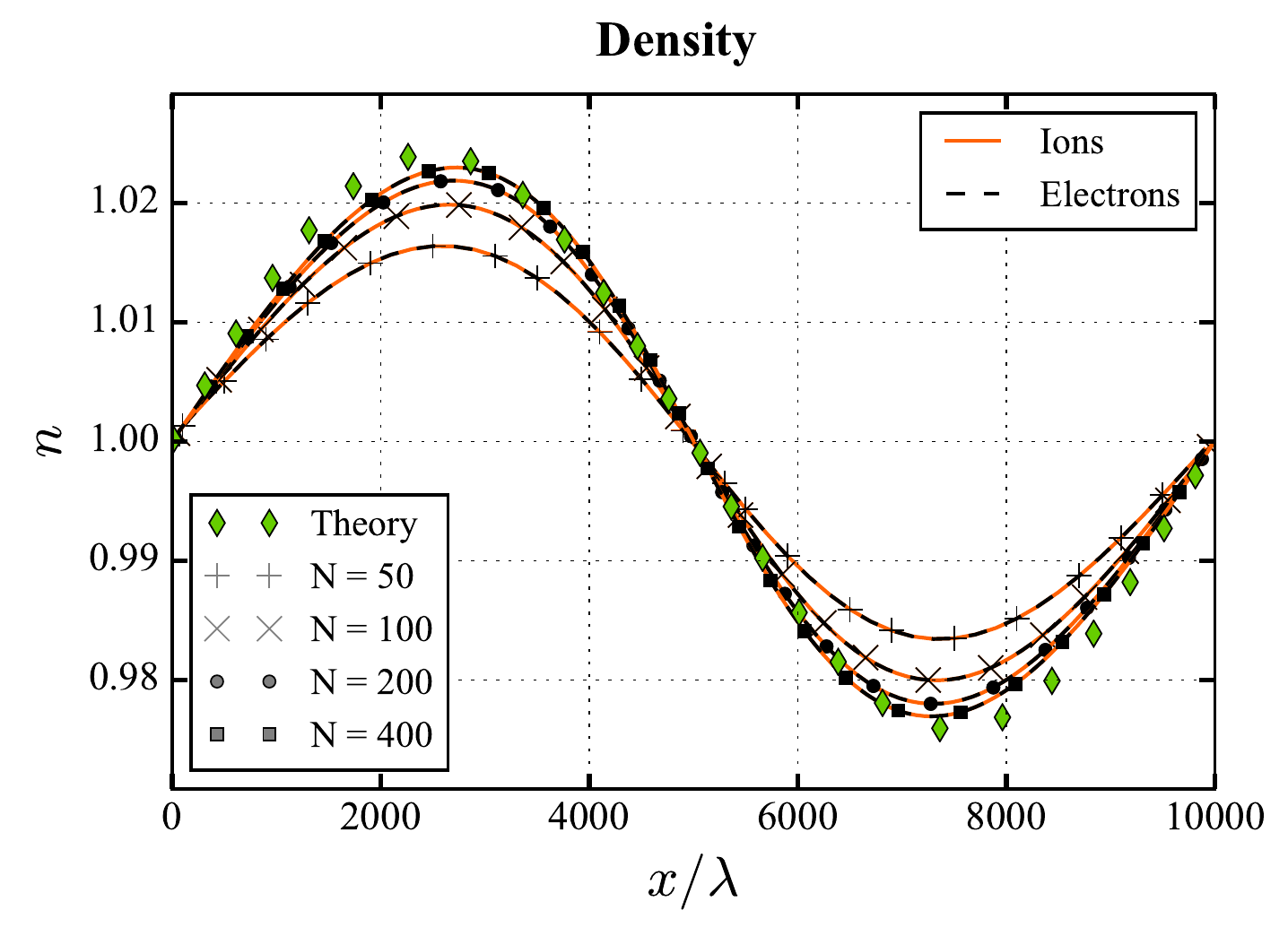} \hfill
		  \includegraphics[trim=0cm 0cm 0cm 0cm, clip=true,width=0.495\textwidth]{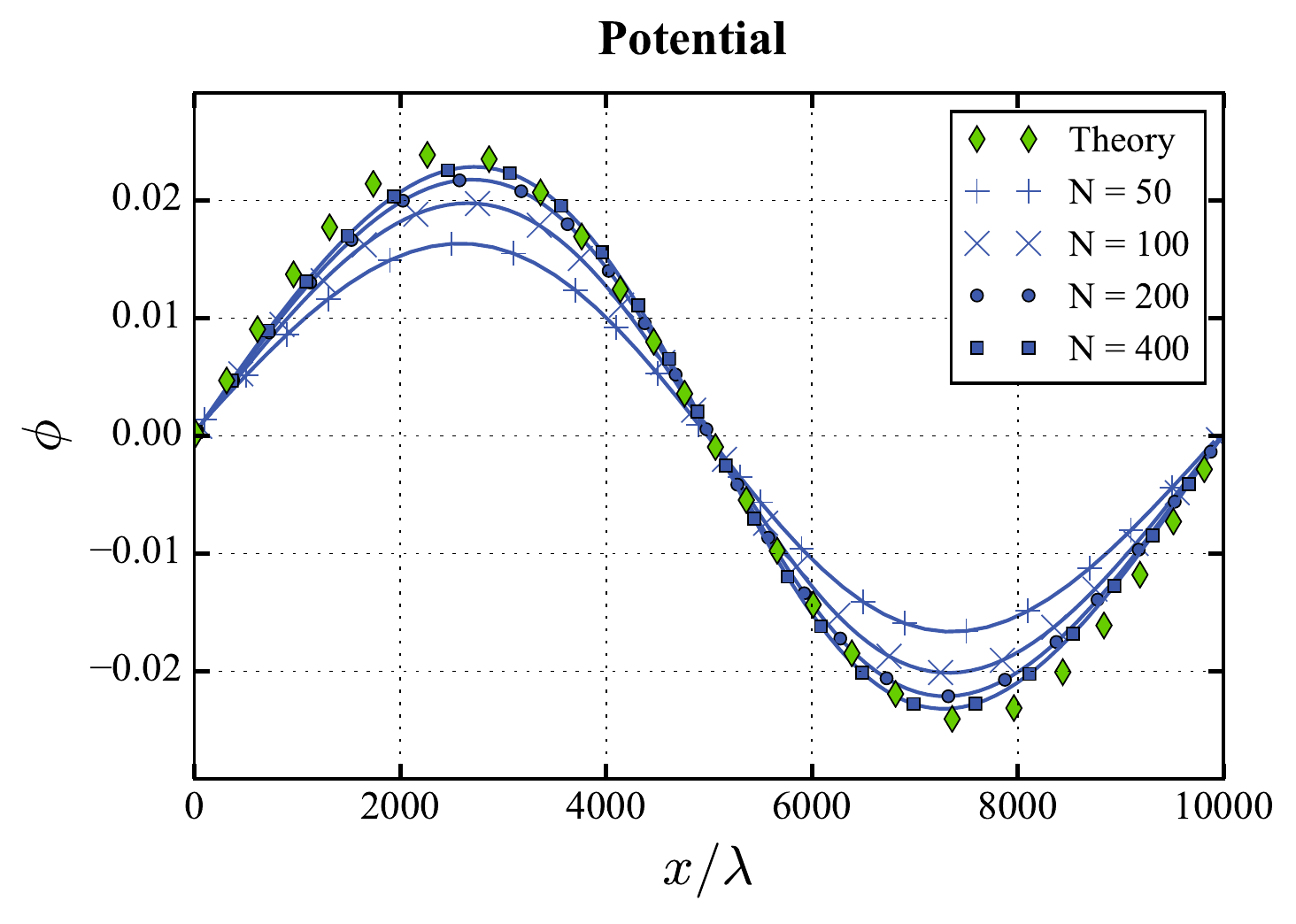}
		  \includegraphics[trim=0cm 0cm 0cm 0cm, clip=true,width=0.495\textwidth]{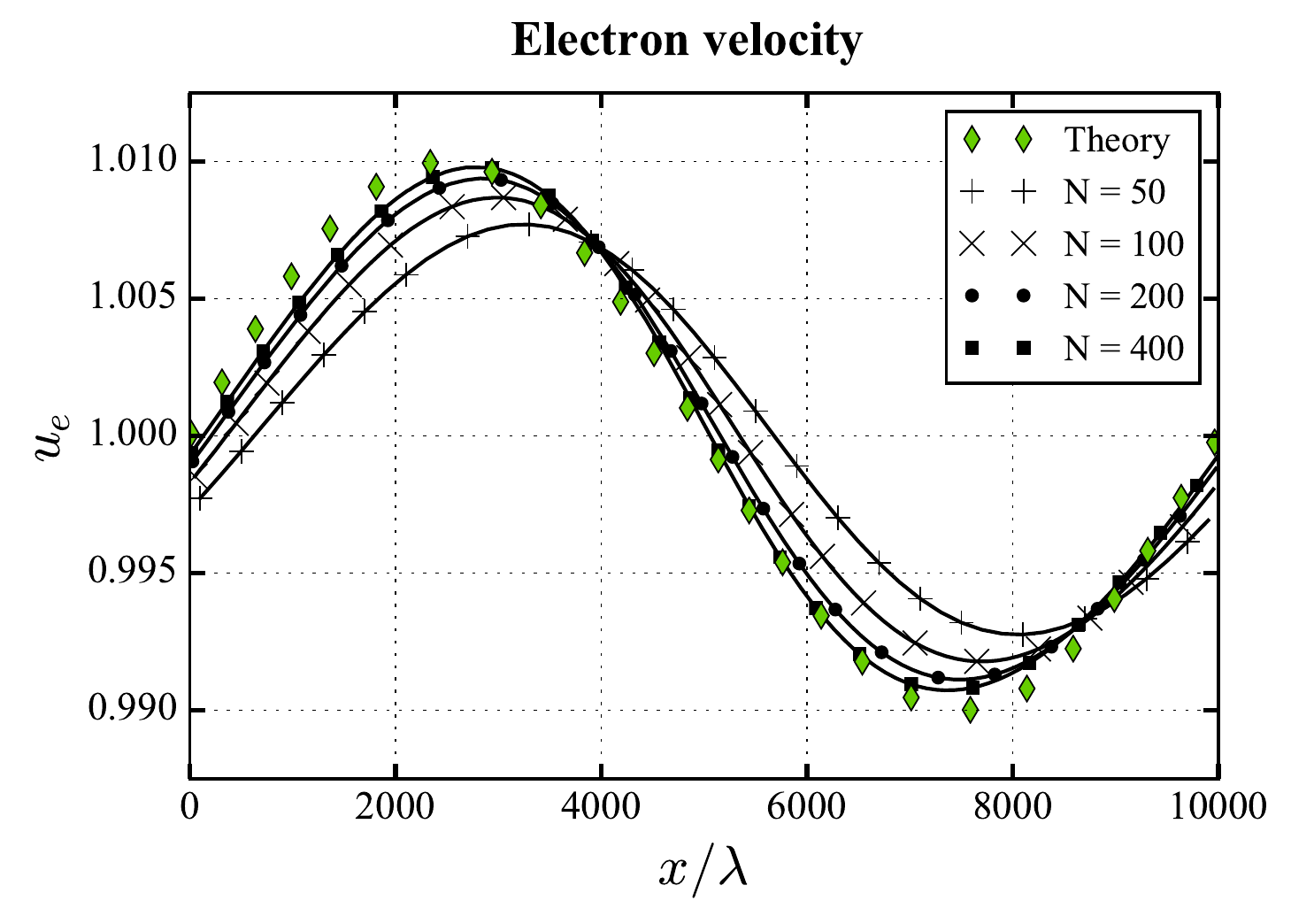}\hfill
		  \includegraphics[trim=0cm 0cm 0cm 0cm, clip=true,width=0.495\textwidth]{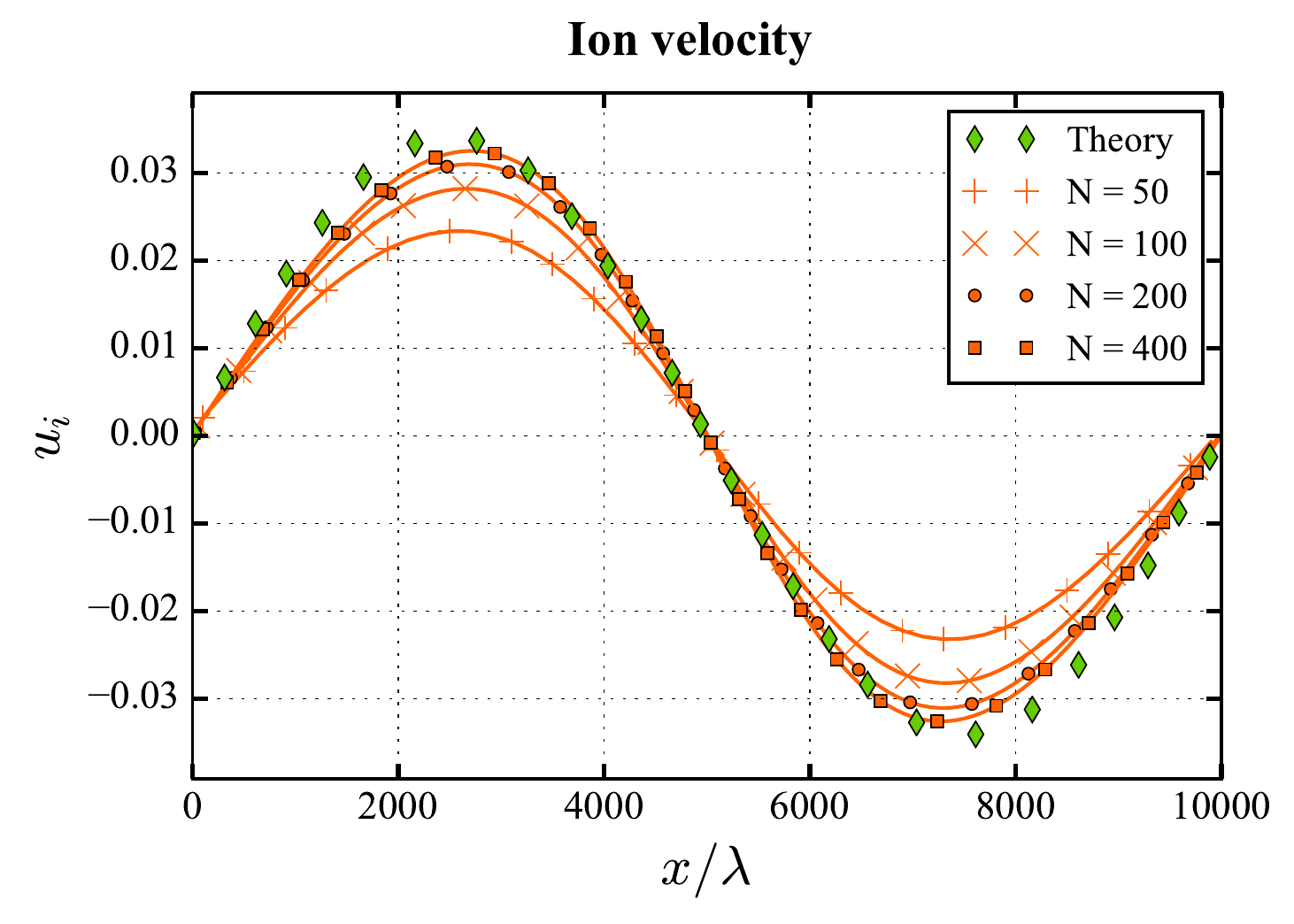} 
	\caption{Solution of the two-stream periodic perturbation in a collisionless thermal plasma at $t = 0.7071$ for different mesh resolutions with the asymptotic preserving scheme with low-Mach correction. The simulations use CFL$^{conv}=0.7$, corresponding to a $\Delta t = 1.4142\times10^{-4},~0.7071\times10^{-4},~0.3535\times10^{-4},$ and $0.176775\times10^{-4}$. Note that $\Delta t > \omega_{pe}^{-1} =  10^{-6}$. The preconditioning of the numerical dissipation for the electrons results in a correct discretization of the electron velocity. }
	\label{3_APMachCorrection}
\end{figure}

As explained in \cite{VKILS99}, when the numerical diffusion is rescaled for the low-Mach regime, the CFL condition is not restricted by the acoustic waves. This allows for the adaptation of the CFL to the eigenvalue $\lambda^\elec = \ueb$ rather than to $\lambda^\elec = \ueb \pm \bar c_\elec$. In Fig.~\ref{4_CFLcomparison}, we present the comparison of the solution for CFL$^{conv}=0.7,~7.1,$ and $28.5$ with the low-Mach correction and $100$ cells. We highlight that the largest CFL corresponds to an ion CFL$_\ion^{conv}=0.28$. The results show that the numerical diffusion is lower when the CFL is adapted to the ion scales. In the simulation with the largest CFL, the time step is $\Delta t = 1.99\times10^{4}\omega^{-1}_{p\elec}$, i.e., $10^4$ larger than an explicit non-AP discretization.


\begin{figure} [!htb] 
	\centering
		  \includegraphics[trim=0cm 0cm 0cm 0cm, clip=true,width=0.495\textwidth]{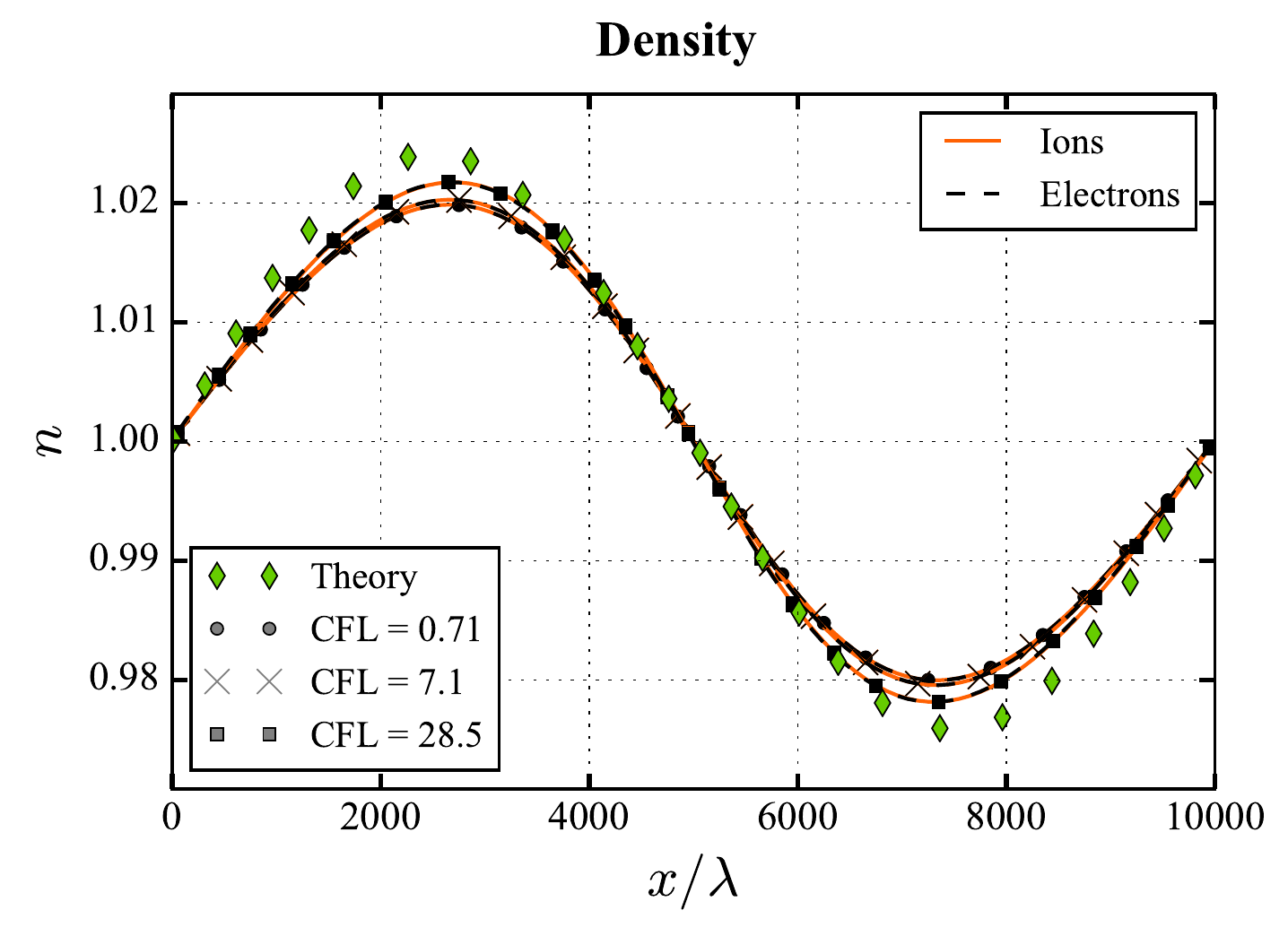} \hfill
		  \includegraphics[trim=0cm 0cm 0cm 0cm, clip=true,width=0.495\textwidth]{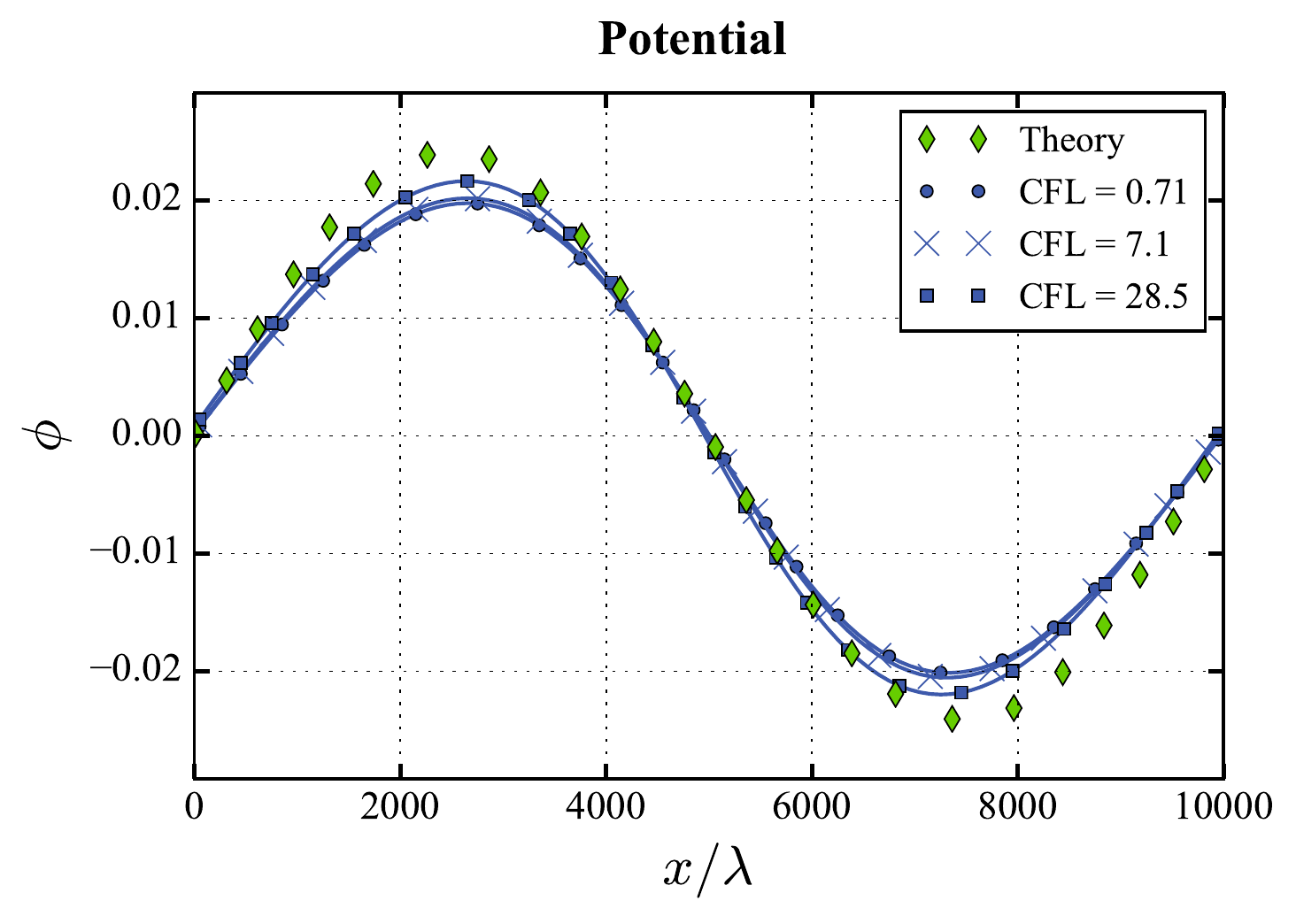}
		  \includegraphics[trim=0cm 0cm 0cm 0cm, clip=true,width=0.495\textwidth]{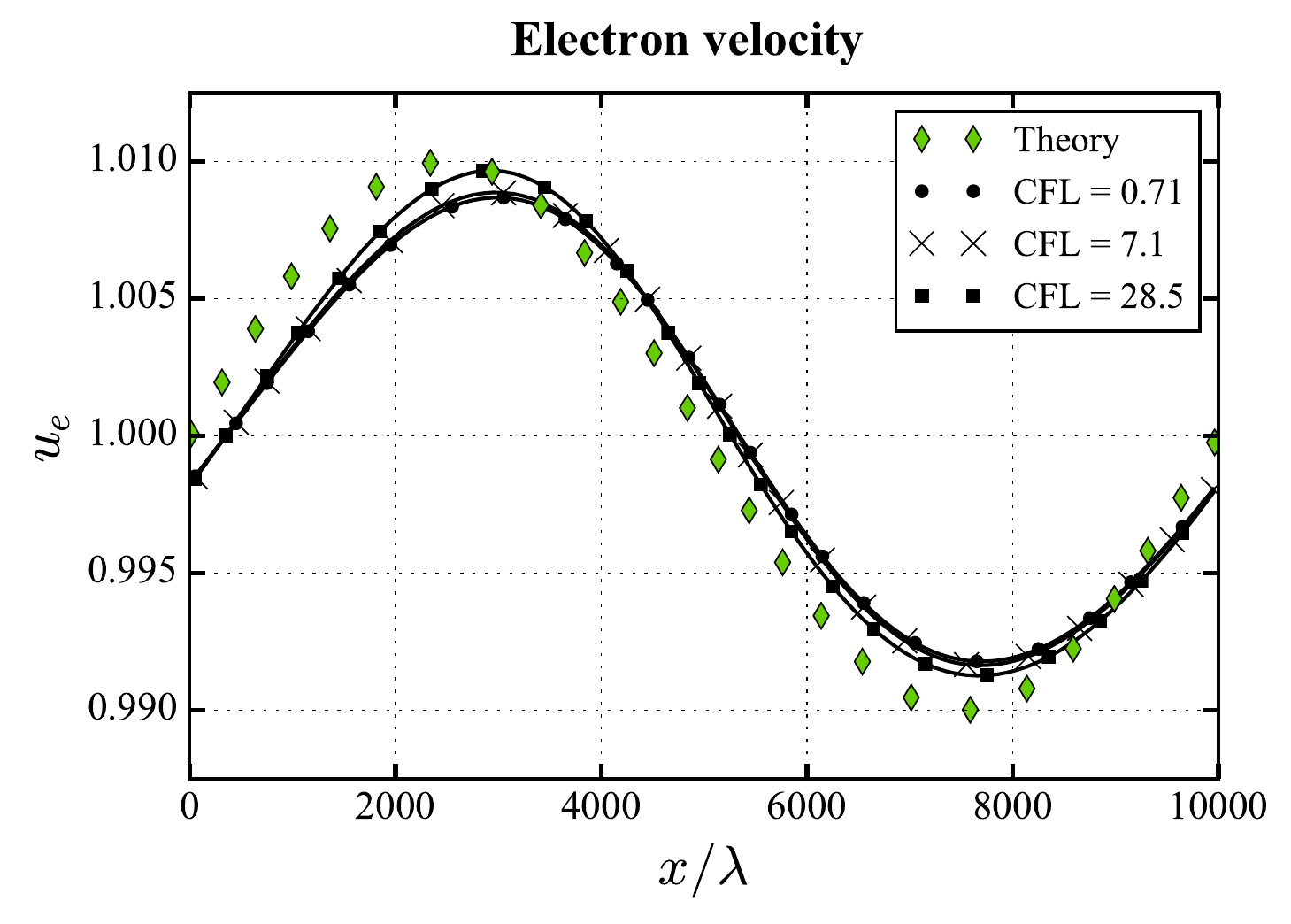} \hfill
		  \includegraphics[trim=0cm 0cm 0cm 0cm, clip=true,width=0.495\textwidth]{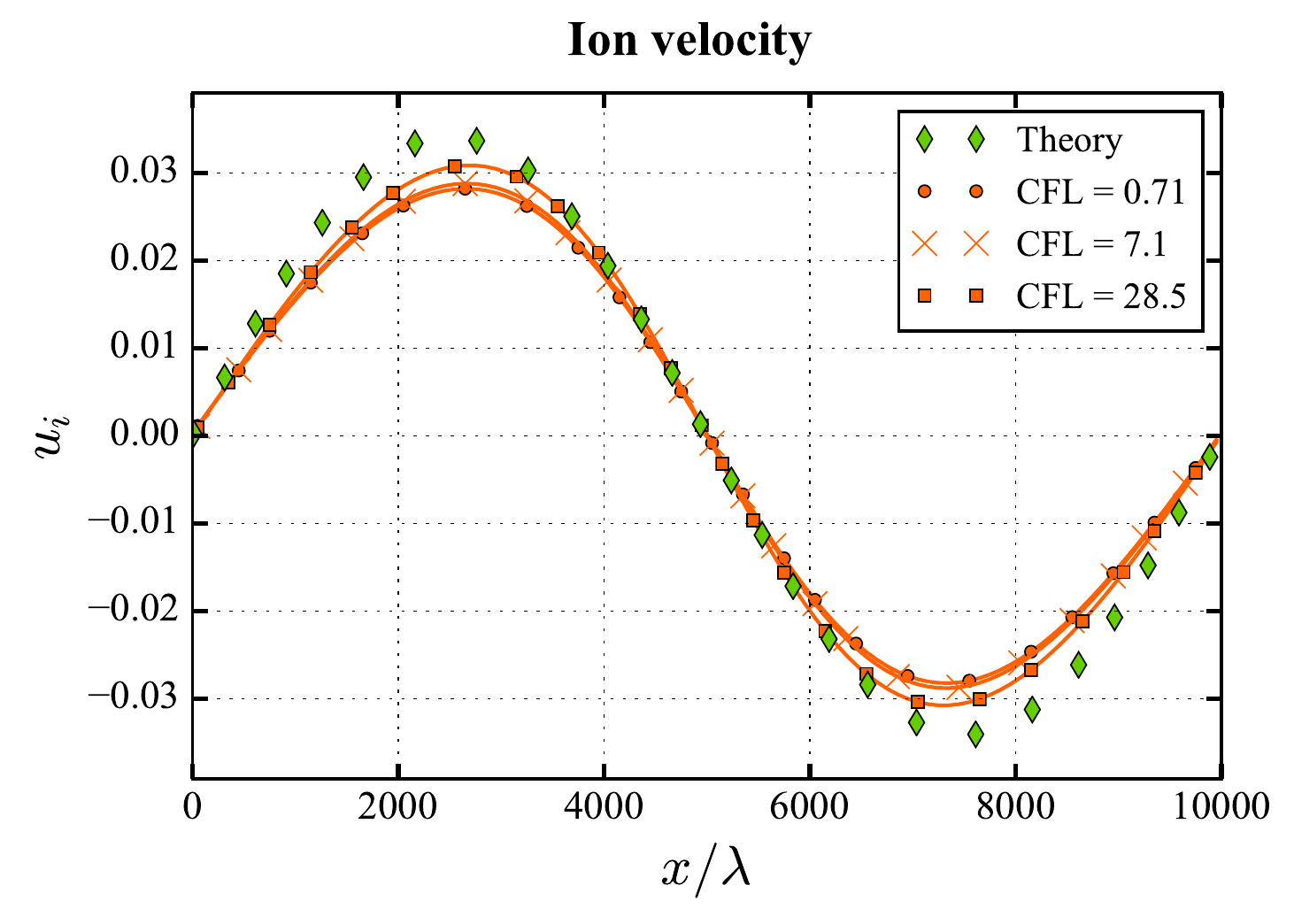}
	\caption{Solution of the two-stream periodic perturbation in a collisionless thermal plasma at $t = 0.7071$ with $N = 100$ points with the asymptotic preserving scheme with low-Mach correction and different time resolution. The time splitting strategy with low-Mach correction allows for a larger time step that does not need to resolve the electron sound waves without an implicit discretization. }
	\label{4_CFLcomparison}
\end{figure}

In Fig.~\ref{5_ErrorConvergence}, the convergence of the $L_2$ error norm is presented for three different cases: the AP scheme without low-Mach correction for the electrons and CFL$^{conv}= 0.7$ (left), the AP scheme with low-Mach correction for the electrons and CFL$^{conv}= 0.7$ (center), and the AP scheme with low-Mach correction for the electrons and CFL$^{conv}= 28.5$ (right). We clearly see that the low-Mach correction of the electrons reduces the error of the electron velocity by one order of magnitude. When the CFL$^{conv}= 0.7$, the convergence of the error corresponds to a first order scheme. In the case of  CFL$^{conv}= 28.5$, the error norm is reduced, especially in the ion quantities. However, we observe that the error convergence is slower than a first-order scheme. This is caused by the low-Mach regime of the ions. Since the perturbation of the ion velocity is small, the Mach of the ions varies from $M_\ion = 0$ to $0.03$. This could be fixed with a low-Mach correction, similar to the one performed in the electrons. Nevertheless, the interest of this paper is low-temperature plasmas, which means that in general, the velocity of the ions is much larger than the ion thermal speed. For that reason, the low-Mach regime of the ions is not treated in the present paper. In the following example, in a low-temperature plasma, we will show how this convergence improves when the Mach of ions is larger.

\begin{figure} [!htb] 
	\centering
	\includegraphics[trim=0cm 0cm 0cm 0cm, clip=true,width=0.32\textwidth]{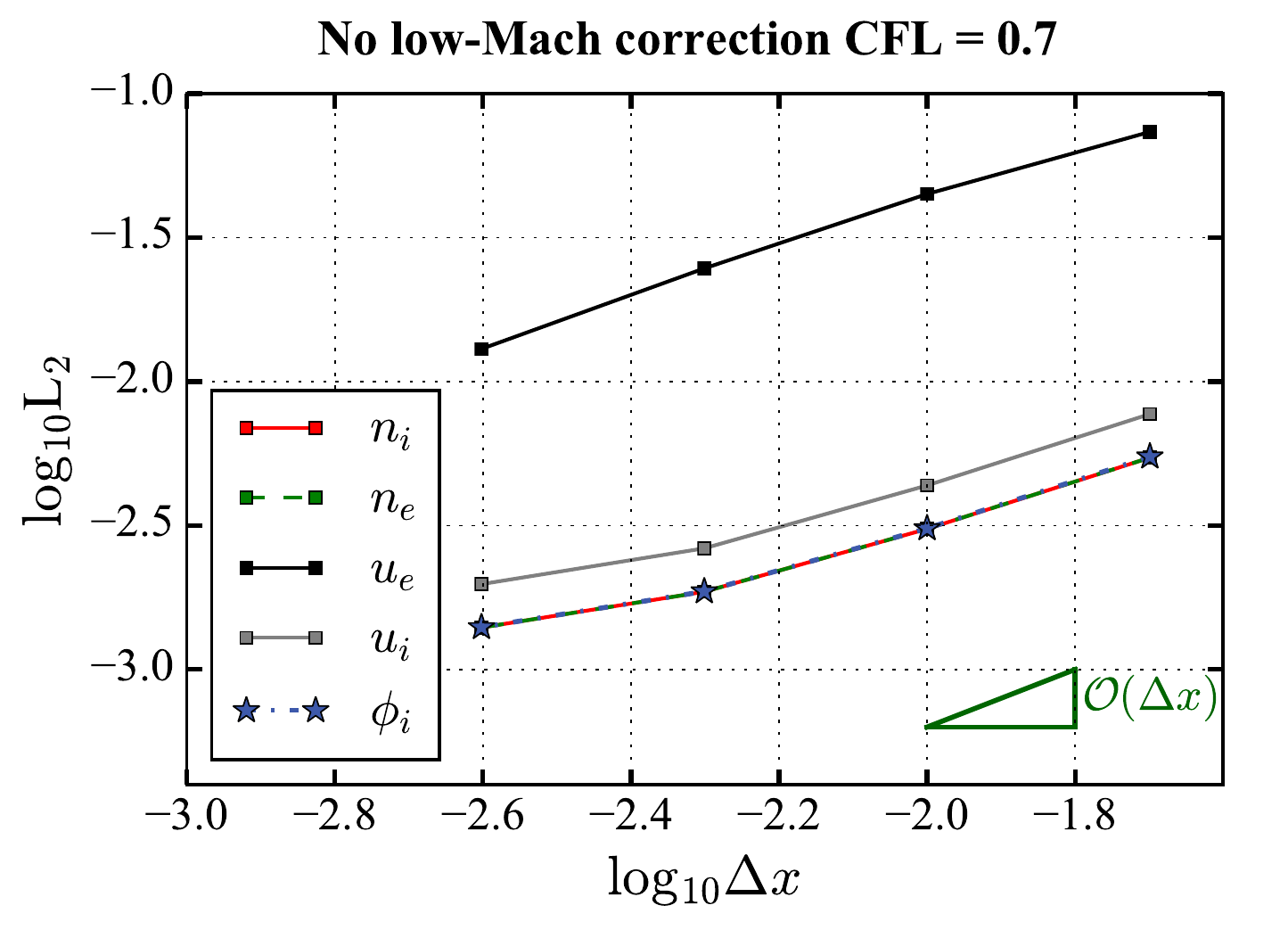} \hfill
	\includegraphics[trim=0cm 0cm 0cm 0cm, clip=true,width=0.32\textwidth]{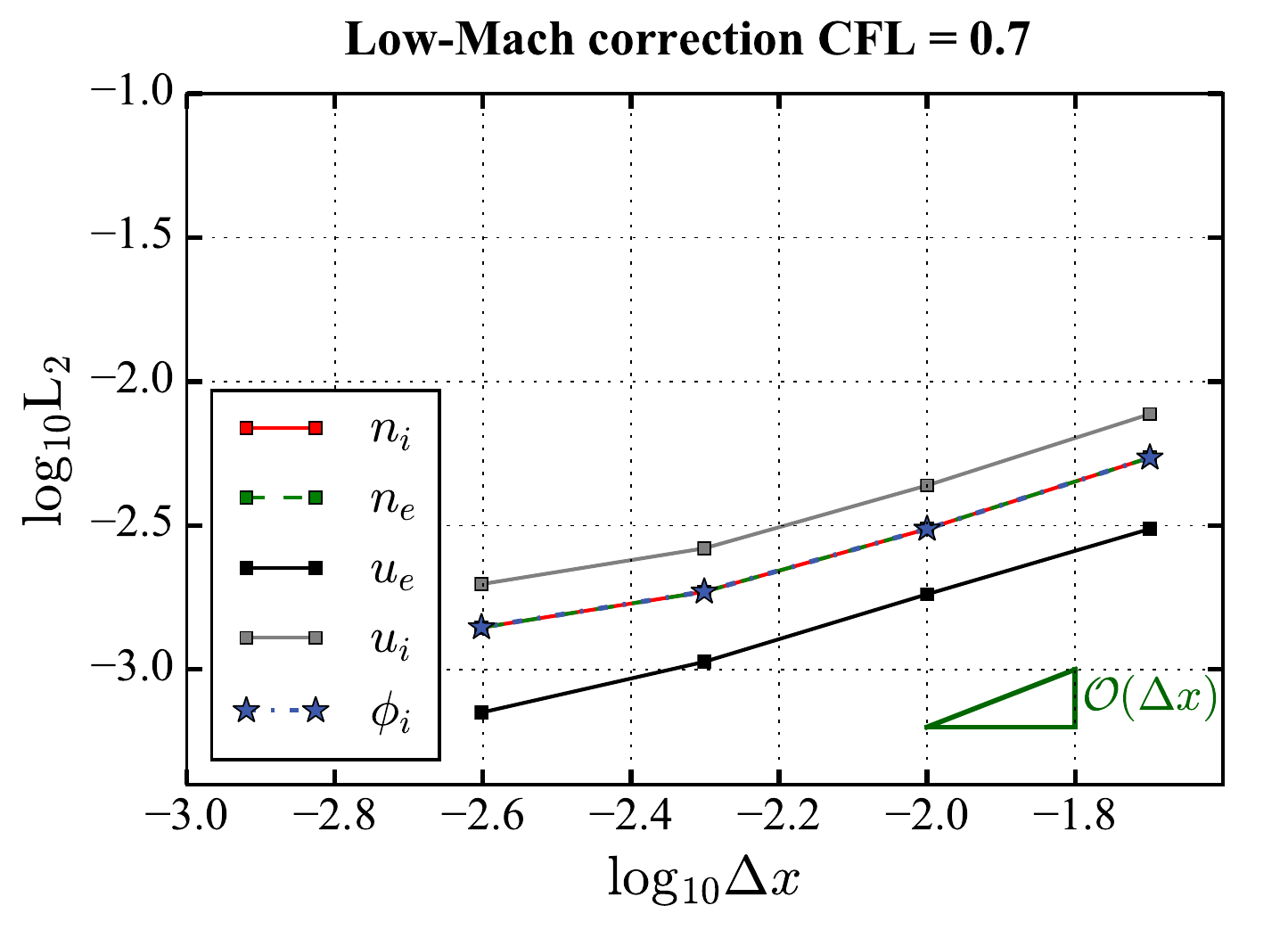} \hfill
	\includegraphics[trim=0cm 0cm 0cm 0cm, clip=true,width=0.32\textwidth]{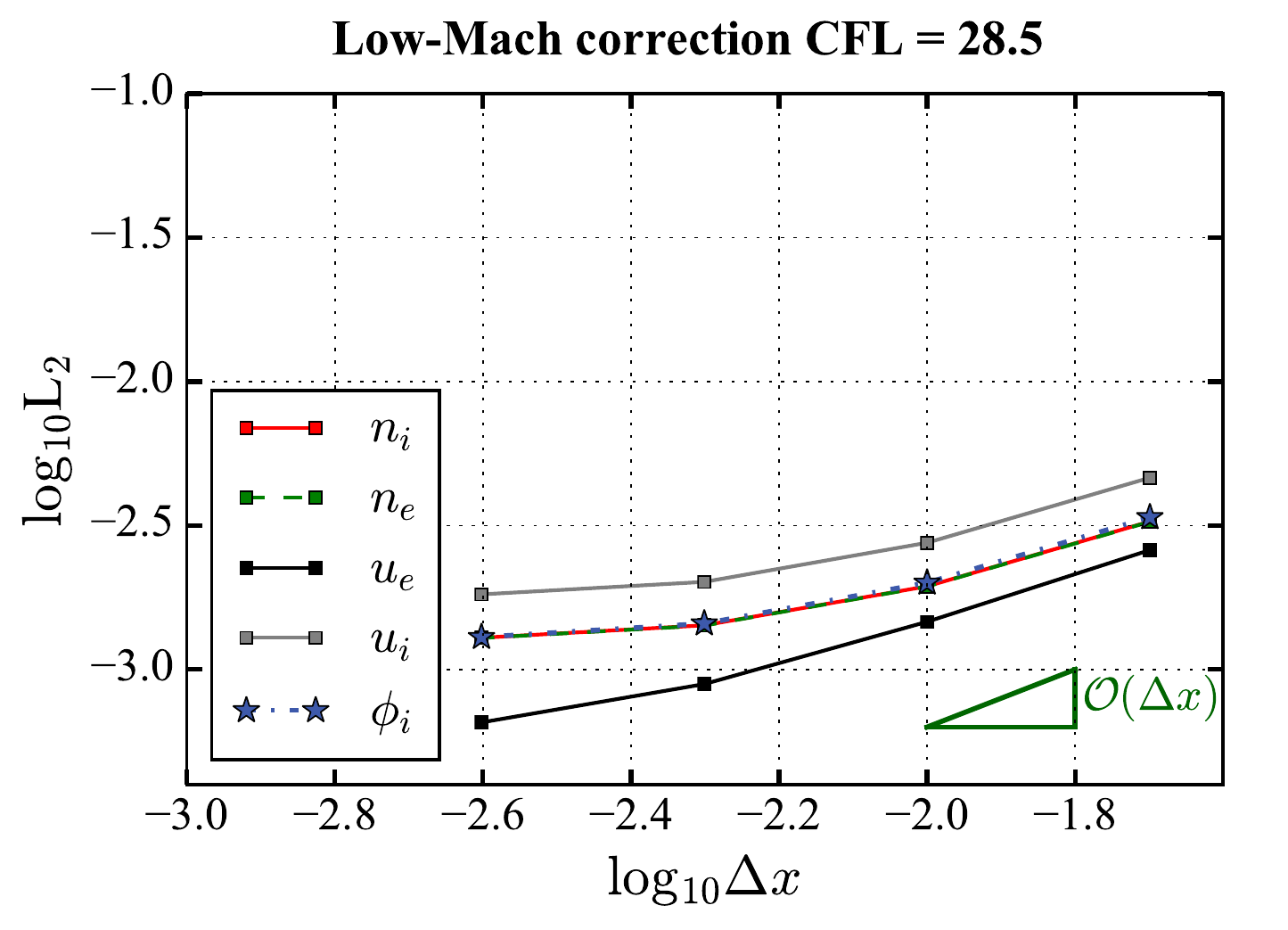} \hfill
	\caption{Comparison of the $L_2$ error norm in the two-stream periodic perturbation in a collisionless thermal plasma. We present the AP scheme without low-Mach correction (left), with low-Mach correction and CFL$^{conv} = 0.7$ (center), and with low-Mach correction and CFL$^{conv} = 28.5$ (right). }
	\label{5_ErrorConvergence}
\end{figure}

\newpage

\subsection{Propagation of a two-stream periodic perturbation in a collisionless low-temperature plasma}
\label{subsec_WB_test}
As mentioned above, in low-temperature plasmas, the bulk velocity of ions is usually much larger than the ion thermal speed. When the ion velocities are comparable or larger than the ion thermal speed, this can produce numerical problems due to the imbalance between the cell-centered source terms and the upwinded fluxes. In order to illustrate the performance of the numerical scheme in a case of low-temperature plasmas, we simulate a two-stream periodic perturbation in a plasma with an ion-to-electron temperature ratio of $\kappaT = 10^{-4}$.

We use the dispersion relation of eq.~\eqref{dispersionRelation} with $\varepsilon = \kappaT = \lambda =10^{-4}$. We choose $n_0 = 1$ and, in this case, $v_0 = 0.7$. Note that $v_0 = 1$ produces a resonance and therefore the perturbations would be too large to study a linear wave propagation. With these conditions, the initial condition reads
\begin{equation}
\Ub(\xb,t=0) =\left(\begin{array}{c}
\rhoeb\\
\ueb\\
\rhoionb\\
\uionb\\
\phib
\end{array}\right)
= 
\left(\begin{array}{c}
1 + 3.3328\times10^{-3}\sin(2\pi \xb)\\
0.7 + 10^{-3}\sin(2\pi \xb)\\
1 + 3.33285\times10^{-3}\sin(2\pi \xb)\\
3.333\times10^{-3}\sin(2\pi \xb)\\
3.3328\times10^{-3}\sin(2\pi \xb)
\end{array}\right).
\end{equation}
The period of the oscillations is 	$T=0.99995$. We run the simulation and compare the results to the initial field after one period.

In Fig.~\ref{6_WBComparison}, we show the comparison between the scheme with the well-balanced source term in eq.~\eqref{eq:WBSource} and a cell-centered implementation of the source. Both simulations use CFL$^{conv}= 25$ and $N = 100$. The approximation of the sign function used follows eq.~\eqref{eq:signFnct} with $u_{\ion_\infty} = v_0$. The results show that the solution without the well-balanced implementation maintain the quasi-neutrality of the solution. However, it develops spurious numerical oscillations. These oscillations are not present with the well-balanced implementation.

\begin{figure} [h] 
	\centering
		  \includegraphics[trim=0cm 0cm 0cm 0cm, clip=true,width=0.495\textwidth]{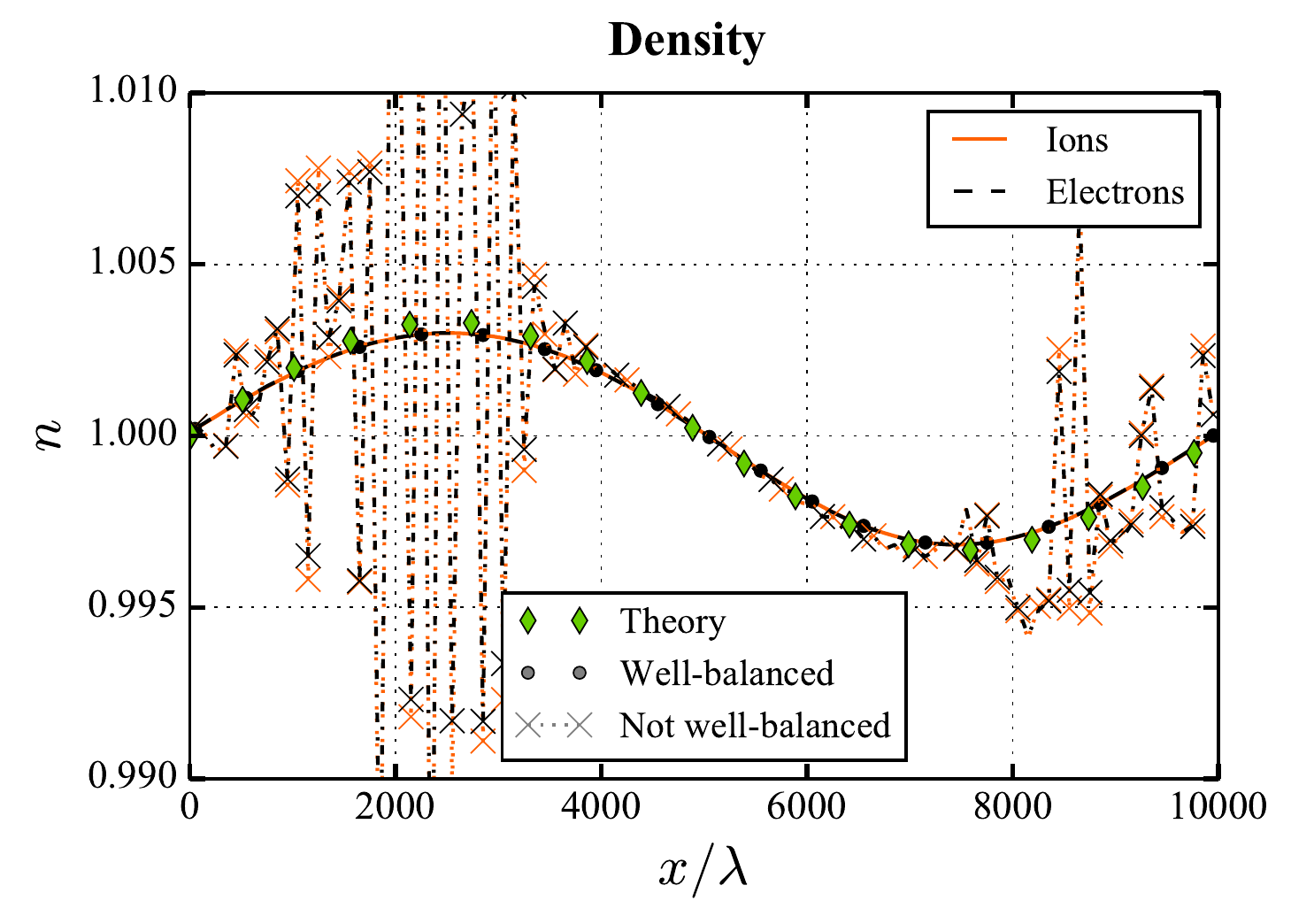} \hfill
		  \includegraphics[trim=0cm 0cm 0cm 0cm, clip=true,width=0.495\textwidth]{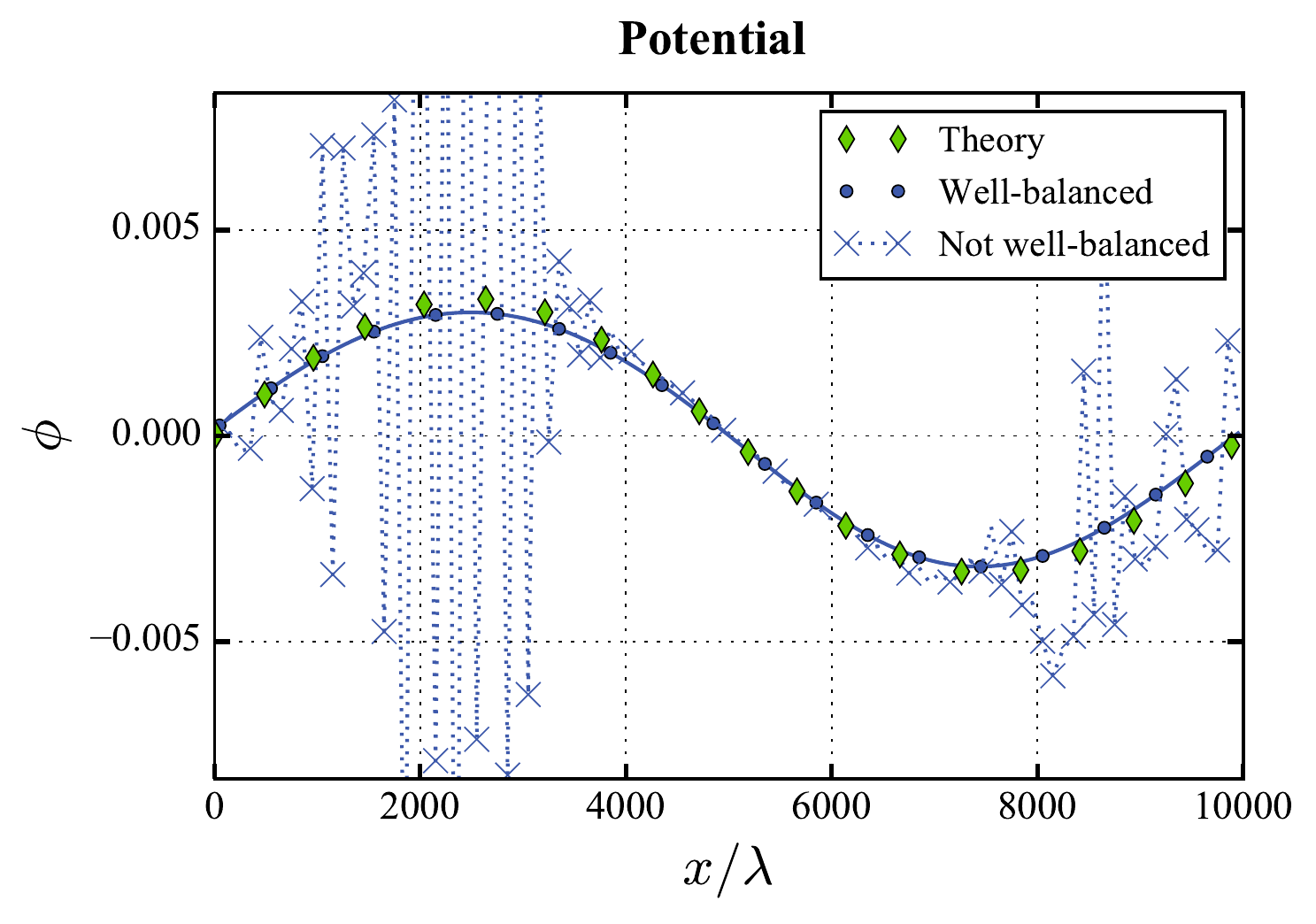}
		  \includegraphics[trim=0cm 0cm 0cm 0cm, clip=true,width=0.495\textwidth]{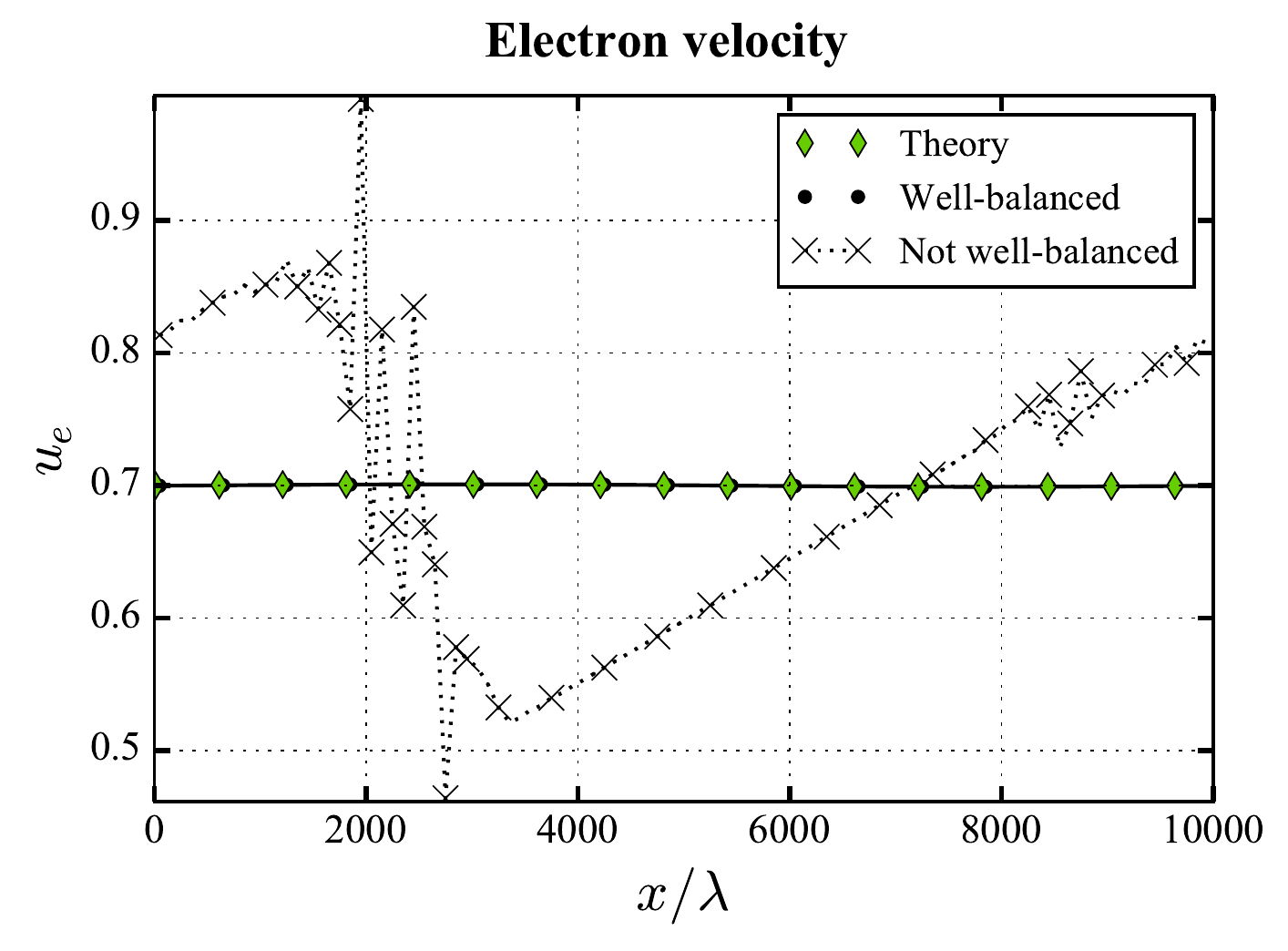} \hfill
		  \includegraphics[trim=0cm 0cm 0cm 0cm, clip=true,width=0.495\textwidth]{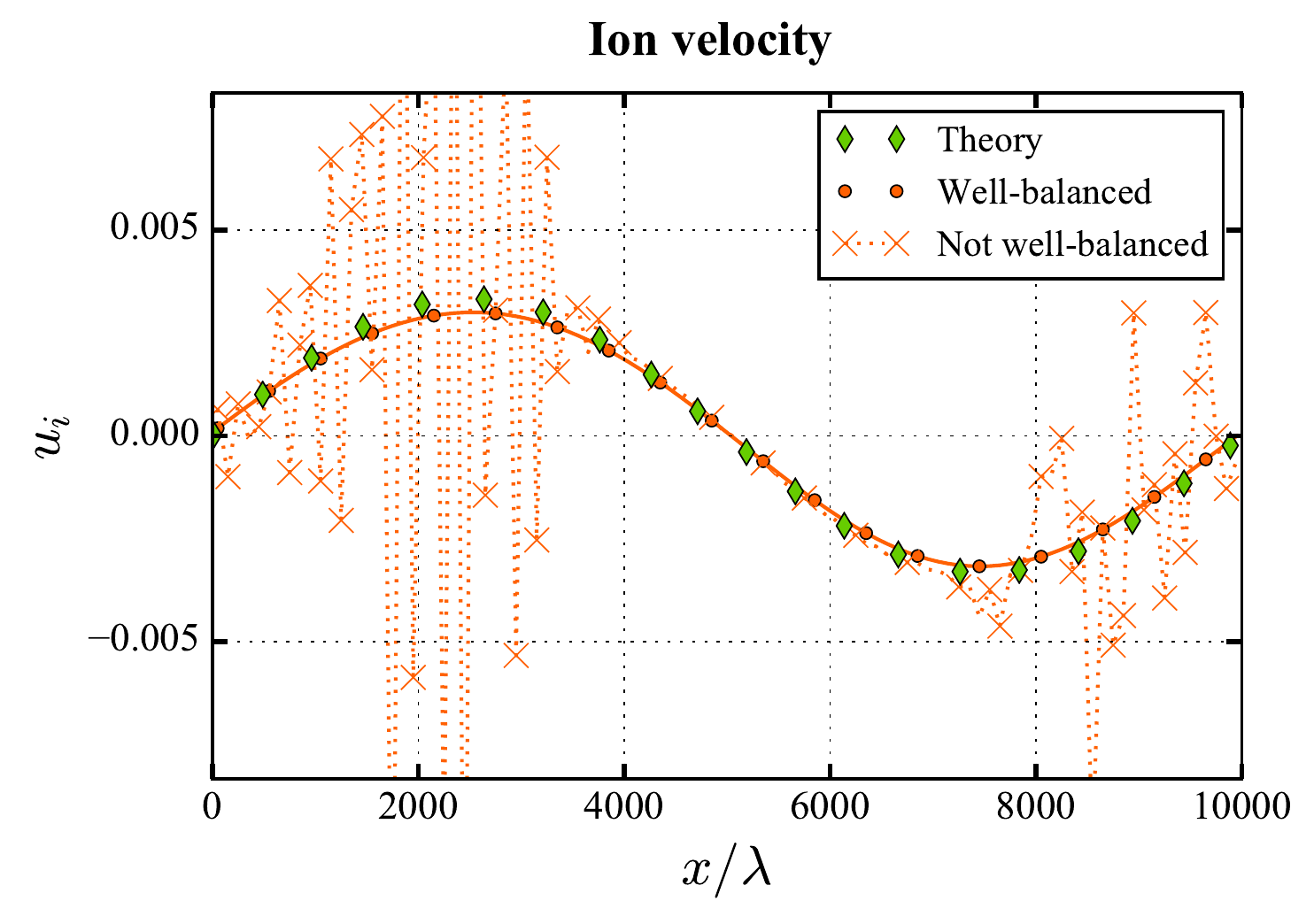}
	\caption{Solution of the two-stream periodic perturbation in a collisionless low-temperature plasma at $t = 0.99995$ using the well-balanced scheme for the ions (solid lines) and with a standard discretization of the Lorentz force (dashed line). The standard discretization develops oscillations that are not present in the well-balanced scheme. The simulations use $N = 100$ points and CFL$^{conv} = 25$. }
	\label{6_WBComparison}
\end{figure}

Under these conditions, the Mach regime of the maximum velocity of the ions is $M_\ion = 0.33$ and the electrons $M_\elec = 0.007$. The convergence for different mesh sizes is presented in Fig.~\ref{7_WB}. In this case, the ions are in a larger Mach regime and therefore the compressible solver is less dissipative than in the thermal plasma case ($\kappa = 1$). As a result, as shown in Fig.~\ref{8_Error}, the error convergence is indeed first-order accuracy in space for this scheme, despite the extremely low Mach regime of the electrons and the limit of $\lambda\rightarrow0$.

\begin{figure} [h] 
	\centering
		  \includegraphics[trim=0cm 0cm 0cm 0cm, clip=true,width=0.495\textwidth]{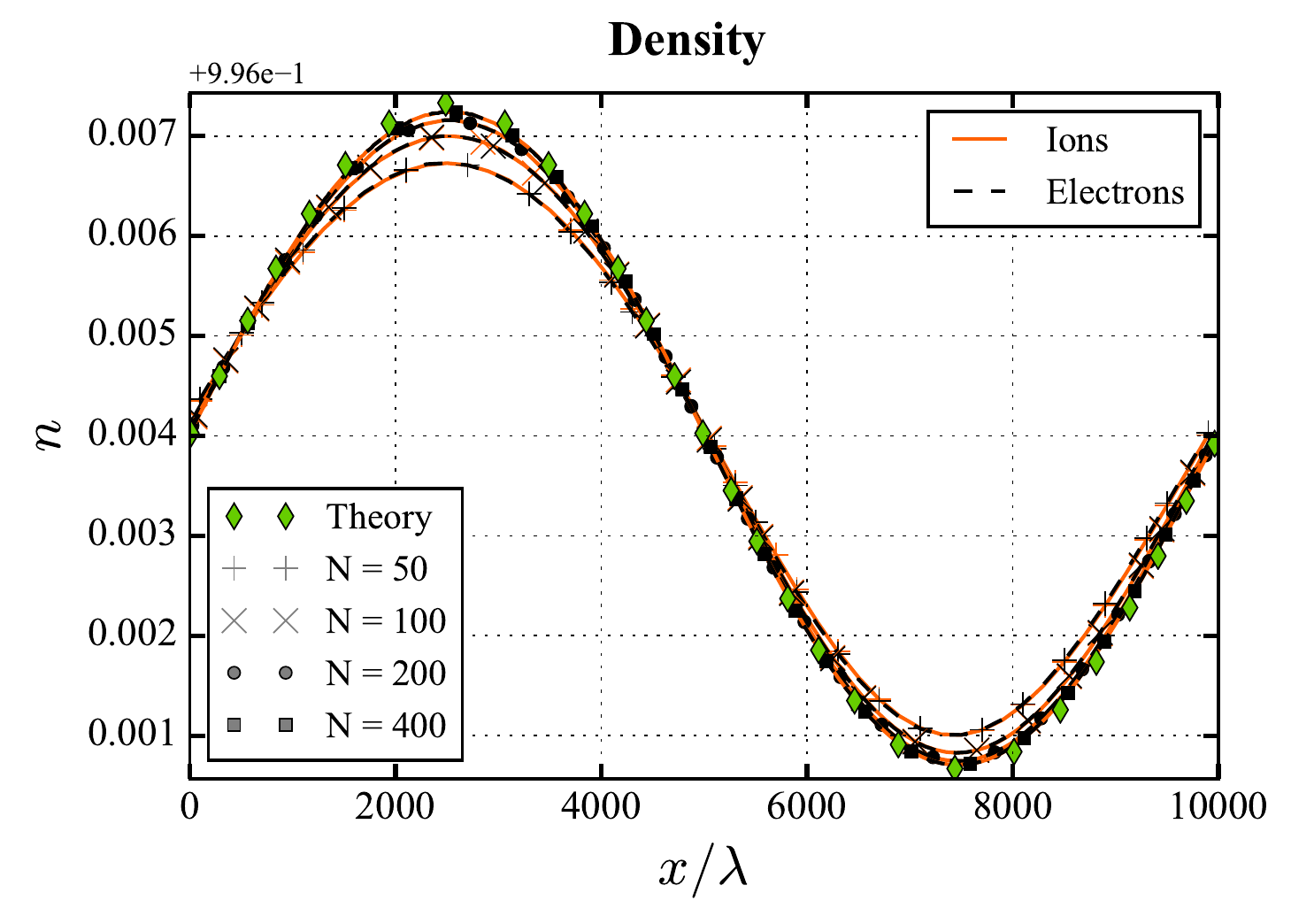} \hfill
		  \includegraphics[trim=0cm 0cm 0cm 0cm, clip=true,width=0.495\textwidth]{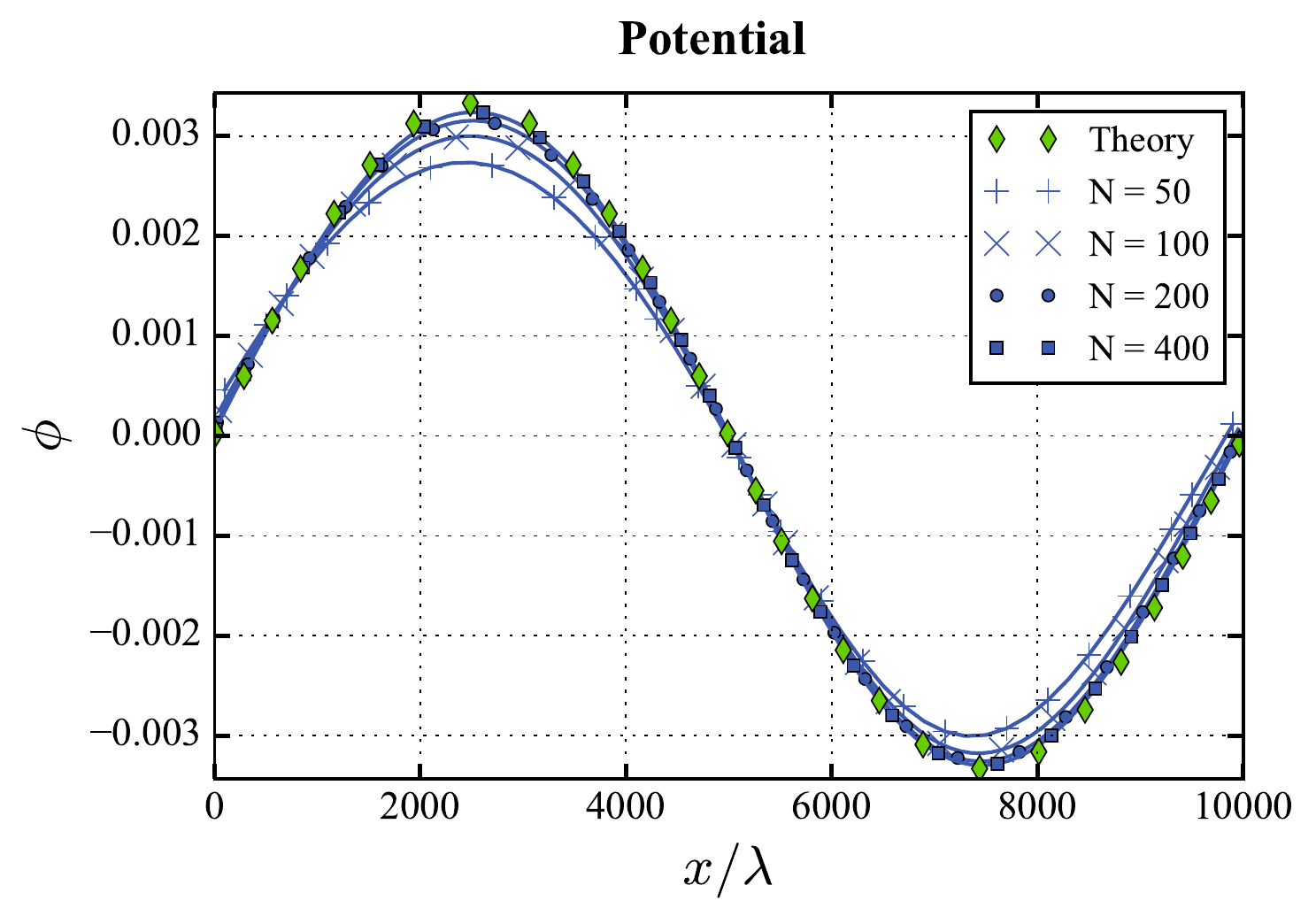}
		  \includegraphics[trim=0cm 0cm 0cm 0cm, clip=true,width=0.495\textwidth]{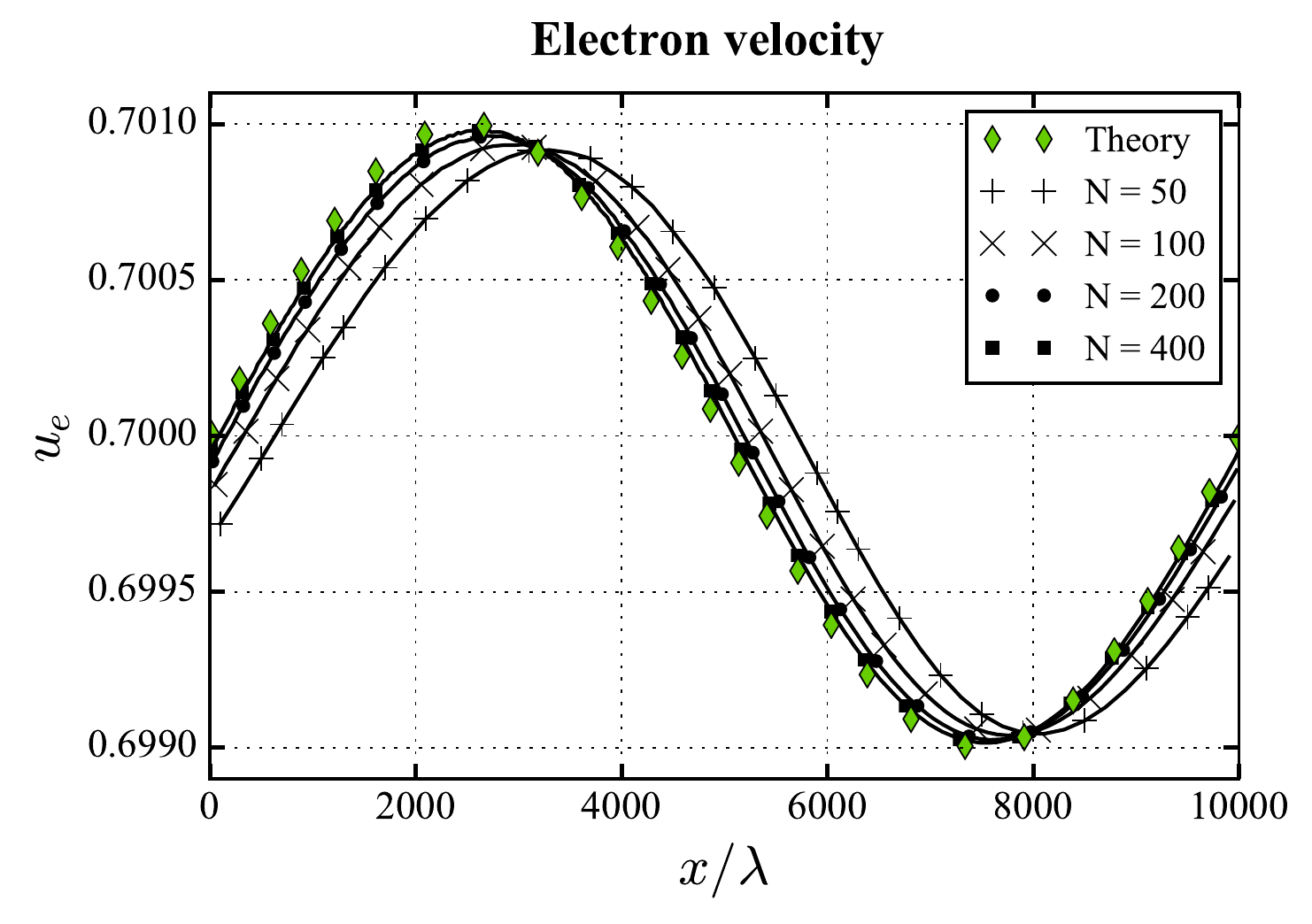}\hfill
		  \includegraphics[trim=0cm 0cm 0cm 0cm, clip=true,width=0.495\textwidth]{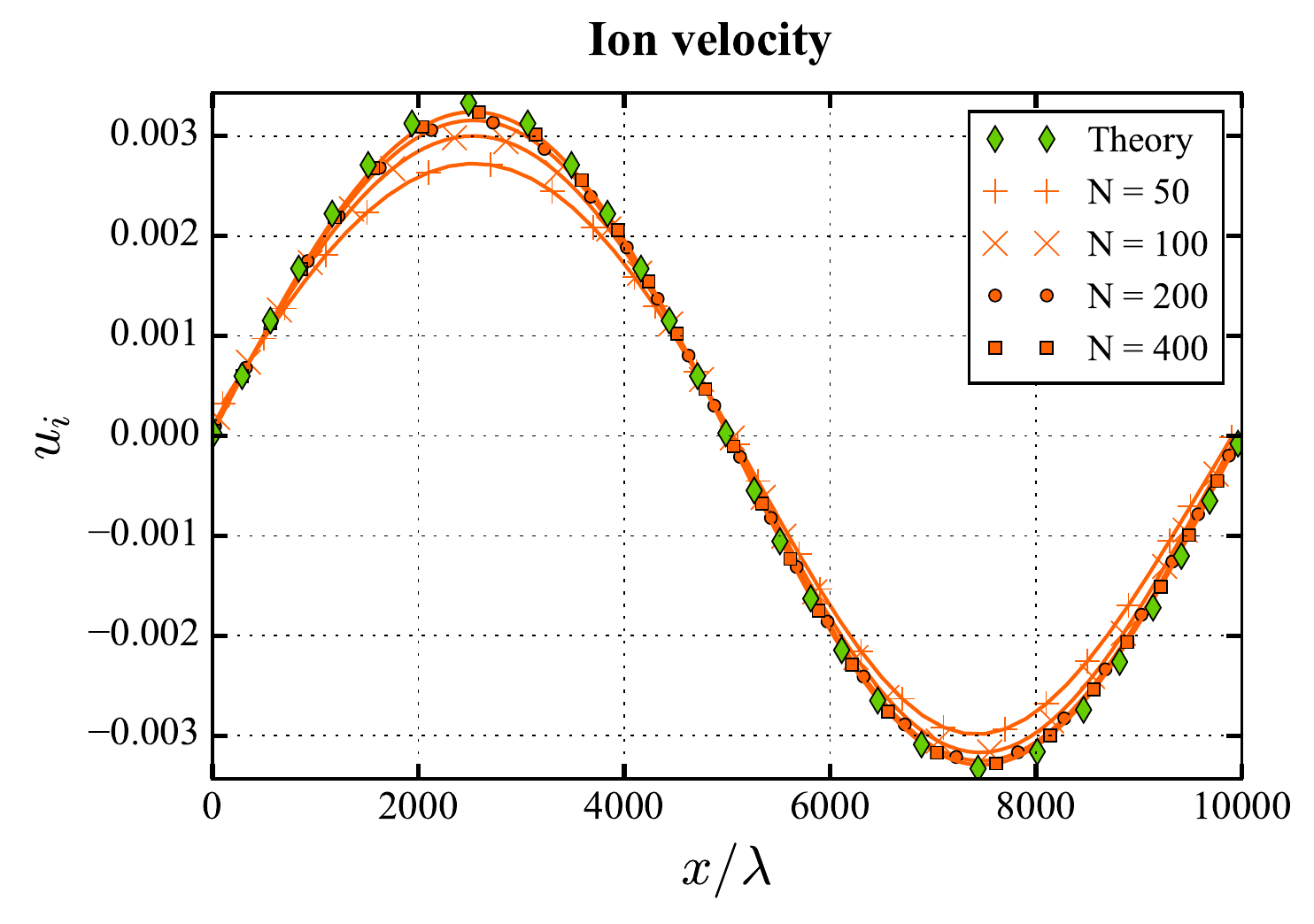} 
	\caption{Solution of two-stream periodic perturbation in a collisionless low-temperature plasma at $t = 0.99995$ for different mesh resolutions with the asymptotic preserving well-balanced scheme. The simulations use CFL$^{conv}=25$. }
	\label{7_WB}
\end{figure}

\begin{figure} [h] 
	\centering
	\includegraphics[trim=0cm 0cm 0cm 0cm, clip=true,width=0.32\textwidth]{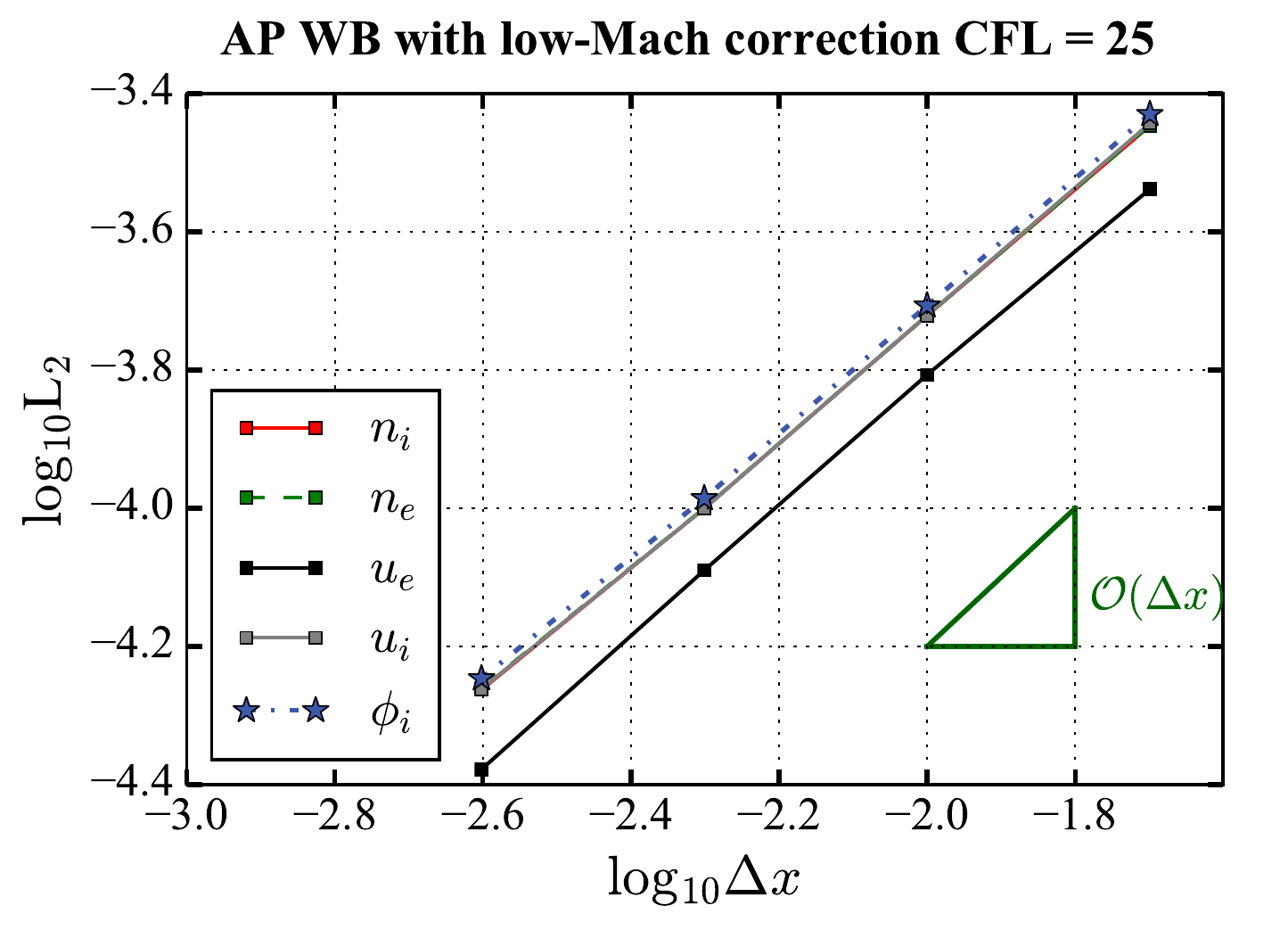} 
	\caption{$L_2$ error norm as function of the mesh size in the simulation of a two-stream periodic perturbation in a collisionless low-temperature plasma. }
	\label{8_Error}
\end{figure}

Finally, we compare the computational performance of the present numerical scheme. We highlight that the Poisson solver with periodic boundaries is solved by adding a Lagrange multiplier to the system that imposes the total charge in the domain to be zero and by solving the linear system with a LU decomposition implemented in the library LAPACK \cite{LAPACK}. We choose to use a linear solver instead of a spectral method as we consider it to be more representative of a realistic multidimentional case with non-periodic boundary conditions. This is due to the fact that most of the iterative methods for solving linear systems do not scale linearly with the number of points and therefore the computational performance will be largely penalized by number of points. 

In Table \ref{tb:Performance}, we present the computational performance  for three different cases: (1) a first-order standard discretization that resolves the Debye length with one cell and the plasma frequency with ten time steps, (2) the same first order standard that a mesh cell that is $\Delta x = 100$ and $\Delta t =0.1 \omega^{-1}_{p\elec}$ (the scheme needs to resolve $\omega^{-1}_{p\elec}$ for stability reasons), and (3) the AP scheme with $\Delta x = 100$ and $\Delta t =2500 \omega^{-1}_{p\elec}$. We highlight that the result of the standard case without resolving the Debye length diverges as the numerical error in the electron velocity becomes very large and the perturbation is not linear after a time-step. 

In Table \ref{tb:Performance}, we can see that the AP scheme produces a dramatic improvement as compared to a standard case when the plasma is quasi-neutral and the characteristic size of the phenomena is very large as compared to the Debye length. As the AP scheme does not need to resolve the Debye length to preserve the quasi-neutrality limit, the AP discretization needs $100$ times less mesh points and $25000$ times less time steps than a scheme that resolves the scales related to the Debye length. Consequently, the CPU time taken for simulating a period of the wave is of the order of $10^9$ times faster than a first-order standard scheme with similar accuracy.

\begin{table}[!htb] 
\caption{Solution of the two-stream periodic perturbation in a collisionless low-temperature plasma: Comparison of the computational performance with the standard and the AP Lagrange-projection scheme.}\label{tb:Performance} 
\begin{center}
\resizebox{\columnwidth}{!}{
\begin{tabular}{ l || c | c | c | c }
Numerical scheme and resolution & Cells & Iters.~for one period & CPU time per iter.~[s] &CPU time for one period [s]\\
  \hline\\[-0.34cm]
  \hline\\[-0.37cm]
  Standard $\Delta x = \lambda$, $\Delta t =0.1 \omega^{-1}_{p\elec}$ & $10000$ & $10^7$ & $41.7175369$ & $4.2\times10^8$ \\
  Standard $\Delta x = 100 \lambda$, $\Delta t =0.1 \omega^{-1}_{p\elec}$& $100$ & $10^7$ & $1.1\times10^{-3}$ & $1.1\times10^4$\\
  AP $\Delta x = 100 \lambda$, $\Delta t =2500 \omega^{-1}_{p\elec}$& $100$ & $400$ & $1.1\times10^{-3}$ & $0.44$
\end{tabular}
}
\end{center}
\end{table}

\subsection{Collisionless isothermal sheath}

The previous numerical experiments simulate a quasi-neutral plasma. Nevertheless, as discussed in Section \ref{sec:problemdefinition}, the set of equations allows for capturing charge separation effects, as shown in Fig.~\ref{Fig:asymLimits}. For length scales much greater than the Debye length, the plasma behaves in general as quasi-neutral. However, when a surface is in contact with the plasma, a charged boundary layer called the plasma sheath is formed. 

We consider a plasma in a 1D domain of unitary size with an electron-to-ion mass ratio $\varepsilon = 10^{-5}$, ion-to-electron temperature ratio $\kappa=0.025$, and $\lambdaSq = 10^{-4}$. The electron flux is imposed on both boundaries by the number of particles crossing the plane with positive velocity component
\begin{equation}
  \rhoeb\ueb(\bar{x}=0,t) = - \frac{\rhoeb}{\sqrt{2\pi\varepsilon}}~~~\text{and}~~~\rhoeb\ueb(\bar{x}=1,t)  = \frac{\rhoeb}{\sqrt{2\pi\varepsilon}}.\label{fluxElectrons}
\end{equation}
The electron and ion density and the ion flux have a Neumann boundary on both sides. Alternatively, the electric potential has a Dirichlet condition on both sides $\phib(\bar{x}=0,t) = \phib(\bar{x}=1,t) = 0$. In the results, the potential is referenced to the potential in the plasma at $\bar{x} = 1/2$.

In the sheath theory \cite{Riemann05}, the ionization is an eigenvalue of the problem. This means that there is only one ionization rate that can produce a steady solution. In our simulation set-up, we find this eigenvalue by changing the ionization frequency such the balance between ionization and ion flux occurs at every time step. Consequently, the ionization frequency at $t^n$ is calculated as
\begin{equation}
  \Damk(t^n) = \frac{|\rhoionb\uib(x=0,t^n)| + |\rhoionb\uib(x=l,t^n)|}{\int_0^l n_e(t^n)\ dx}.
\end{equation}

The flow field is initialized as follows
\begin{equation}
\Ub(\xb,t=0) =(\rhoeb,  ~\rhoeb\ueb, ~\rhoionb,~\rhoionb\uionb,~\phib)^T = (1,  ~0, ~1,~0,~0)^T .
\end{equation}

As the assumptions of the classical theory are not taken by our numerical model, we might expect a result that is slightly different from eq.~\eqref{eq:SheathTheory}. For this reason, we will use as a reference solution a high-order highly resolved simulation described in Section \ref{sec:problemdefinition}. In Fig.~\ref{9_ReferenceSolution}, we show the steady solution reference solution that uses a standard discretization, as in Section \ref{sec:standardDiscretization}, with TVD third-order reconstruction and third-order TVD Runge-Kutta \cite{Gottlieb98} scheme. In the reference solution, the domain is resolved with $N = 10^4$ mesh points and CFL=$0.9$. By doing this, we ensure that the Debye length is resolved with $100$ points and so is the electron plasma period. Since the scheme is high order and highly resolved, the discretization errors are expected to be much smaller than a first order scheme with lower resolution. 

The solution shows two plasma sheaths beside the left and the right boundaries of around $6$ -- $7$ Debye lengths width. As shown in the densities, the plasma is quasi-neutral in the rest of the domain, whereas it is positively charged in the sheath. The plasma potential at the wall is very similar to the theory, despite the electron inertial and the ion pressure gradient effects that are neglected in the theory. The flux of ions and electrons are equal in the whole domain. The velocities are also equal in the quasi-neutral region and they differ when the ions reach $|\uionb| = 1$, i.e., the Bohm's velocity in our units. This point is the plasma-sheath transition that agrees with the Bohm's criterion \cite{Bohm49}.

\begin{figure} [H] 
	\centering
		  \includegraphics[trim=0cm 0cm 0cm 0cm, clip=true,width=0.495\textwidth]{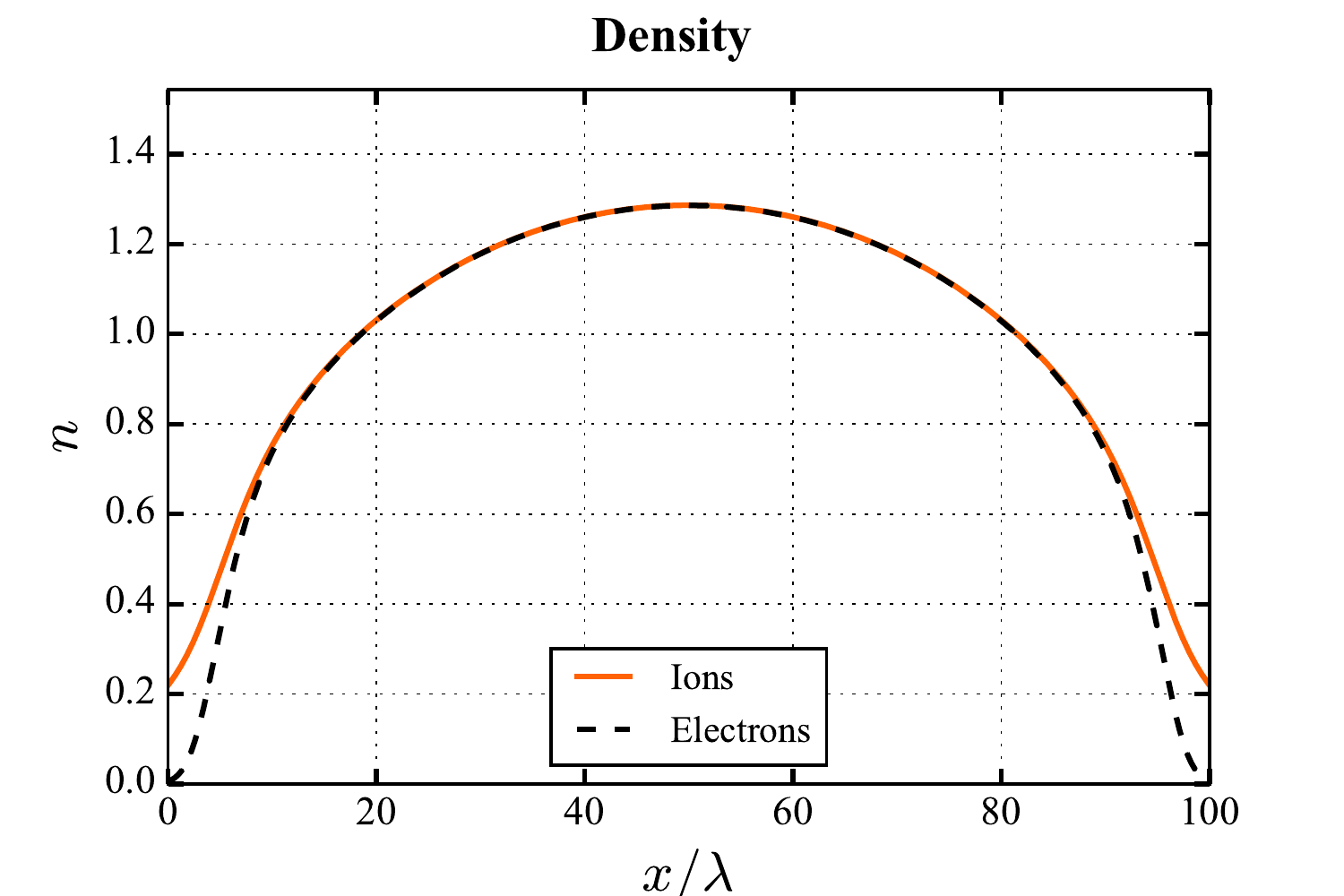} \hfill
		  \includegraphics[trim=0cm 0cm 0cm 0cm, clip=true,width=0.495\textwidth]{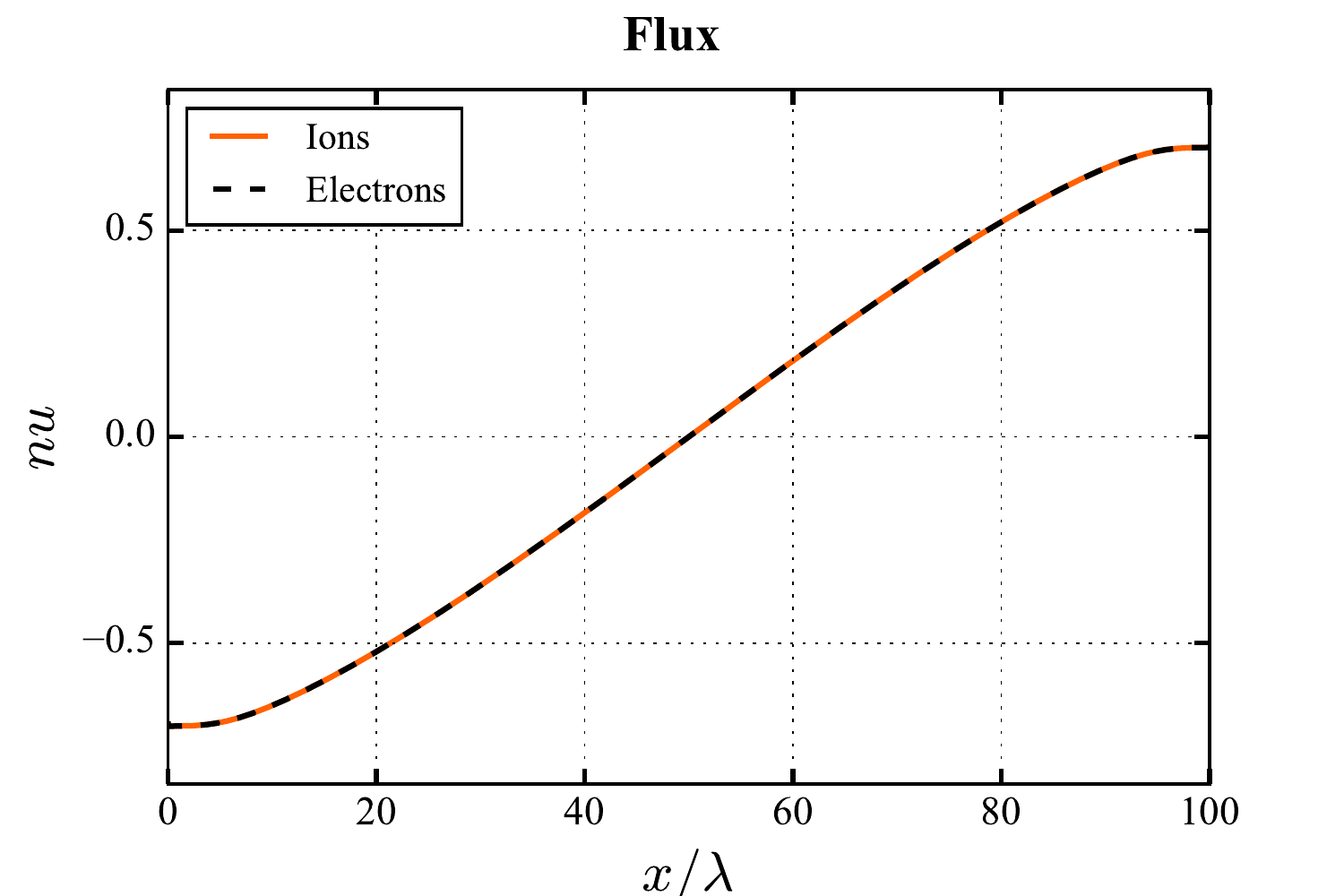}
		  \includegraphics[trim=0cm 0cm 0cm 0cm, clip=true,width=0.495\textwidth]{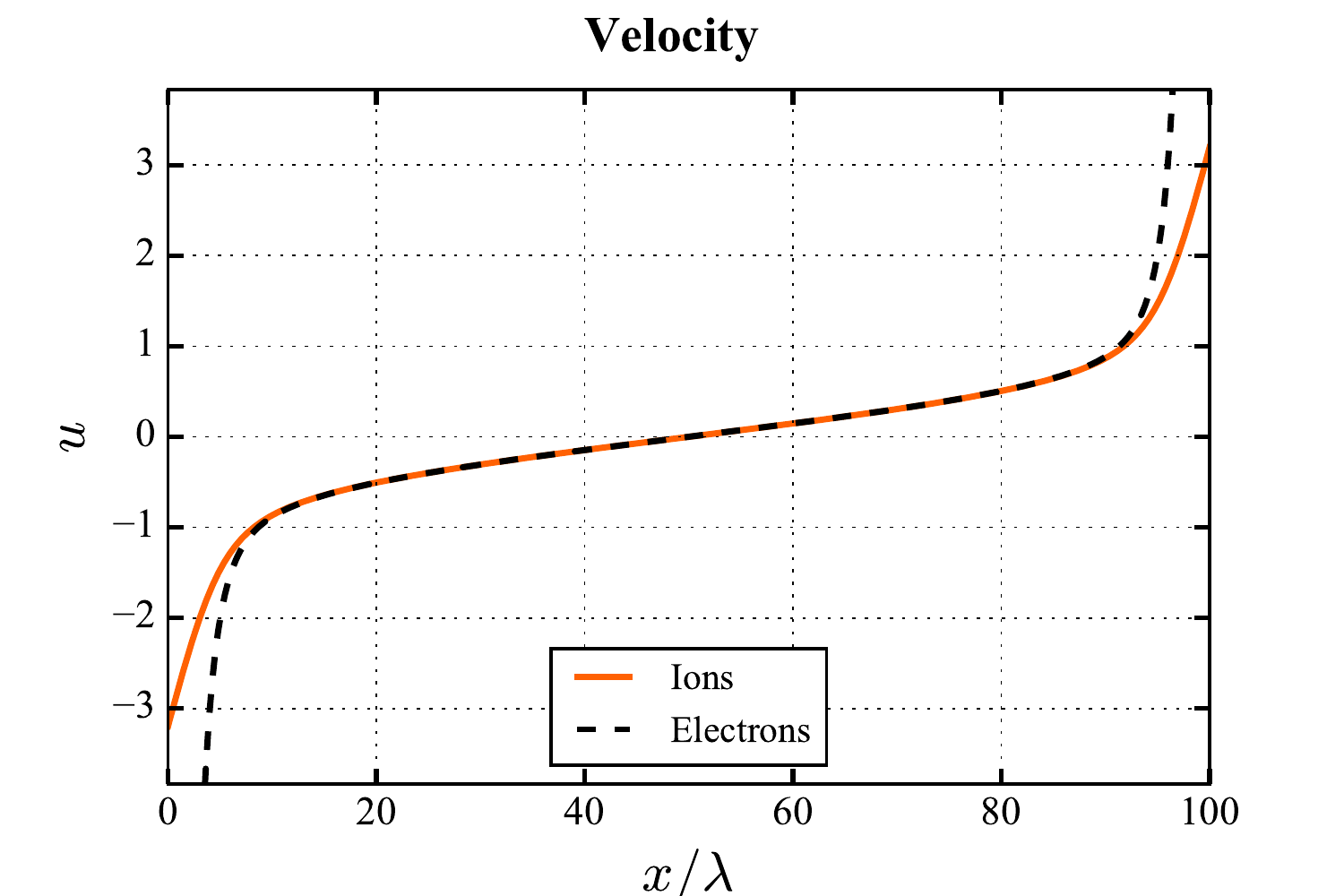} \hfill
		  \includegraphics[trim=0cm 0cm 0cm 0cm, clip=true,width=0.495\textwidth]{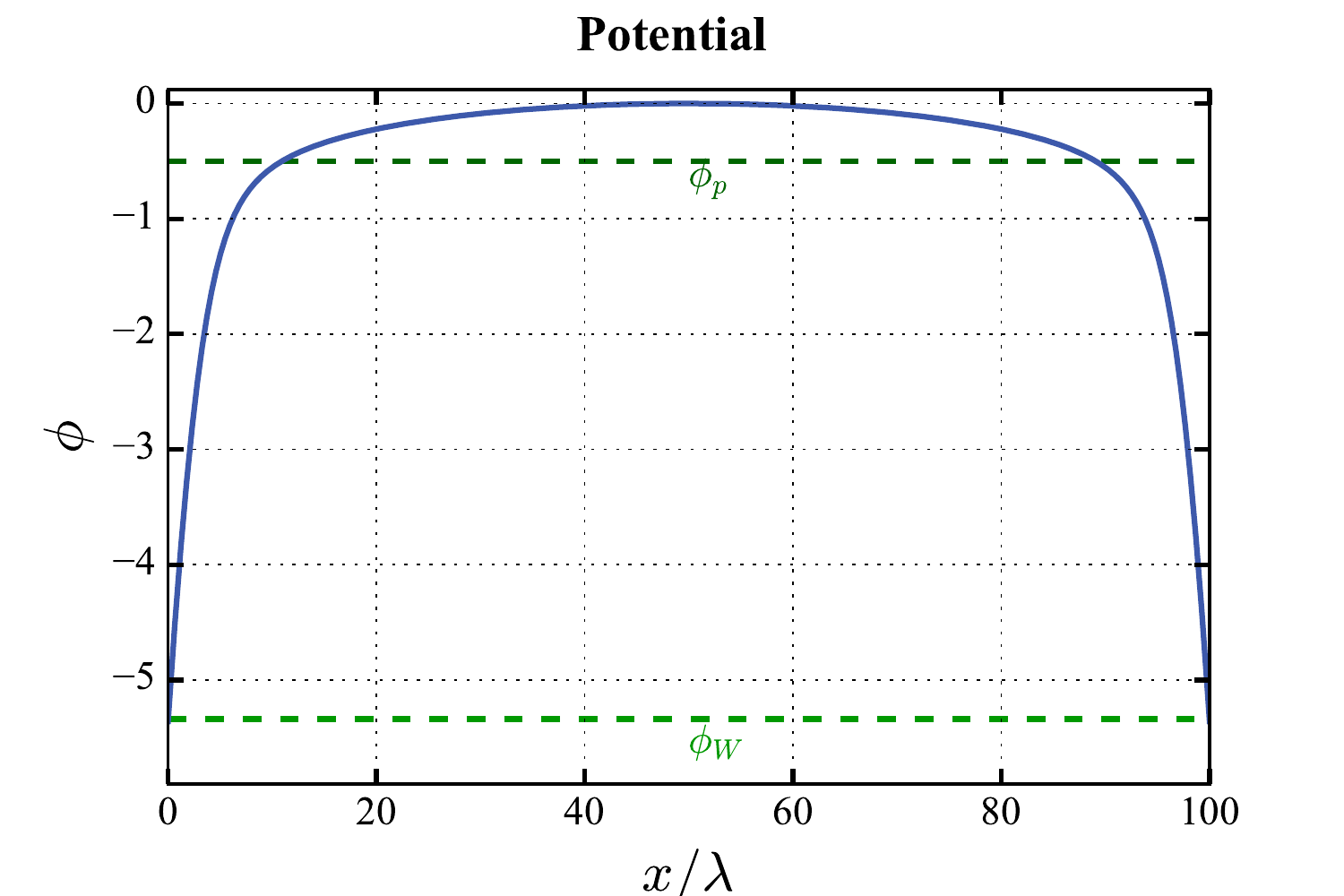}
	\caption{Collisionless isothermal sheath: Solution of a plasma between two floating walls with a third order TVD scheme using $N=10^4$ mesh points. The proposed numerical set-up is able to captures the physics as predicted by the theory \cite{Riemann05}.}
	\label{9_ReferenceSolution}
\end{figure}


In Fig.~\ref{10_FirstOrderHLL}, we present the converged results using a first-order standard discretization with CFL $= 0.1$ and $N=10^3$. The local error is calculated with the reference solution. The main difference with the reference solution can be clearly observed in the electron flux. The error of the electron flux is as large as the actual solution. This is due to the poor behaviour of the HLL scheme under low-Mach conditions, as explained in the previous examples.

\begin{figure} [H] 
	\centering
		  \includegraphics[trim=0cm 0cm 0cm 0cm, clip=true,width=0.355\textwidth]{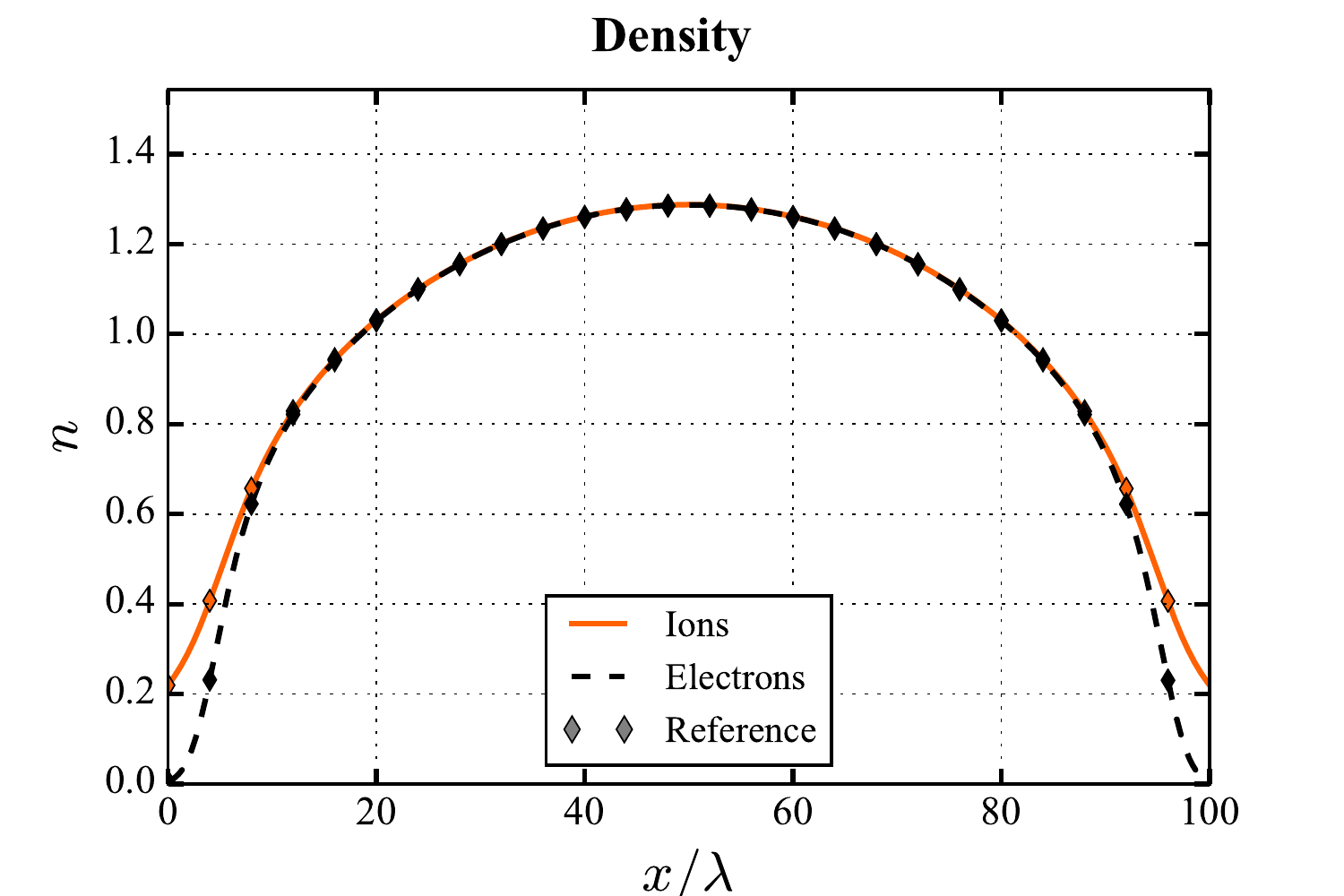} 
		  \includegraphics[trim=0cm 0cm 0cm 0cm, clip=true,width=0.355\textwidth]{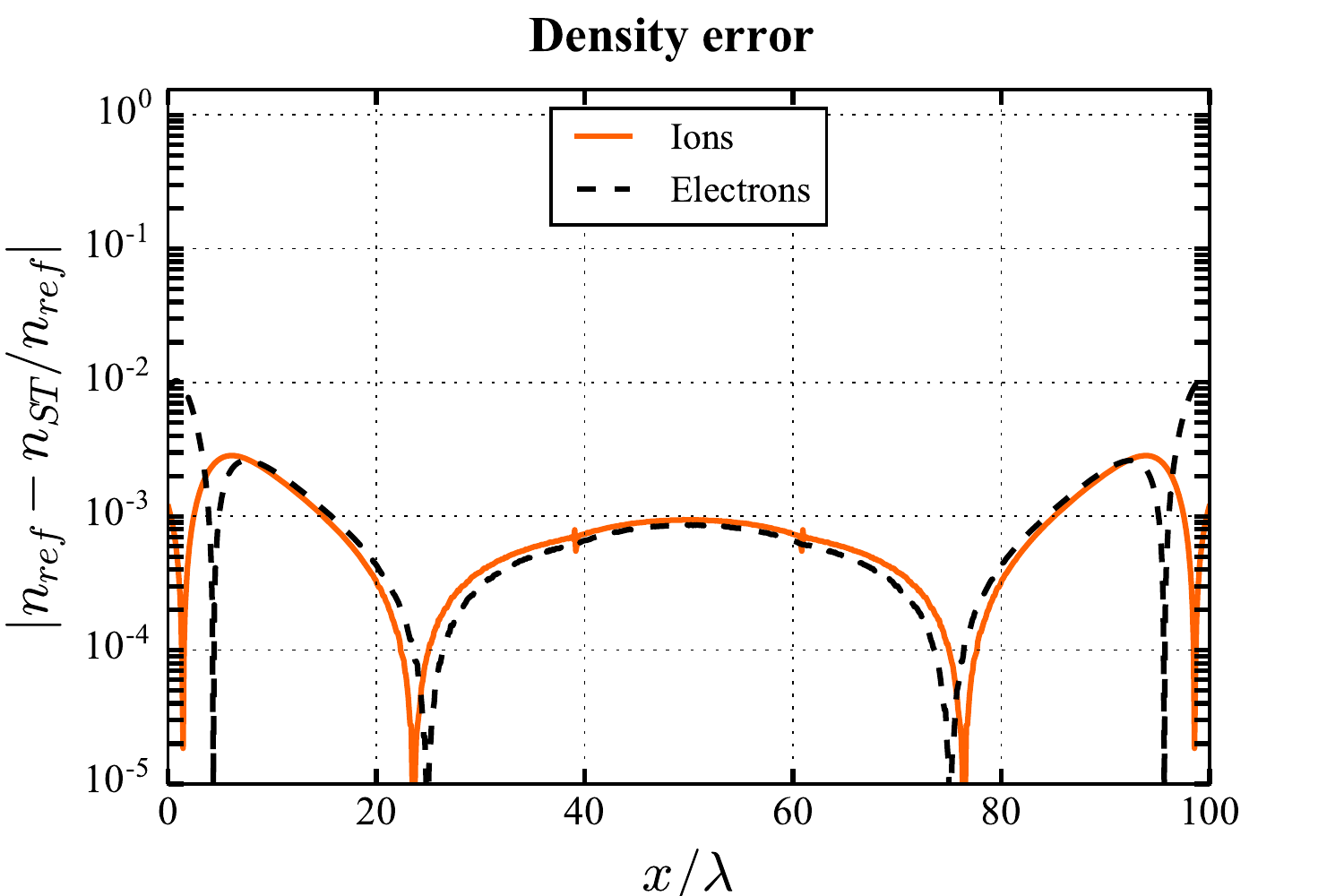}
		  \includegraphics[trim=0cm 0cm 0cm 0cm, clip=true,width=0.355\textwidth]{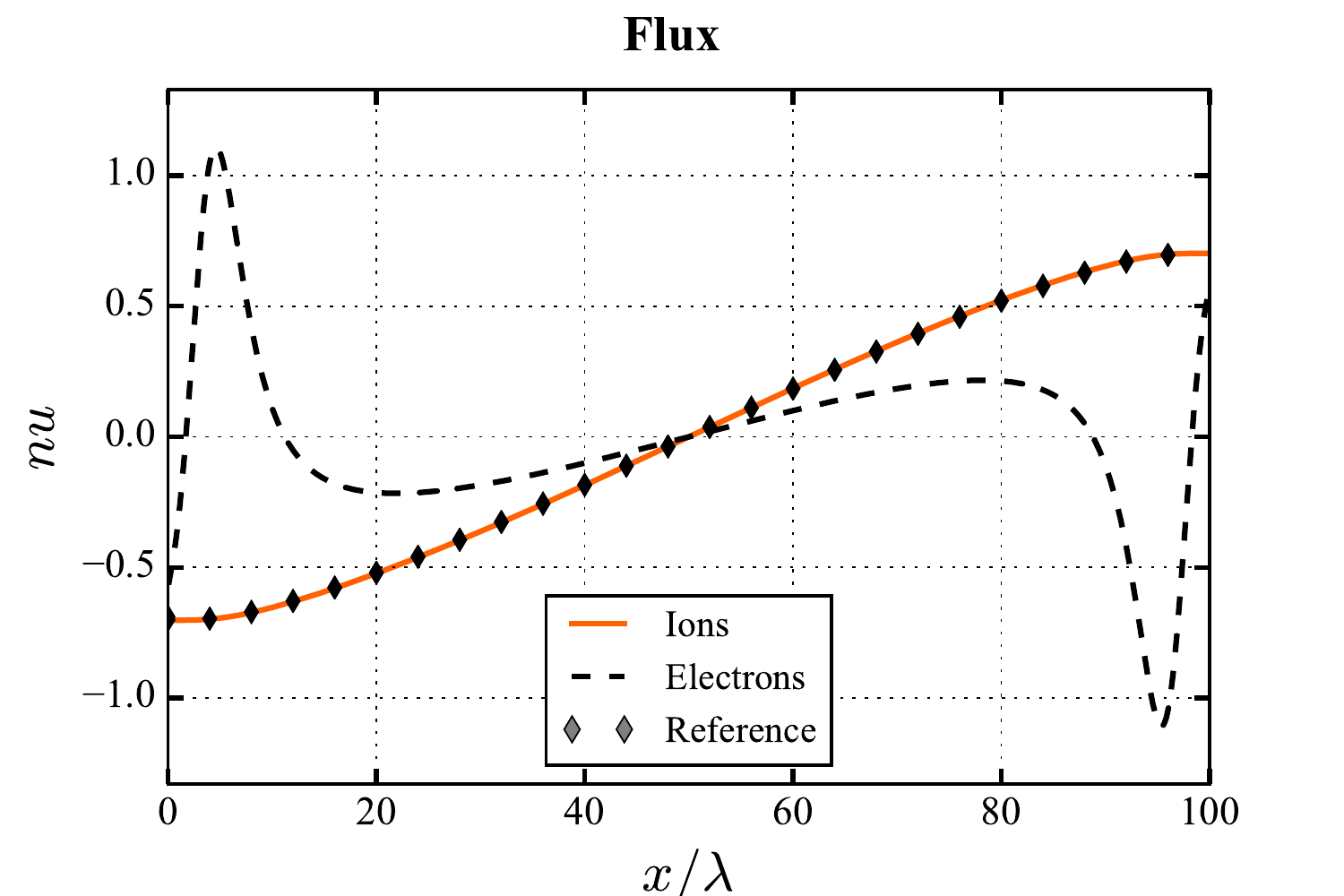} 
		  \includegraphics[trim=0cm 0cm 0cm 0cm, clip=true,width=0.355\textwidth]{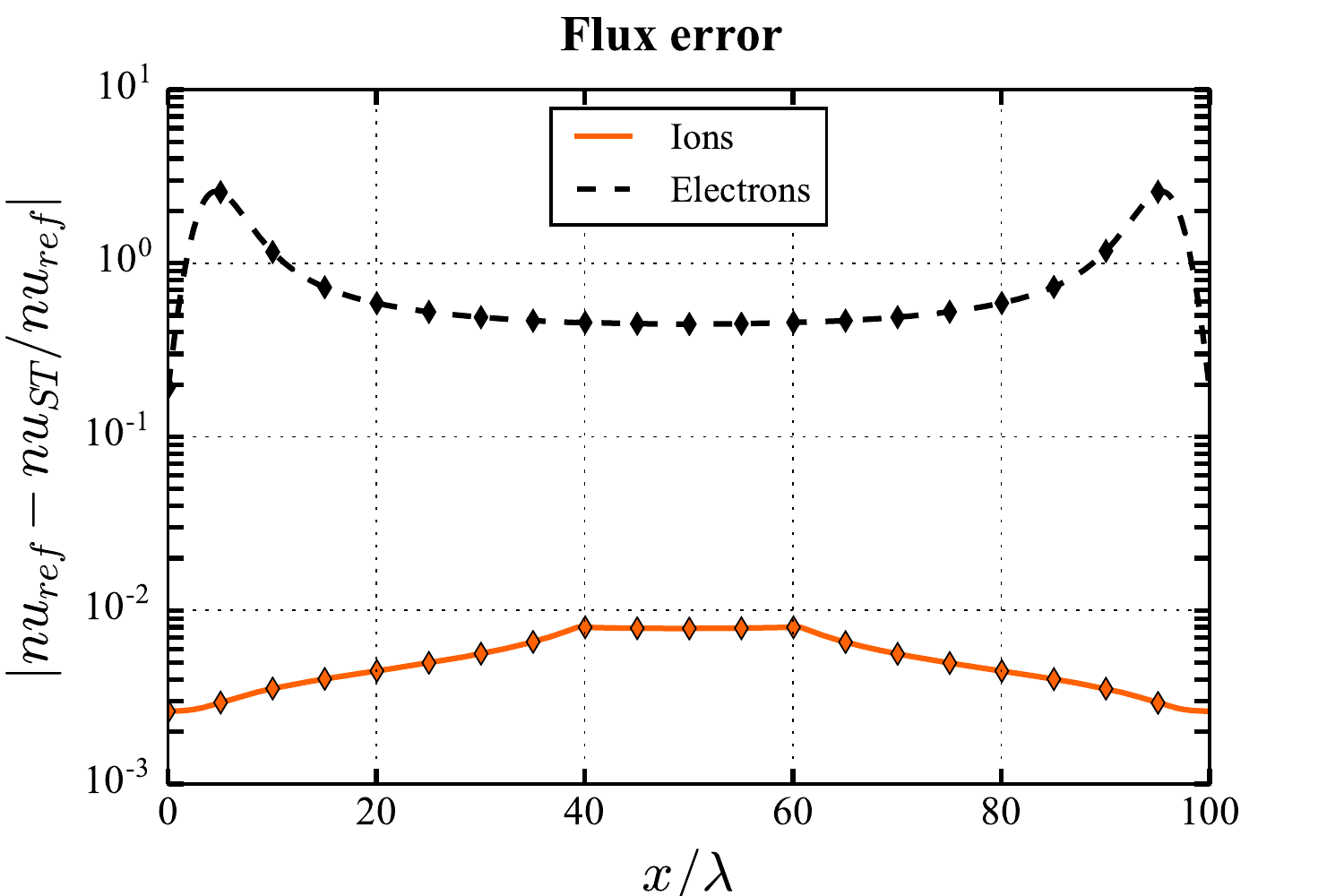}
		  \includegraphics[trim=0cm 0cm 0cm 0cm, clip=true,width=0.355\textwidth]{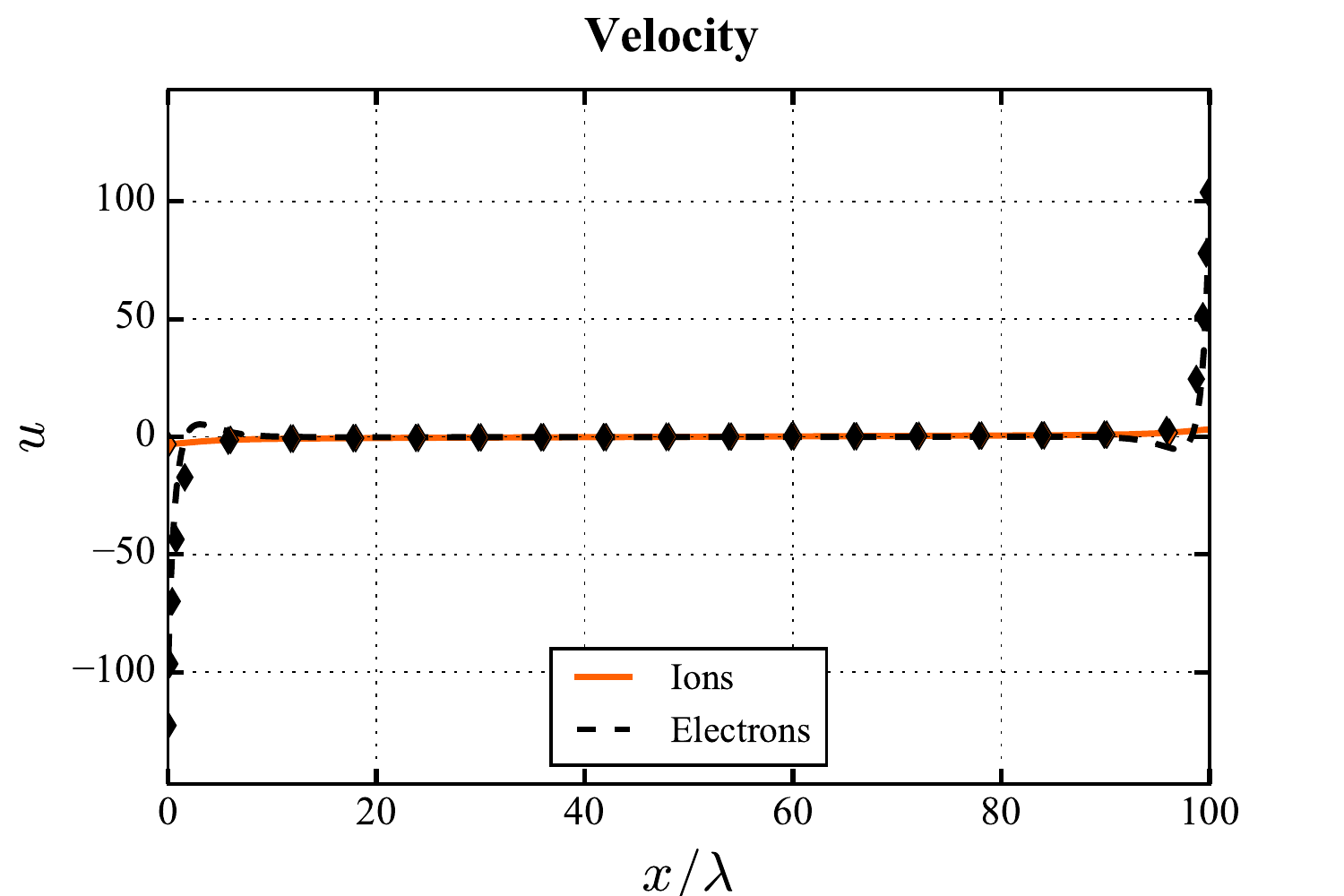} 
		  \includegraphics[trim=0cm 0cm 0cm 0cm, clip=true,width=0.355\textwidth]{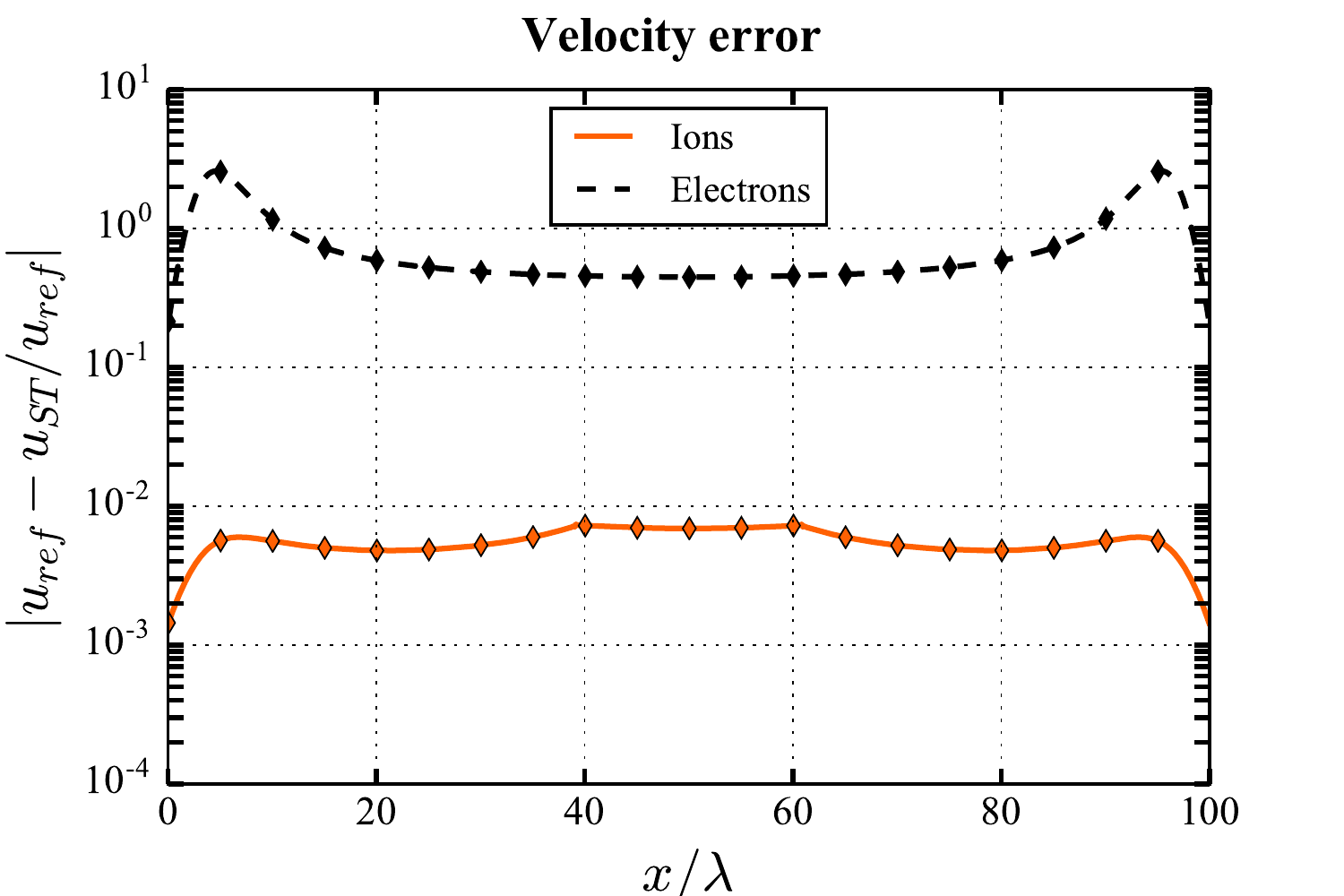}
		  \includegraphics[trim=0cm 0cm 0cm 0cm, clip=true,width=0.355\textwidth]{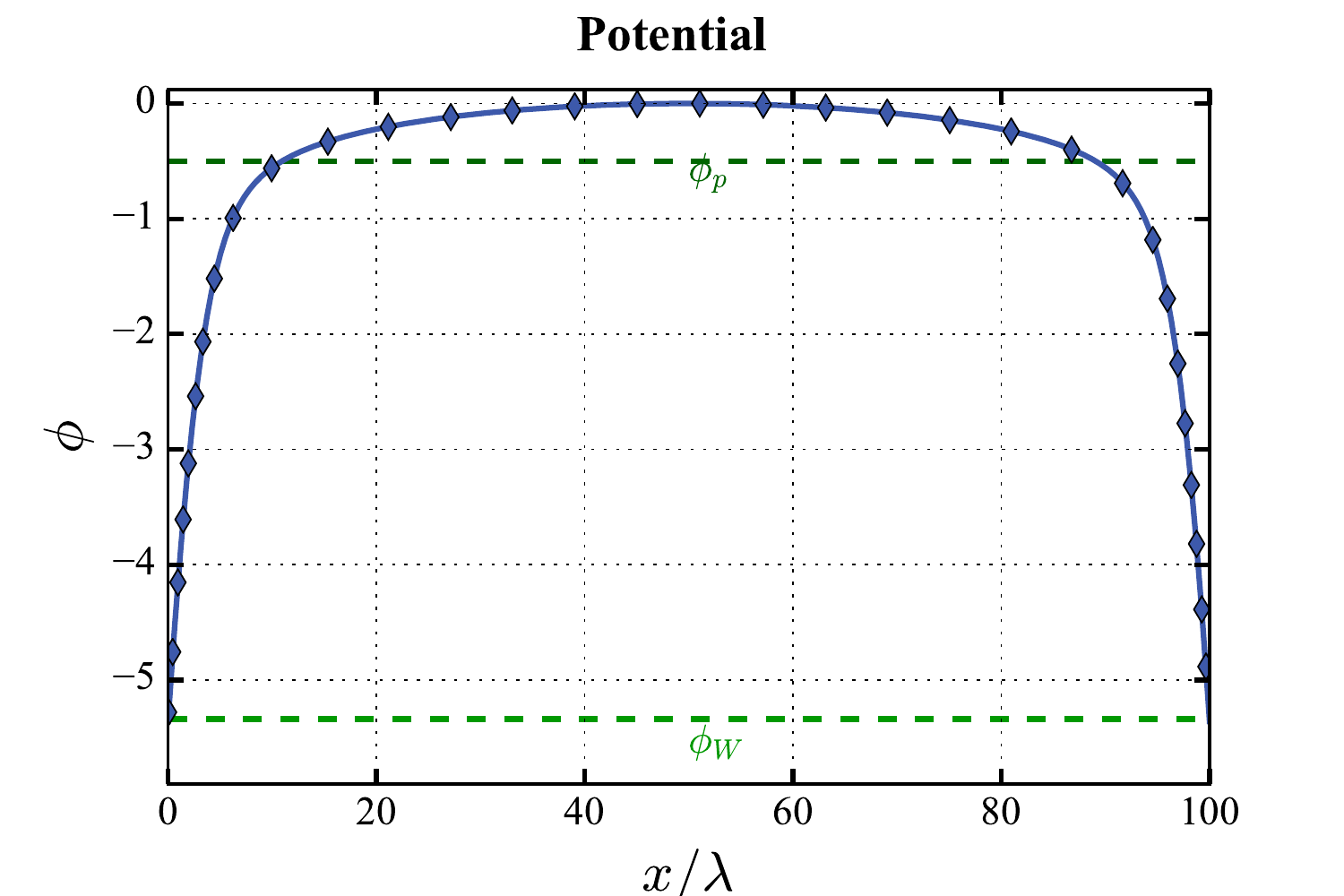} 
		  \includegraphics[trim=0cm 0cm 0cm 0cm, clip=true,width=0.355\textwidth]{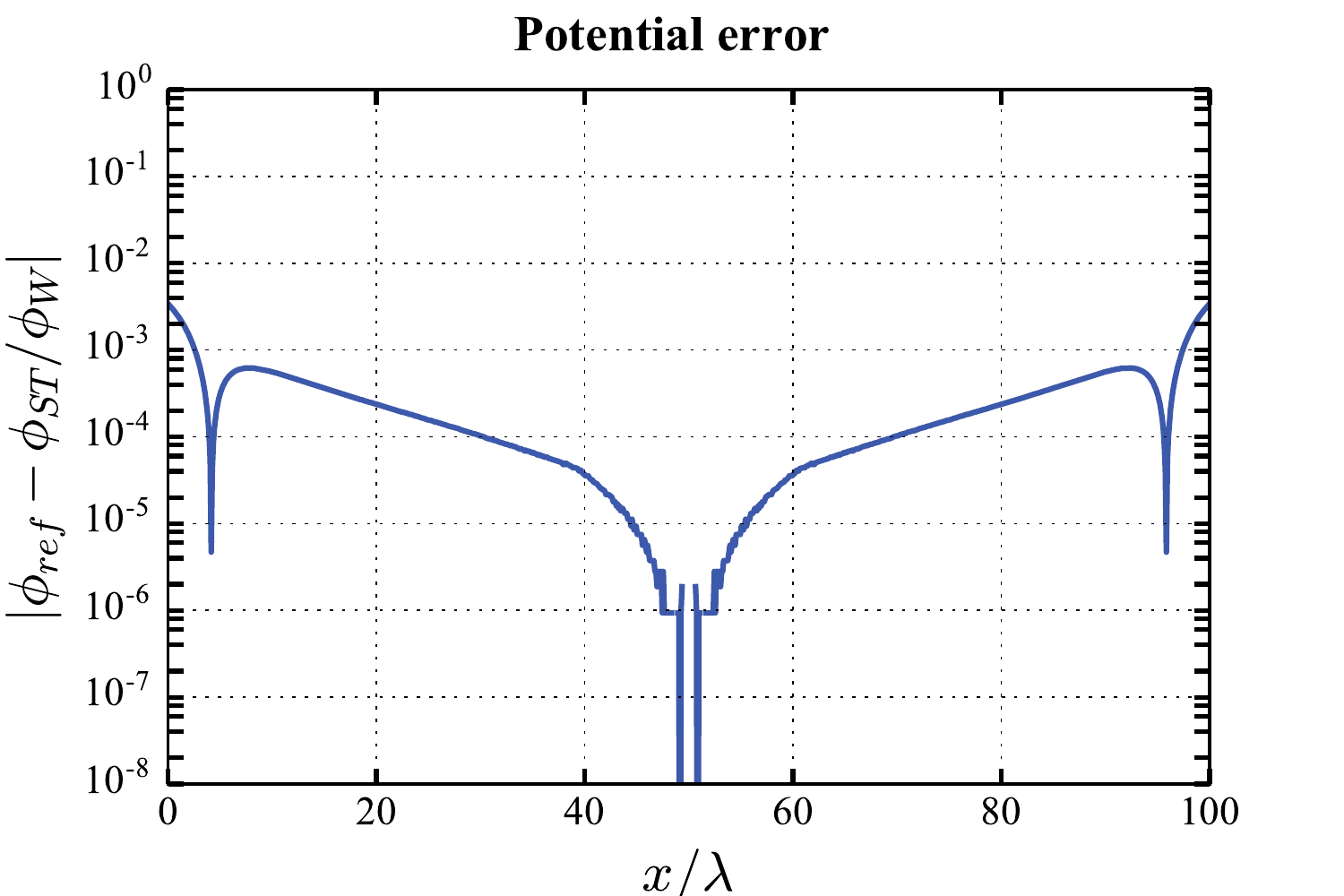}
	\caption{Collisionless isothermal sheath: Converged first-order solution with a HLL scheme and $\text{CFL} =0.1$ and $10^3$ mesh points. The simulation shows a large error in the electron flux due to the large numerical error at low-Mach speeds.}
	\label{10_FirstOrderHLL}
\end{figure}

The results using the AP Lagrange-projection scheme with CFL $^{conv}= 0.4$ and $N=10^3$ are presented in Fig.~\ref{11_APscheme}. We note that, as we want to resolve the plasma sheath, in this case we resolve the Debye length in our discretization. The results show that the plasma sheaths are properly captured with the scheme and the simulation agrees with the reference solution. Additionally, the electron flux is correctly captured, reducing by two orders of magnitude the numerical error. The rest of magnitudes have similar values to the reference solution, with a first-order discretization and ten times lower resolution than the reference solution.

\begin{figure} [H] 
	\centering
		  \includegraphics[trim=0cm 0cm 0cm 0cm, clip=true,width=0.355\textwidth]{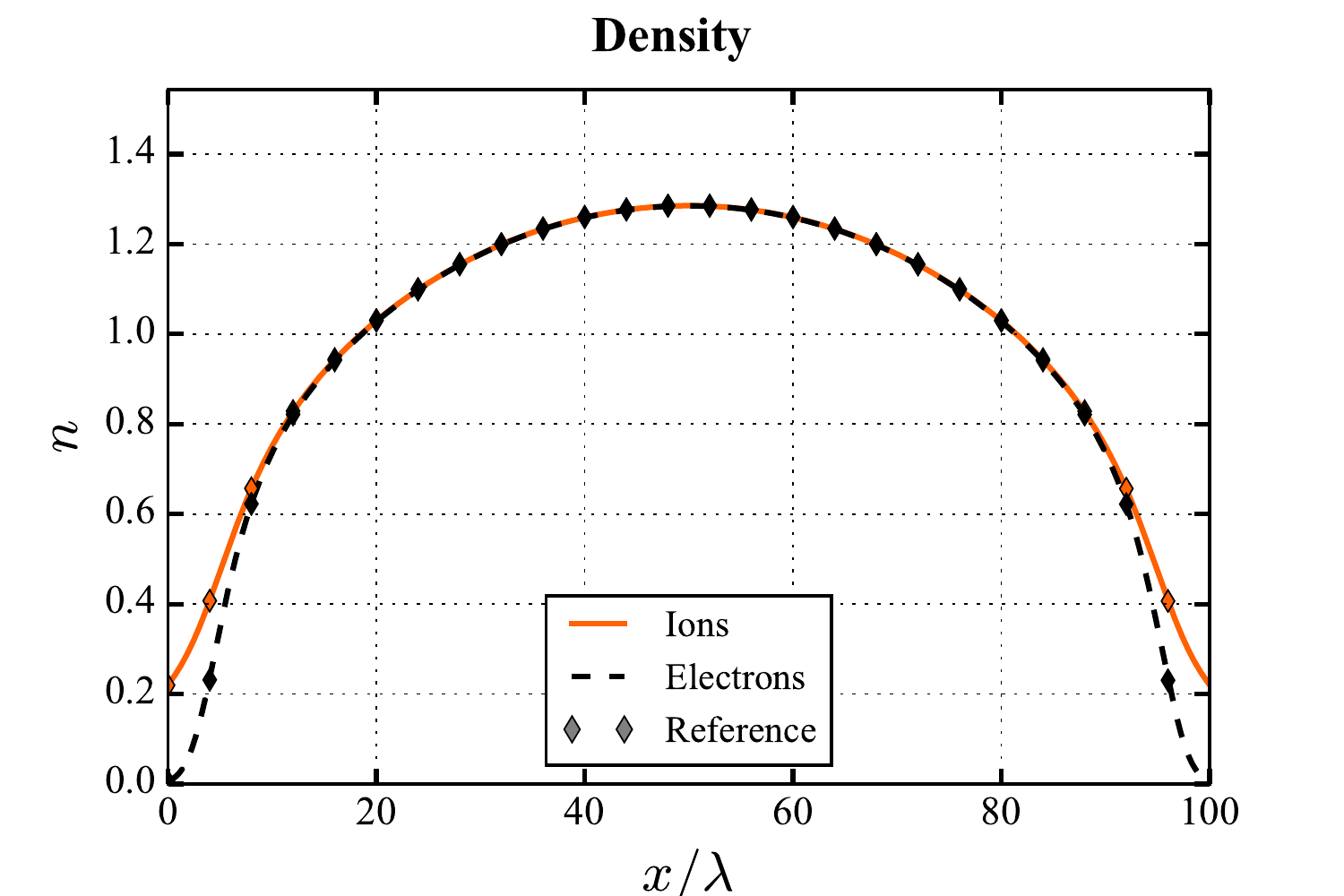} 
		  \includegraphics[trim=0cm 0cm 0cm 0cm, clip=true,width=0.355\textwidth]{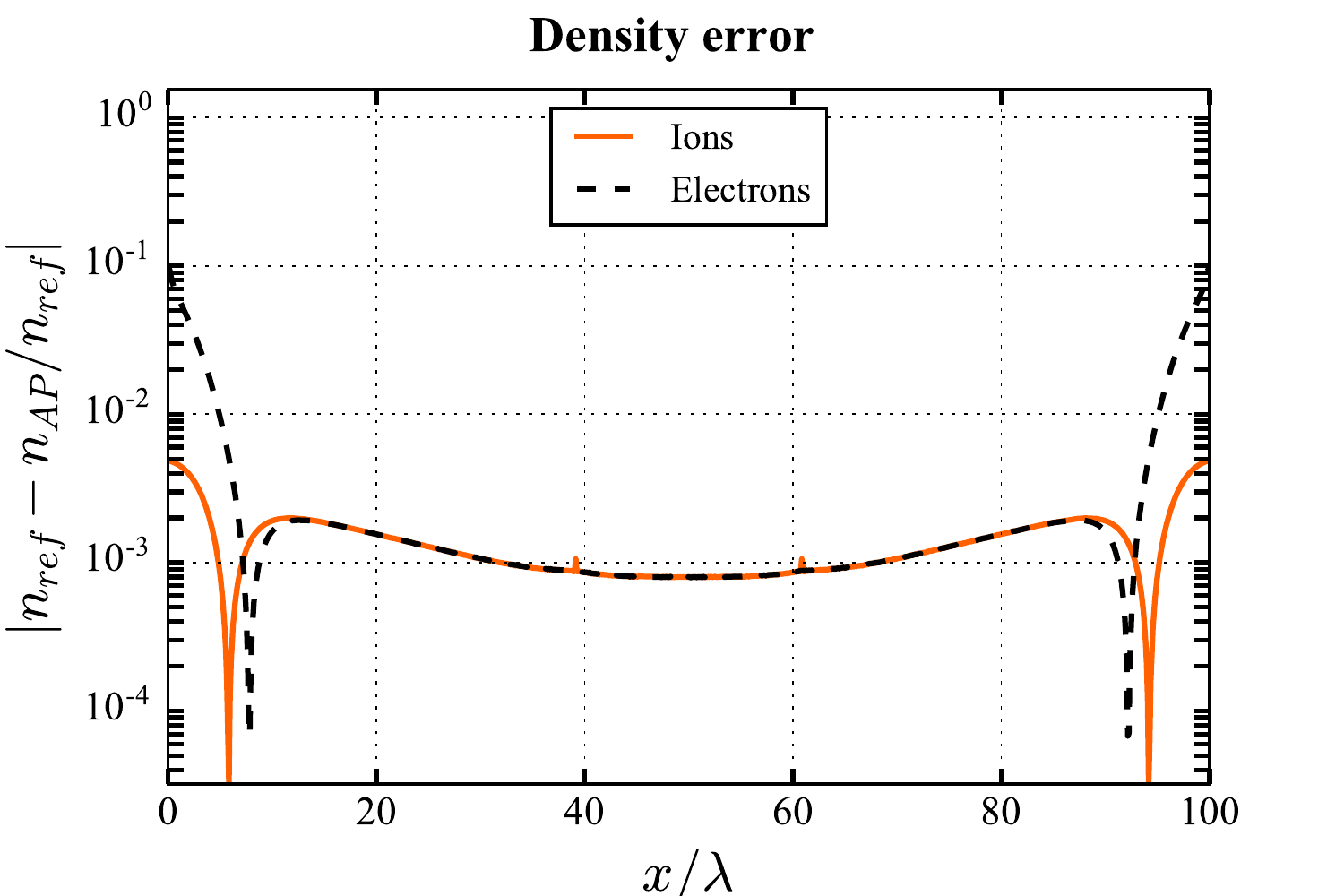}
		  \includegraphics[trim=0cm 0cm 0cm 0cm, clip=true,width=0.355\textwidth]{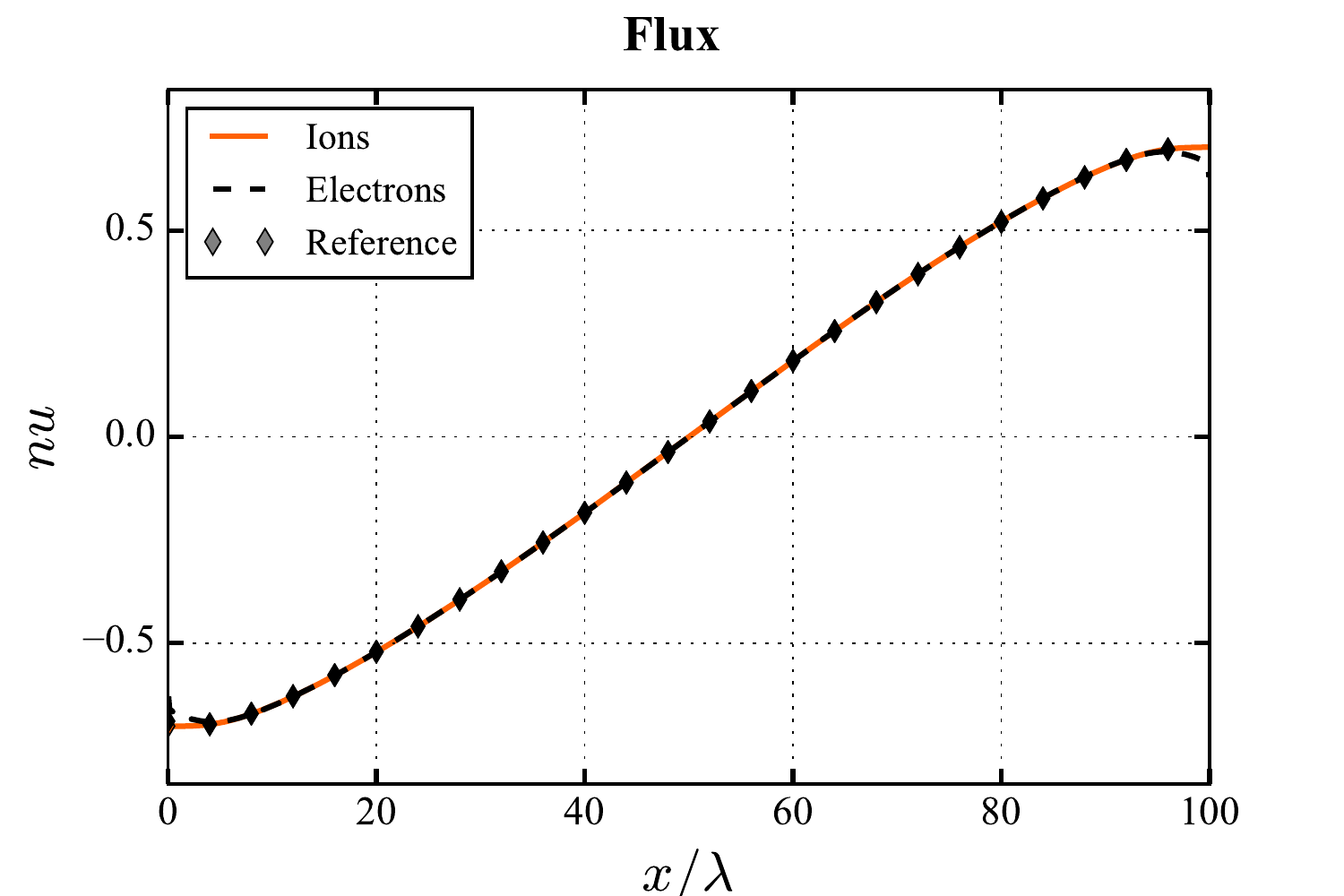} 
		  \includegraphics[trim=0cm 0cm 0cm 0cm, clip=true,width=0.355\textwidth]{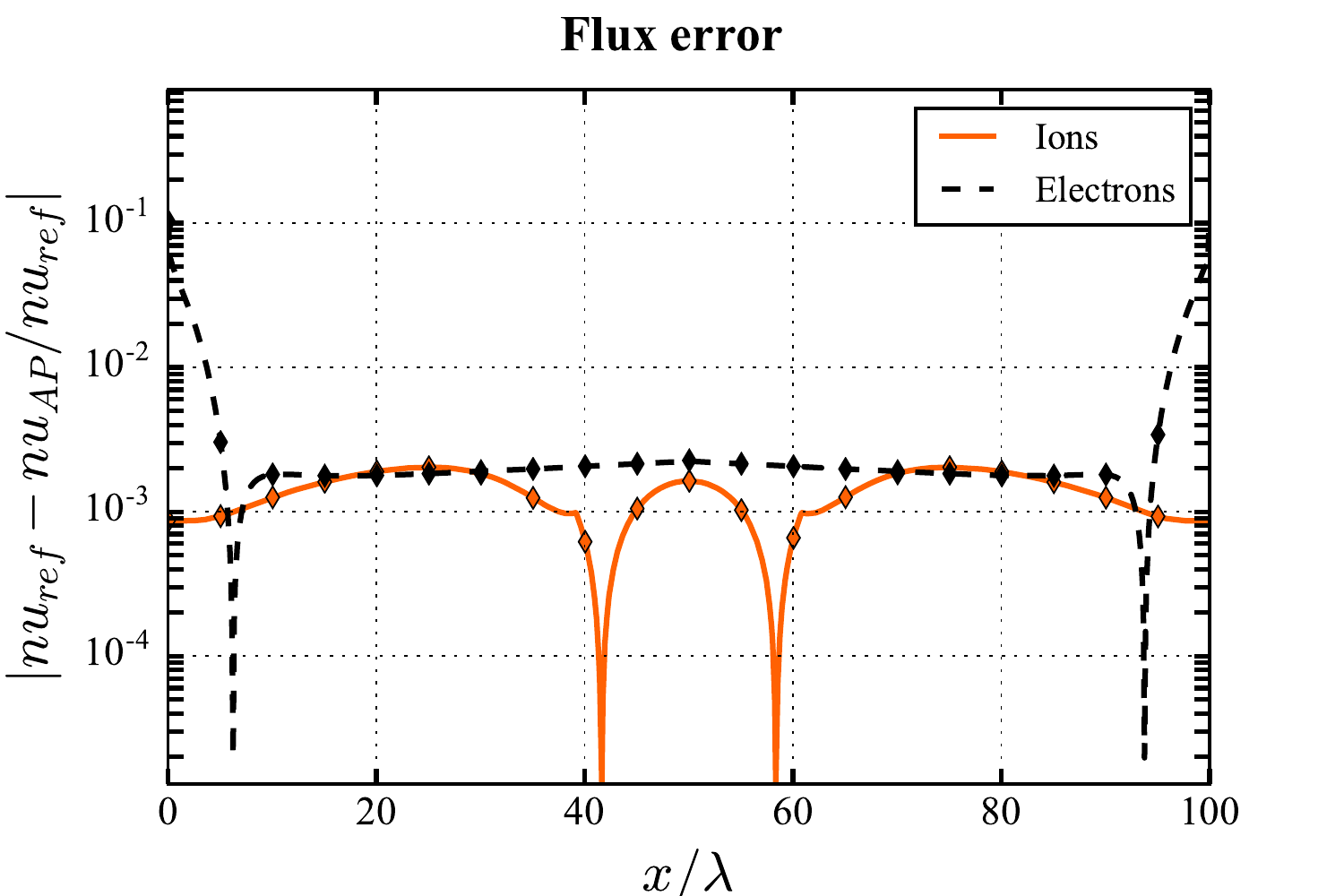}
		  \includegraphics[trim=0cm 0cm 0cm 0cm, clip=true,width=0.355\textwidth]{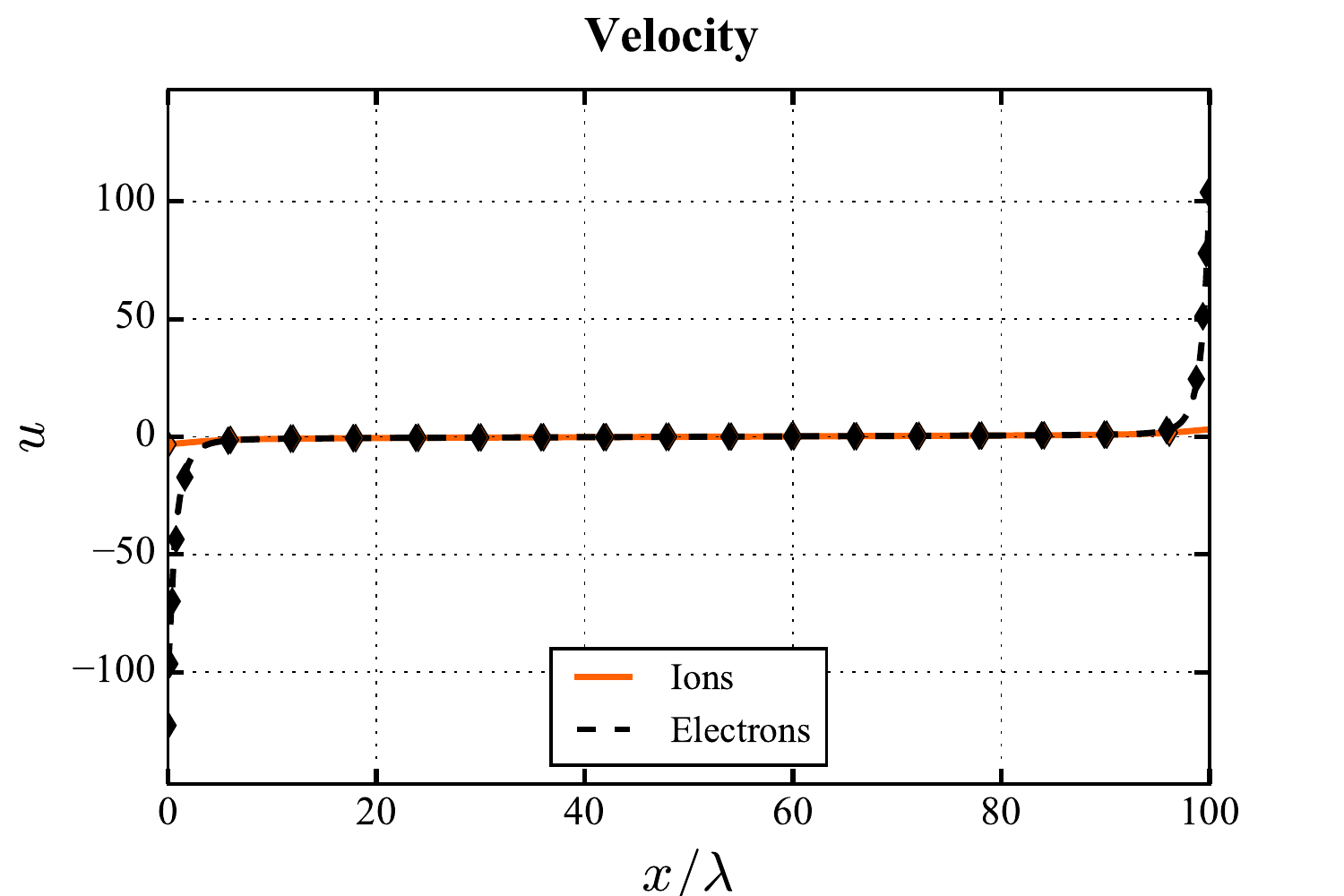} 
		  \includegraphics[trim=0cm 0cm 0cm 0cm, clip=true,width=0.355\textwidth]{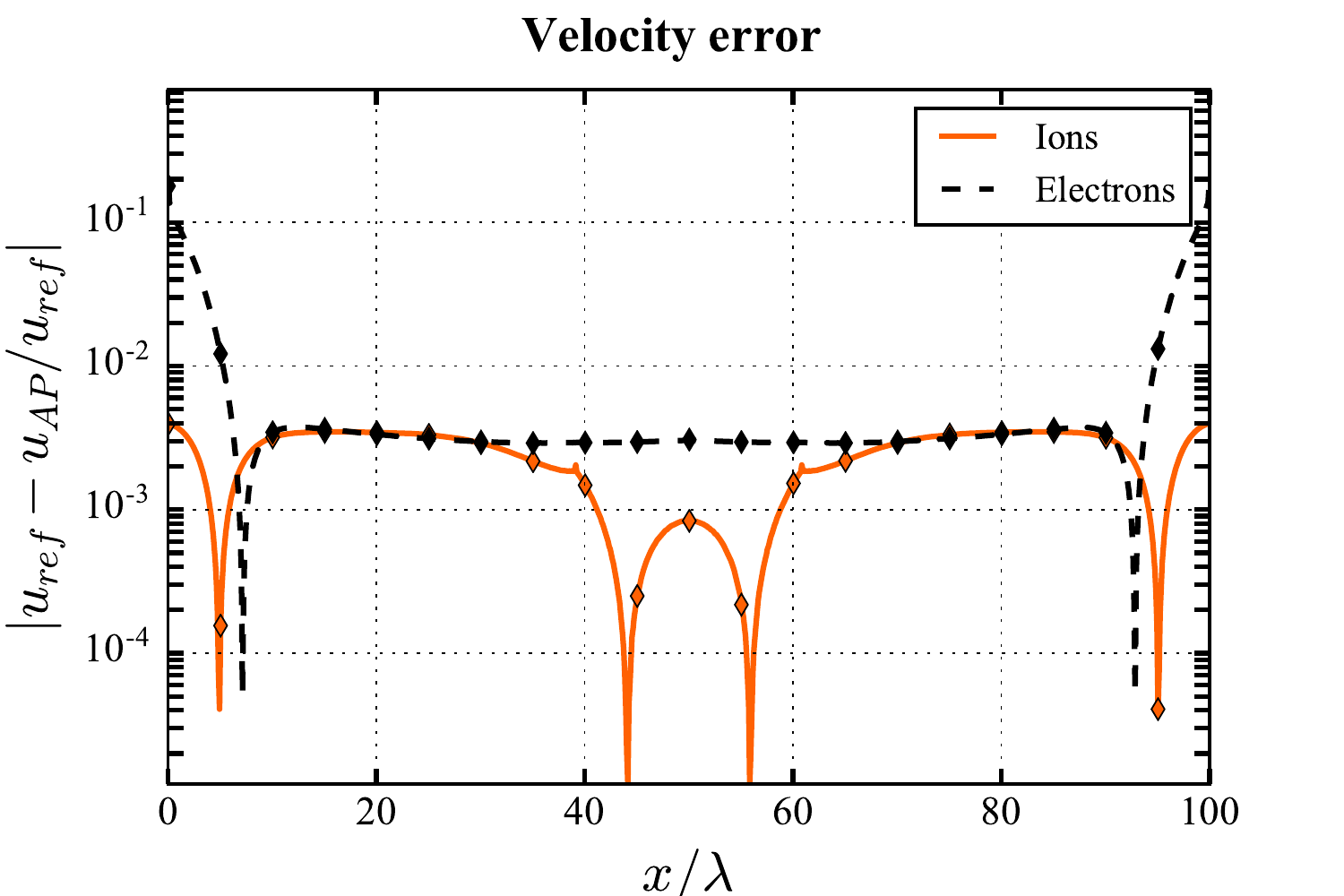}
		  \includegraphics[trim=0cm 0cm 0cm 0cm, clip=true,width=0.355\textwidth]{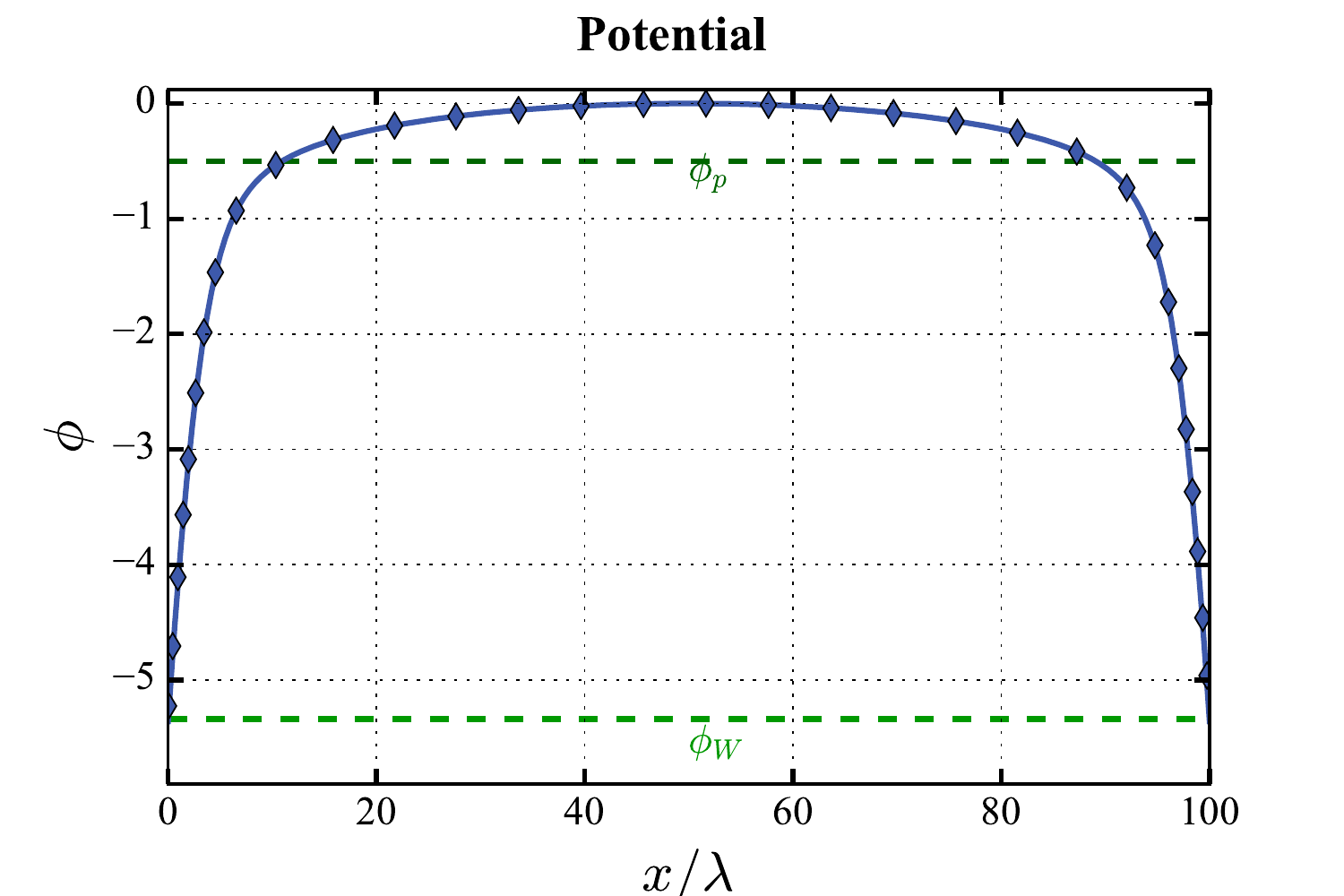} 
		  \includegraphics[trim=0cm 0cm 0cm 0cm, clip=true,width=0.355\textwidth]{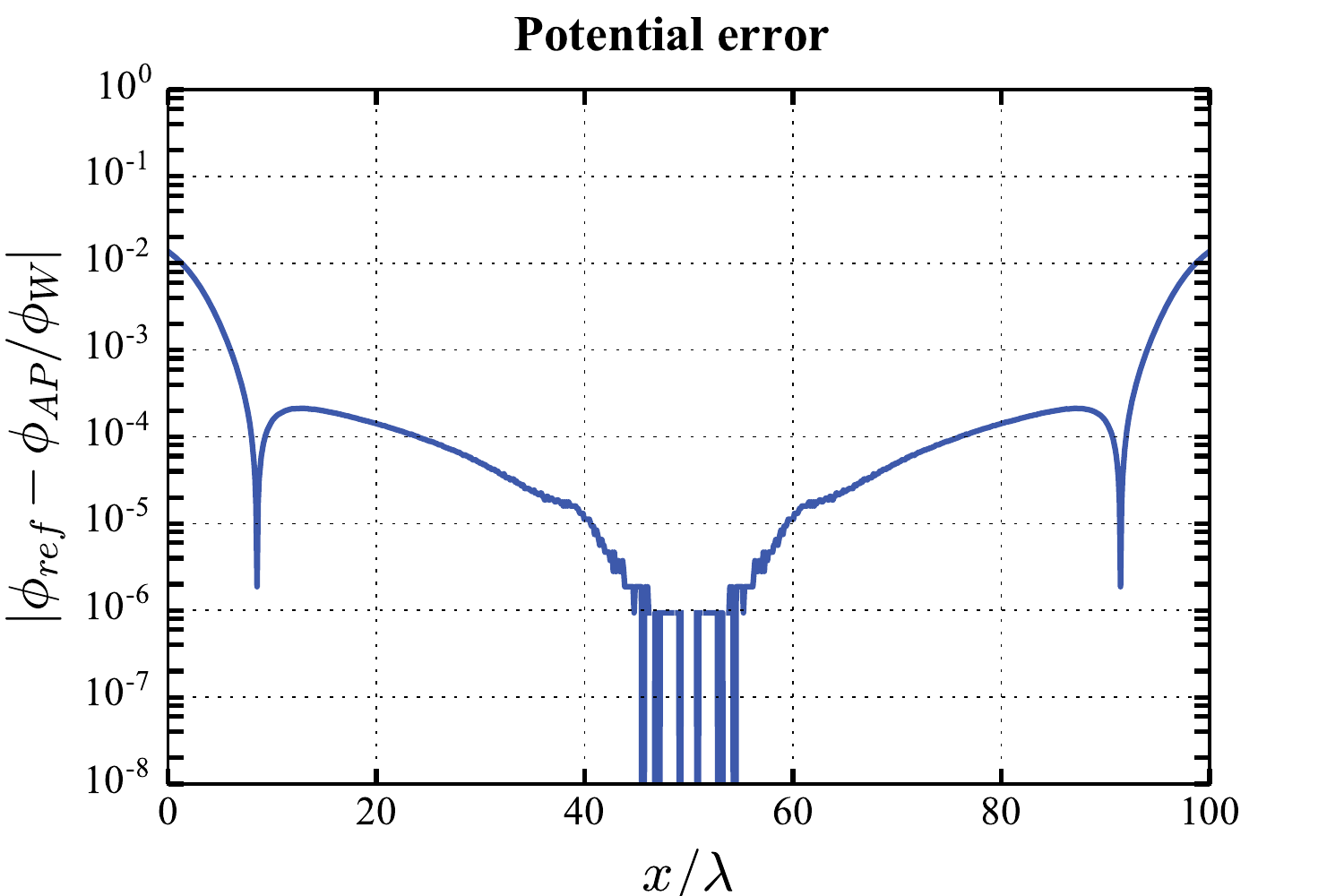}
	\caption{Collisionless isothermal sheath: Converged first-order solution with the asymptotic preserving scheme with low-Mach correction with CFL $= 0.4$ and $10^3$ mesh points. The simulation reduces dramatically the error in the electron flux.}
	\label{11_APscheme}
\end{figure}

Even though the error in the electron flux seems to not have a large impact in the other variables in Fig.~\ref{10_FirstOrderHLL}. This is only true for the steady state solution. However, if we analyze the evolution of the simulation to steady state, the solutions are radically different even though they have similar steady state solution. In Fig.~\ref{12_Transitent}, the solution at $t=75$ is presented for the reference solution, the HLL first order scheme and the AP scheme. We observe that whereas the solution of the AP scheme is very similar to the reference solution, the HLL standard scheme develops a plasma instability. This plasma instability is identified as a two-stream instability that is induced by the numerical error in the electron velocity. The instability is not present in the HLL scheme when the CFL is reduced to CFL$=0.01$ (see \cite{AlvarezLaguna18}. 

\begin{figure} [H] 
	\centering
		  \includegraphics[trim=0cm 0cm 0cm 0cm, clip=true,width=0.325\textwidth]{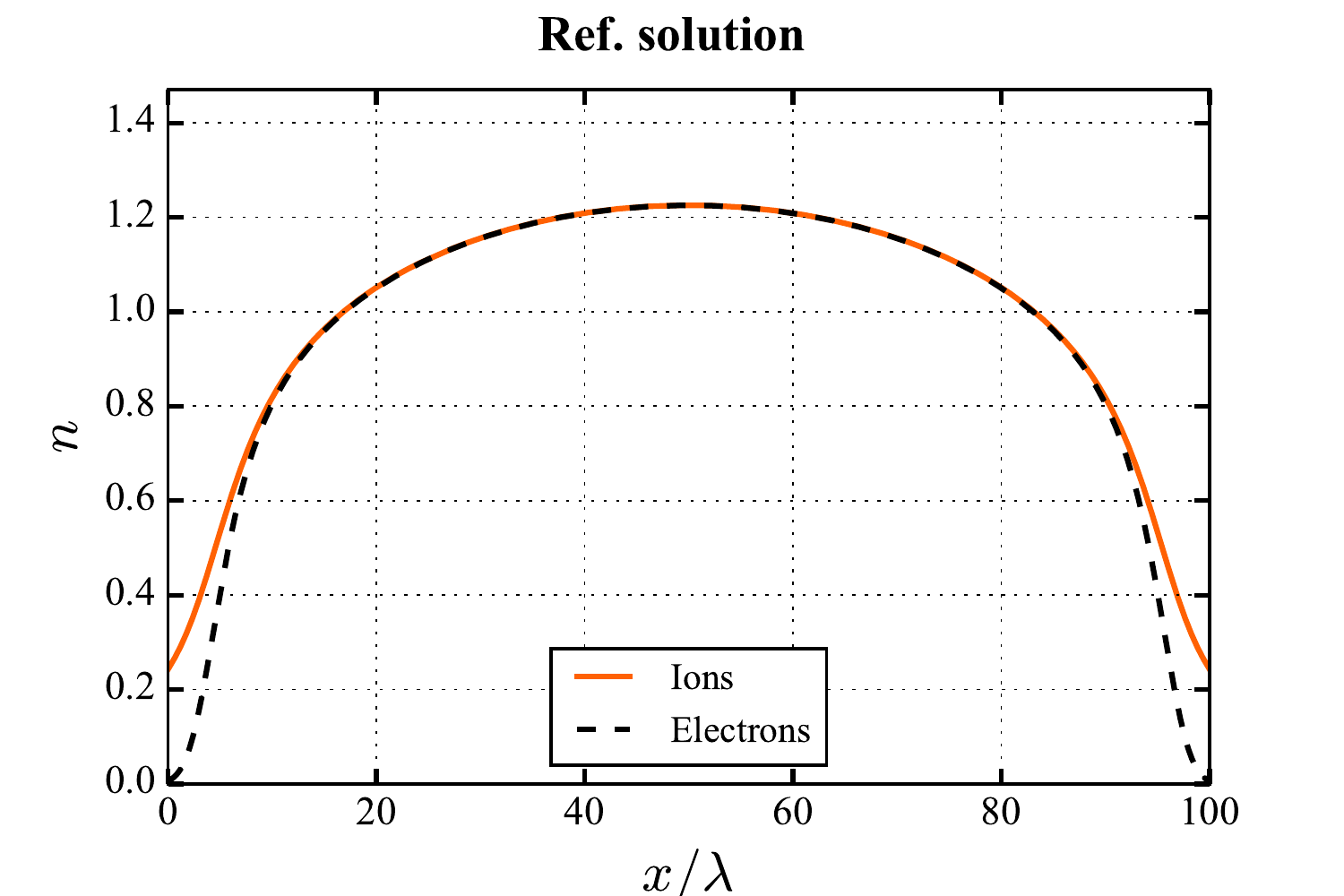} 
		  \includegraphics[trim=0cm 0cm 0cm 0cm, clip=true,width=0.325\textwidth]{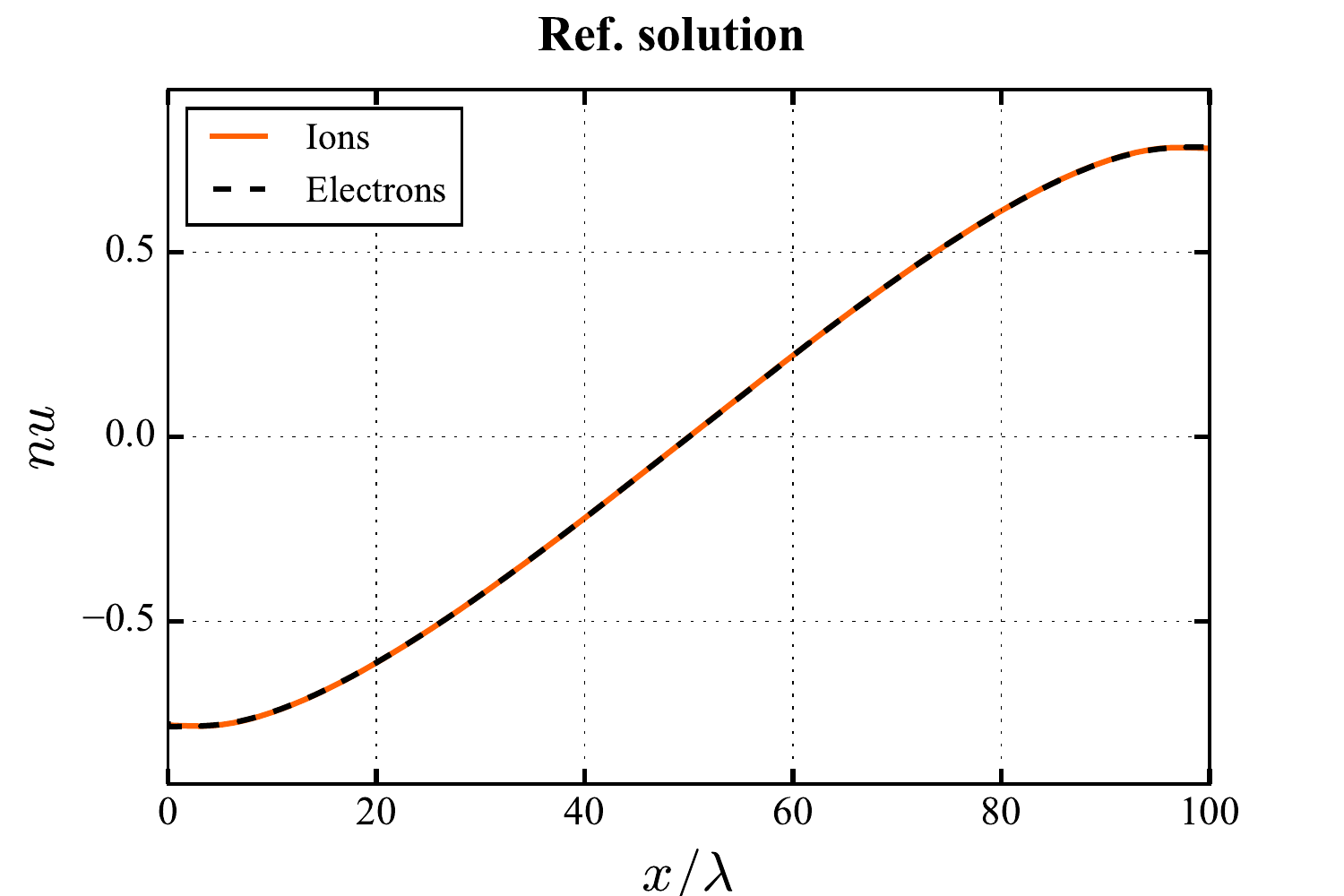}
		  \includegraphics[trim=0cm 0cm 0cm 0cm, clip=true,width=0.325\textwidth]{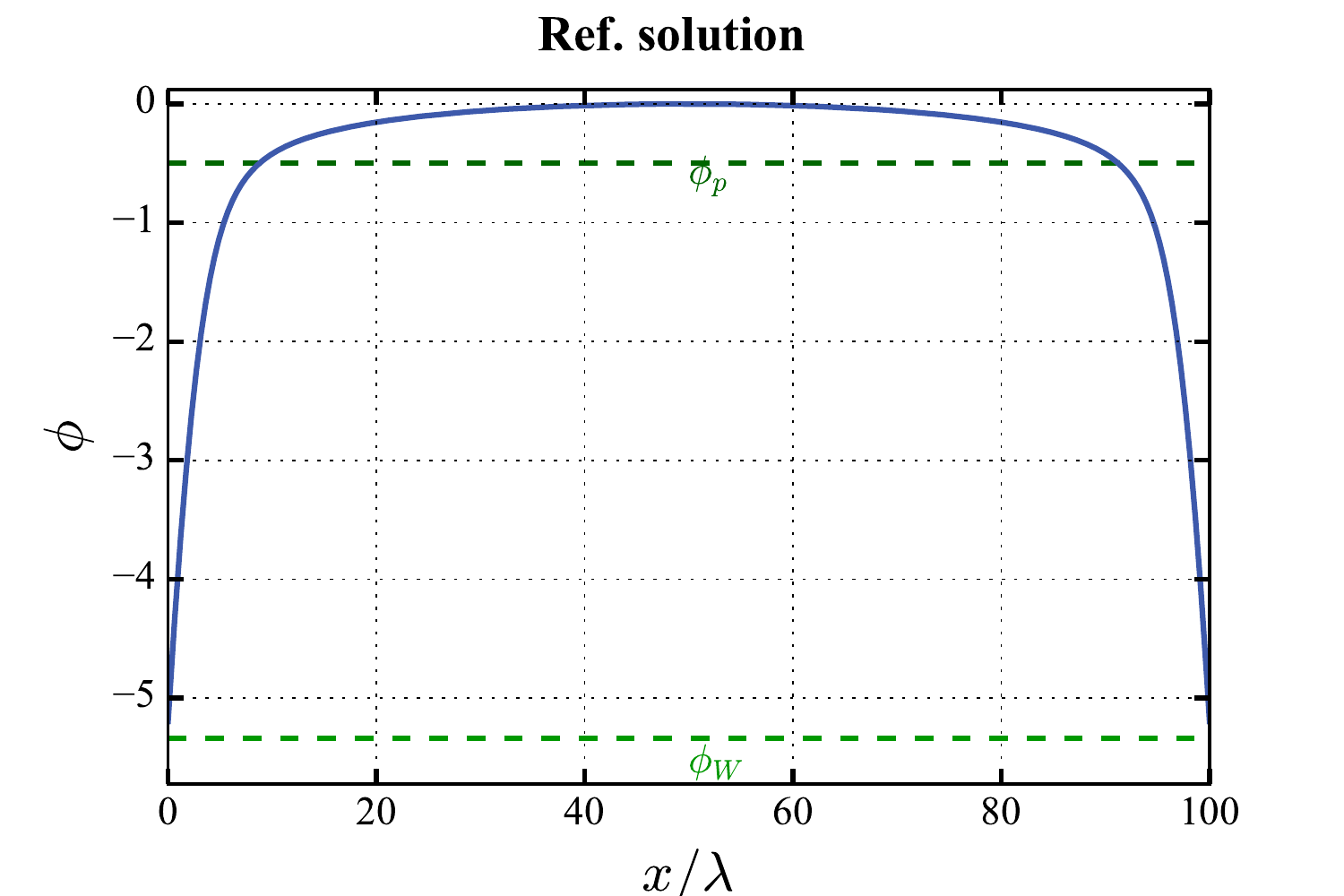}\\ 
		  \includegraphics[trim=0cm 0cm 0cm 0cm, clip=true,width=0.325\textwidth]{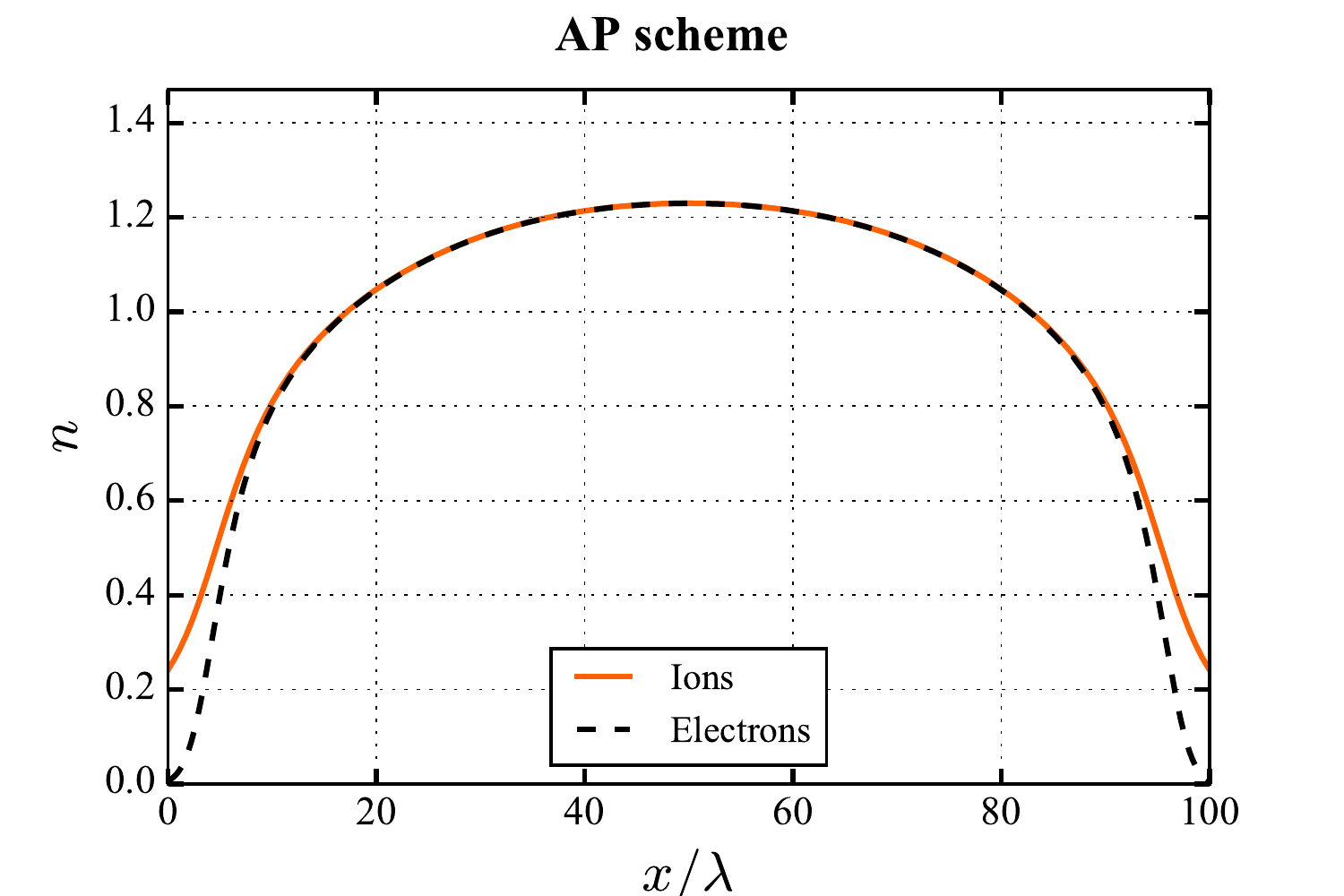}
		  \includegraphics[trim=0cm 0cm 0cm 0cm, clip=true,width=0.325\textwidth]{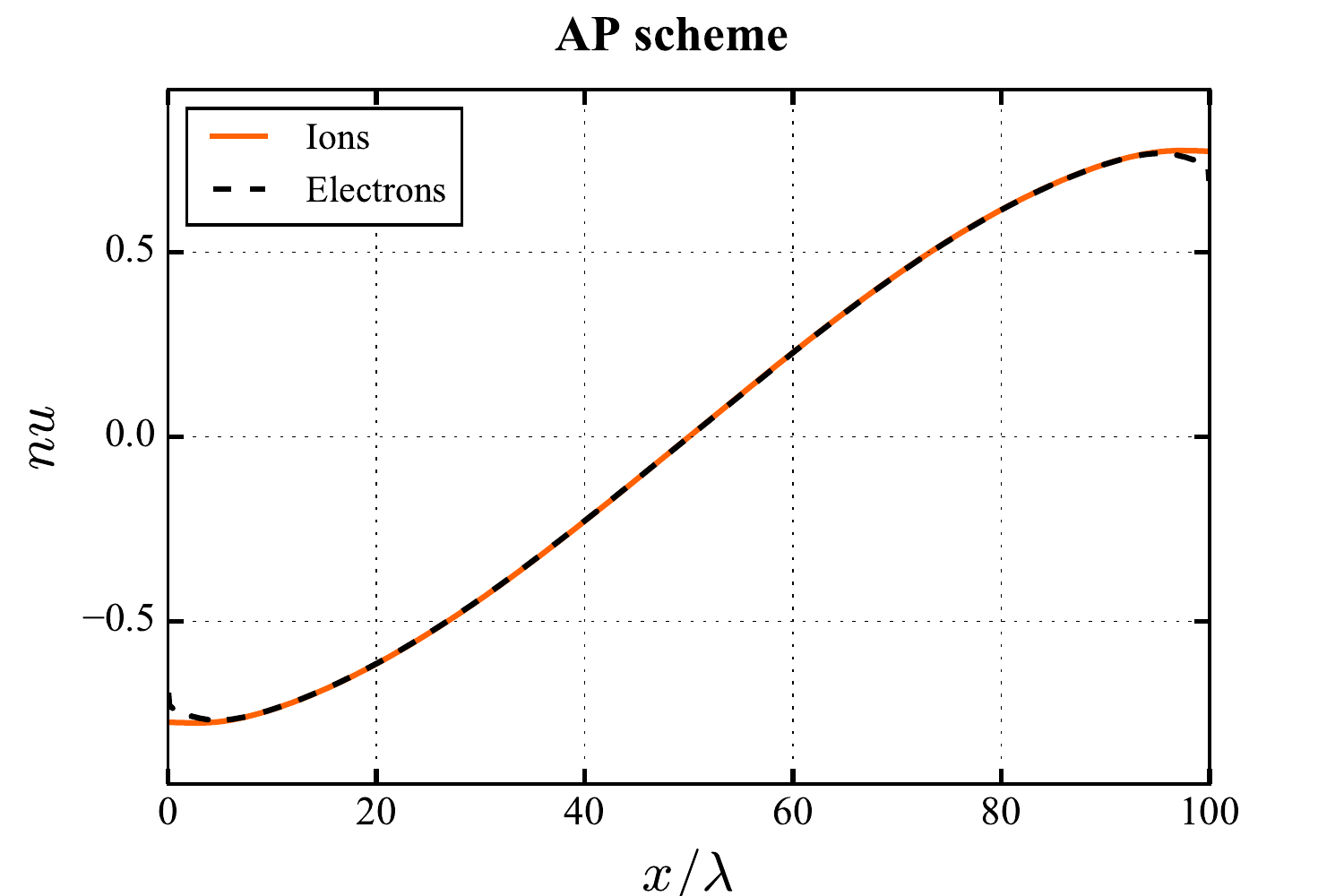} 
		  \includegraphics[trim=0cm 0cm 0cm 0cm, clip=true,width=0.325\textwidth]{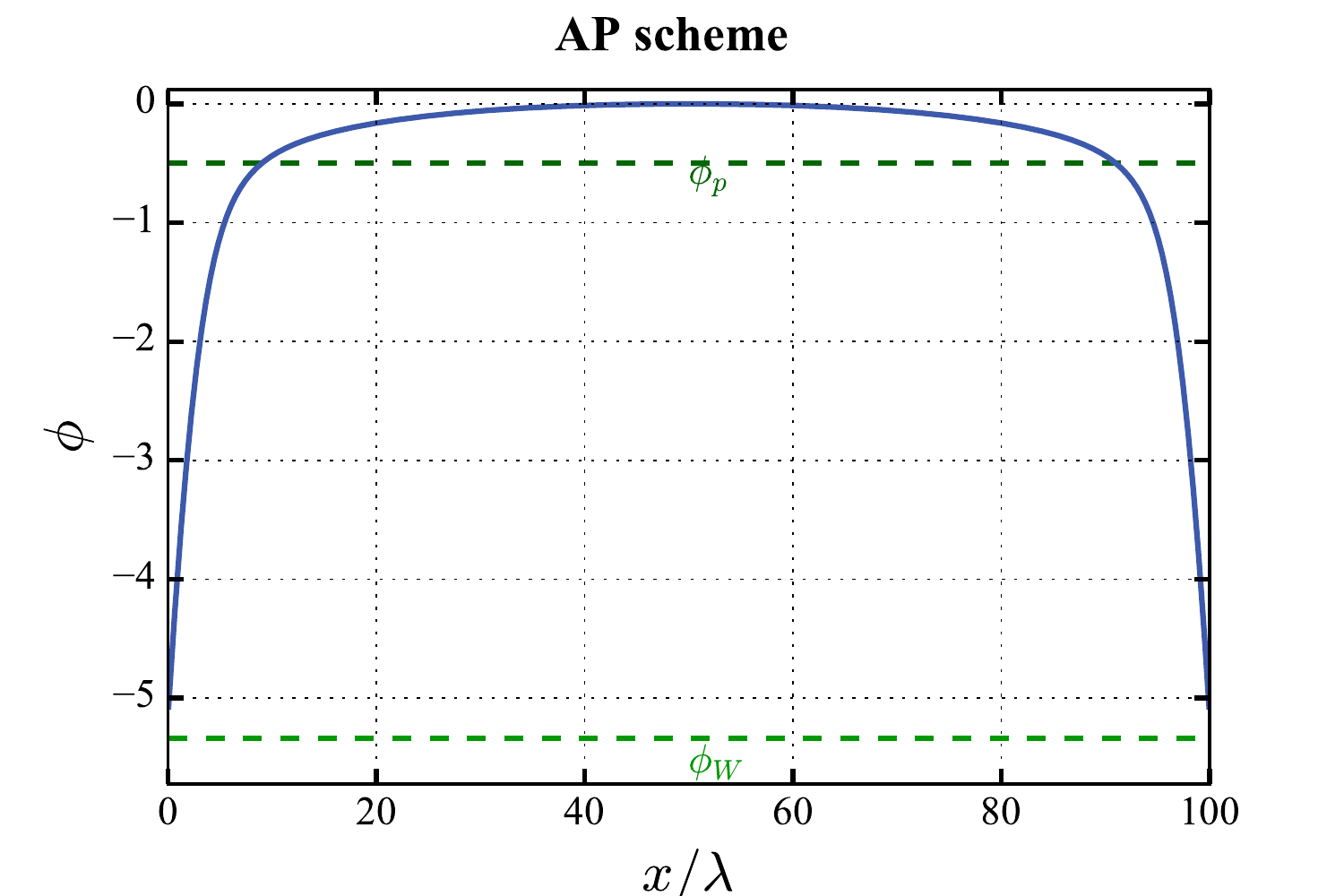}\\
		  \includegraphics[trim=0cm 0cm 0cm 0cm, clip=true,width=0.325\textwidth]{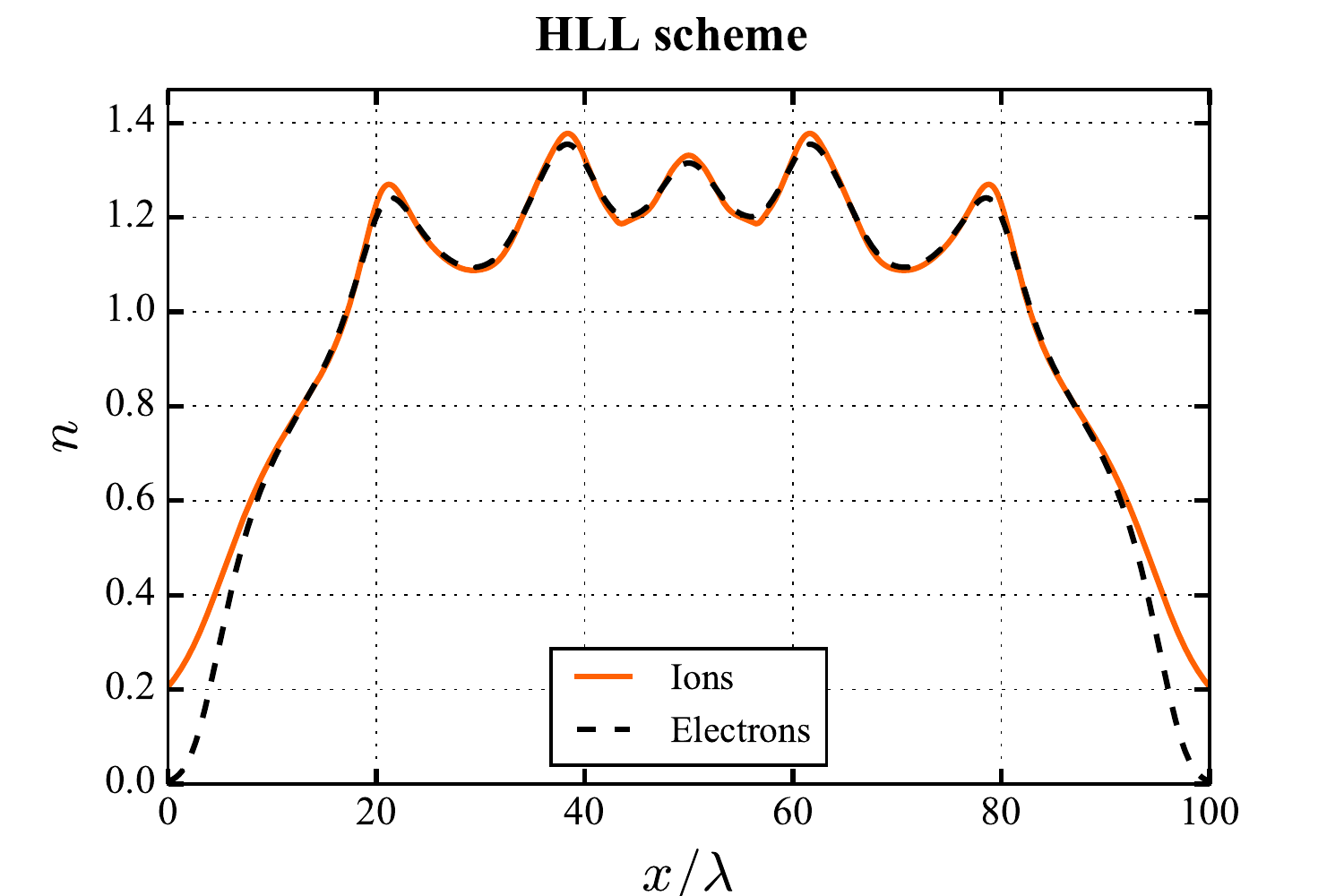}
		  \includegraphics[trim=0cm 0cm 0cm 0cm, clip=true,width=0.325\textwidth]{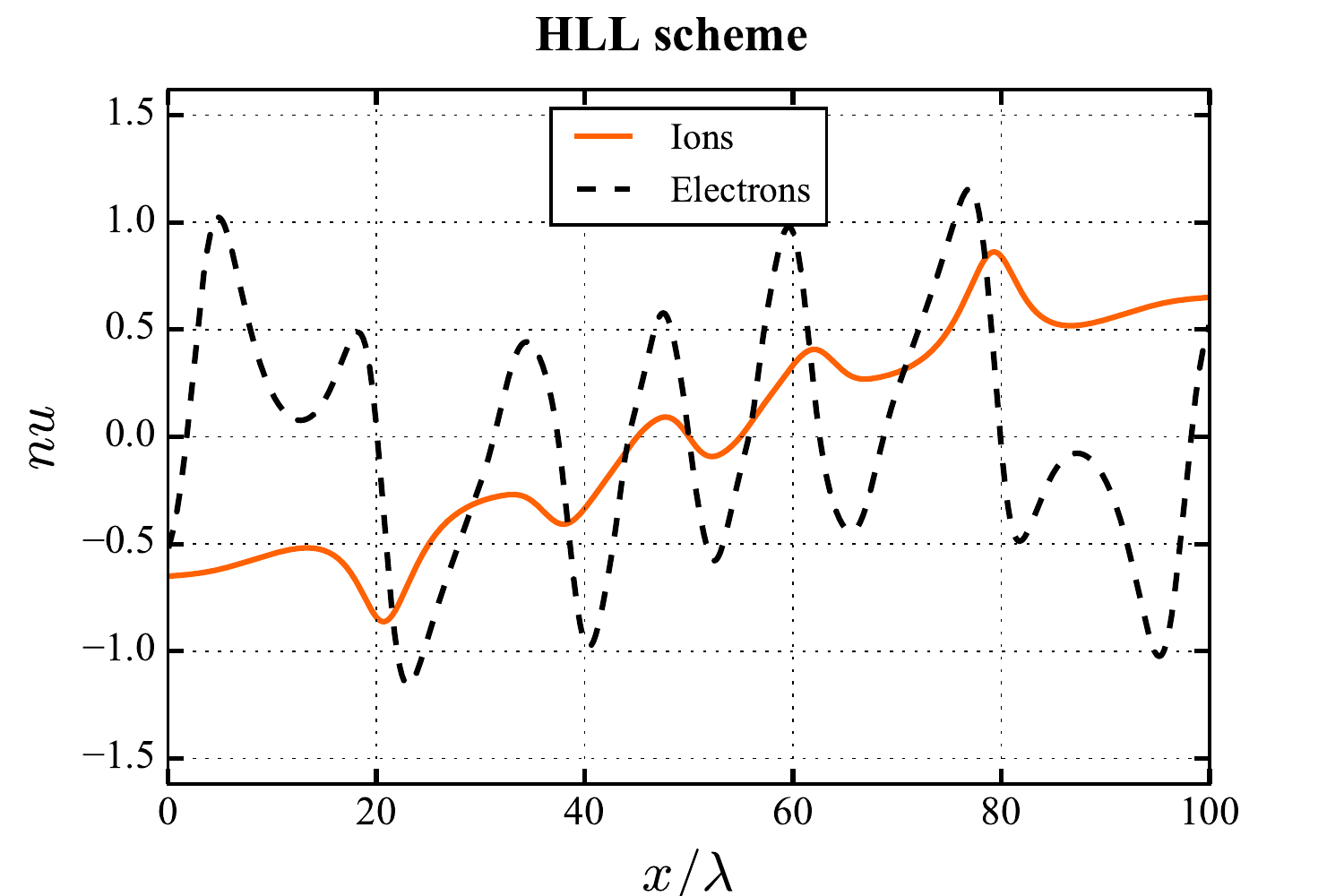} 
		  \includegraphics[trim=0cm 0cm 0cm 0cm, clip=true,width=0.325\textwidth]{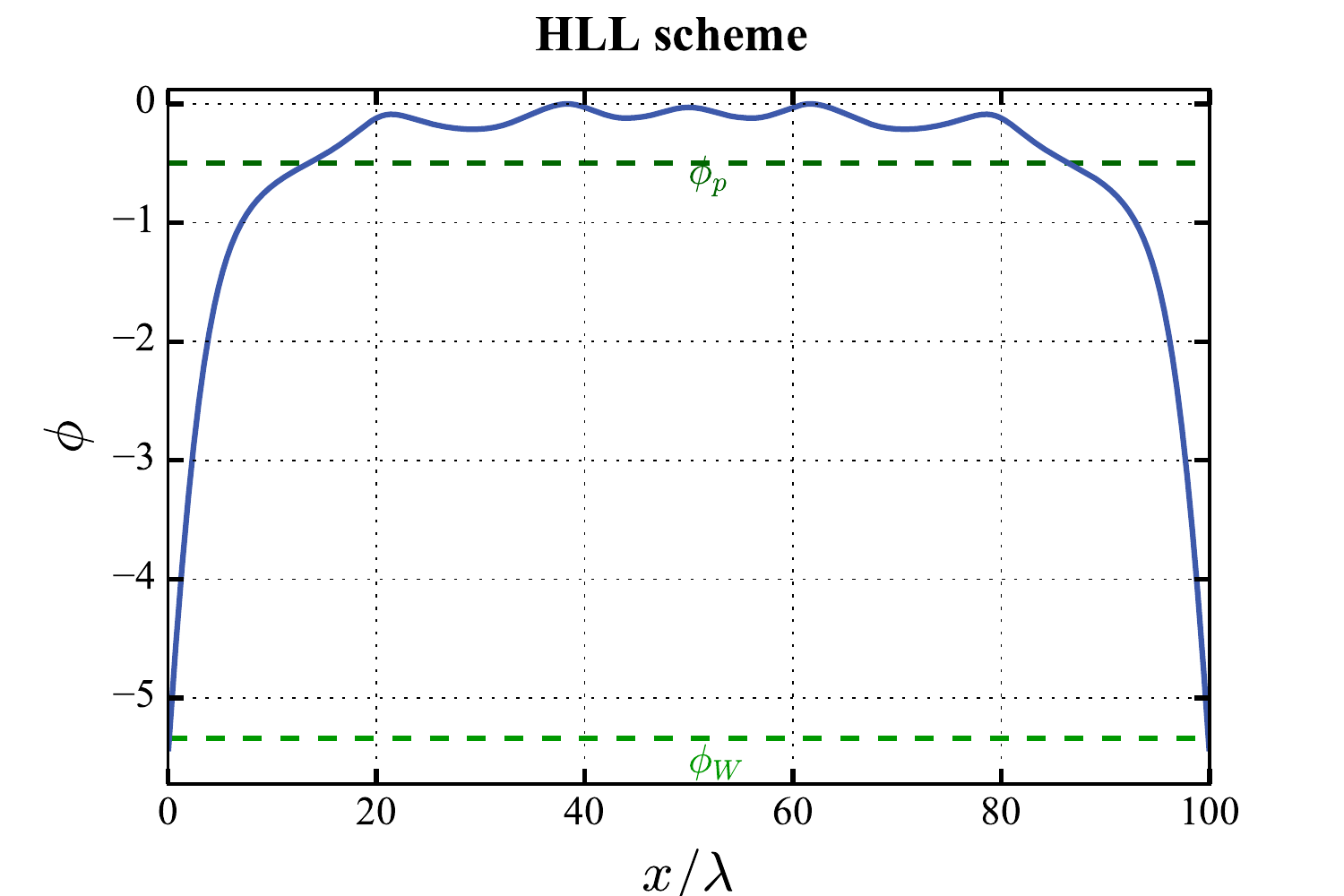}\\
	\caption{Collisionless isothermal sheath: Comparison of the HLL first-order standard discretization and the asymptotic preserving scheme with low-Mach correction discretization at $t=75$. The first-order standard discretization develops a numerical stability that is not present in the reference solution.}
	\label{12_Transitent}
\end{figure}

This numerical example illustrates that the AP scheme provides a more stable discretization that avoids spurious numerical instabilities in the quasi-neutral region, while still being able to accurately capture the regions with charge separation. Additionally, the scheme increases dramatically the accuracy of the computation of electron dynamics as well as increases by several orders of magnitude the computational efficiency. 

\section{Summary and conclusions}

In low-temperature plasma applications, the plasma-wall interaction is fundamental to explain the conditions inside the bulk of the plasma. For this reason, the plasma sheaths need to be included in the computational domain to fully explain the characteristics of the discharge. Consequently, the numerical model needs to be able to resolve the Debye length in order to capture the potential drop that forms beside the walls of the device. Similarly, the electron inertia is seen to play a fundamental role in low-pressure conditions. However, in most parts of the domain, the plasma behaves as quasi-neutral with massless electrons. A numerical method that is able to tackle accurately and efficiently these different regimes is a long-standing problem in computational plasmas. In this paper, we have presented a numerical scheme for the isothermal plasma equations coupled to Poisson's equation that proves to be accurate and stable both in the quasi-neutral and charge regimes, and both when electrons behave as massless and when electron inertia plays a role.

In Section \ref{sec:AsymBehavior}, we have discussed the asymptotic behaviour of the plasma equations when both the electron mass and the Debye length tend to zero. We find a different asymptotic behaviour as compared with previous work and that the zero-th order solution of the electric potential depends only on the electron dynamics when the electron-to-ion mass ratio is included in the asymptotic expansion. This relation is known in plasma physics as Boltzmann density distribution. Alternatively, the electron momentum equation at $\mathcal{O}(1)$ is similar the low-Mach limit of the Euler equations.

Based on this asymptotic behaviour, we have proposed a novel operator splitting strategy that relies on the Lagrange-projection scheme. Our approach solves for the electron acoustic system coupled to the Poisson equation in a first step and the electron transport together with the ion dynamics in a second step. We show that the first step preserves the asymptotic limit when both the electron mass and the Debye length tends to zero. This means that even when these scales are not resolved, the numerical scheme is able to find the correct asymptotic solution, without the need for an implicit scheme. Consequently, the scheme is able to achieve a great improvement in the accuracy of the solution as the numerical dissipation is properly scaled to the physics. More importantly, the numerical method improves dramatically the computational efficiency when these scales are not needed to be resolved. 

The numerical scheme has the advantage of solving a standard Poisson equation, as compared to other AP and semi-implicit methods that solves for more complex parabolic or elliptic equations. Additionally, the acoustic and electrostatic step involves only the electron dynamics and, therefore, the proposed model can be used in hybrid approaches where the electrons are treated as a fluid and the ions with a kinetic approach. Moreover, we have proposed a well-balanced treatment of the Lorentz force in the ions that proves to be fundamental in order to avoid spurious numerical oscillations when the temperature of the ions is much lower than this of electrons.

In the numerical results, we have simulated a quasi-neutral two-stream perturbation in thermal and low-temperature plasmas. Additionally, we have simulated a plasma sheath in a low-temperature plasma. The two-stream perturbation allowed us for assessing the performance of the model when the plasma is quasi-neutral but the electron velocity differs from this of ions. We also have proved that the scheme is able to reproduce the numerical solution even when the Debye length and the electron sound speed are not resolved. This allows us for achieving an efficient numerical scheme as compared to a scheme that needs to resolve the scales related to the Debye length or the electron acoustic waves. In our numerical experiments, the AP scheme is $10^9$ faster to simulate a plasma period of a quasineutral wave than the standard discretization that needs to resolve the Debye length and the electron plasma frequency. Additionally, the electron velocity is properly captured, avoiding the problems associated to the low-Mach limit of electrons.

Finally, we have proposed a numerical set-up to simulate the steady solution of a plasma between two floating plates. The set-up is able to find self-consistently a solution that agrees with the analytical plasma sheath theory. This solution contains three different regions: (i) the quasi-neutral plasma where the ions and electrons have the same density and velocity, (ii) the plasma-sheath transition where ions are accelerated to the Bohm's speed, and (iii) the sheath where the plasma is electrically charged and the electrons are accelerated to speeds comparable to their thermal speed. The proposed AP scheme is able to reproduce accurately these three regions. The standard discretization is not able to reproduce the electron flux as the standard numerical scheme is not able to reproduce the limit when the electron inertia is very small, even when the Debye length is properly resolved. Additionally, the standard scheme develops non-linear instabilities that are triggered by the numerical error in the electron velocity.

In the present paper, we have discussed the fundaments of an asymptotic preserving scheme for a fluid plasma model in low-temperature plasma conditions. We have highlighted that the ideas presented in this paper are the basis for future developments, extending the scheme to multi-dimensions, the inclusion of the energy equation, and the development of high-order schemes including this splitting strategy. Similarly, the scheme presented in the present paper would be specially advantageous for a non-uniform mesh where the quasi-neutral region of the plasma does not need to be resolved, whereas the boundaries resolve the Debye length in order to capture the plasma sheath.

\section*{Acknowledgements}
The first author is funded by the postdoctoral
fellowship from the {\it Fondation Mathematique Jacques Hadamard (FMJH)}, LabEx Math\'ematique Hadamard  - ANR11-LABX-0056 and LabEX LASIPS - ANR10-LABX-0032. The authors would like to thank Dr.~Nagi N.~Mansour and Dr.~Samuel Kokh for very useful discussions during the 2018 NASA Summer Program at NASA Ames Research Center. T.~M.~ is supported by a ``Chaire Jean d'Alembert" of University Paris Saclay at Ecole Polytechnique.
 \newpage

\section*{Appendix A: Approximate Riemann solver of the acoustic step}

We can write the system of eq.~\eqref{acoustic_NewVar_m} as 
\begin{equation}
\dtb\left(\begin{array}{c}
  \tauelec\\
  \ueb
\end{array}\right)
+ 
\dm
\left(\begin{array}{c}
 -\ueb\\
  \varepsilon^{-1}\tauelec^{-1}
\end{array}\right)
=
\left(\begin{array}{c}
  0\\
  \varepsilon^{-1}\dxb\phi
\end{array}\right).\label{acoustic_conservative}
\end{equation}
Note that the pressure is written as function of the conservative variable $\tauelec$, i.e., $\pelec = \varepsilon^{-1}\rhoeb = \varepsilon^{-1}\tauelec^{-1}$. We define the conservative variables of the electron acoustic step, the flux term, and the source term as
\begin{equation}
\UelecAcous = \left(\tauelec,~\ueb\right)^T,~~~\FelecAcous = \left(-\ueb,~\varepsilon^{-1}\tauelec^{-1}\right)^T,~~~\text{and}~~~\SelecAcous = \left(0,~\varepsilon^{-1}\dxb\phi\right)^T.
\end{equation}

The Jacobian matrix of the flux term in the LHS of eq.~\eqref{acoustic_conservative} reads
\begin{equation}
  \frac{\partial\FelecAcous}{\partial \UelecAcous} = \left[ \begin{array}{c c}
    0 & -1\\
    -\varepsilon^{-1}\tauelec^{-2} & 0
 \end{array} \right].
\end{equation}
The eigenvalues of this matrix are $\lambda = \left(-\rhoeb\varepsilon^{-1/2},~ \rhoeb\varepsilon^{-1/2}\right)$. Consequently, the Lax-Friedrich numerical flux for the previous Riemann problem reads
\begin{equation}
\FNumElecAcous(\UelecAcousL,\UelecAcousR) = \left(\begin{array}{c}
 -\left[\frac{\uebR + \uebL}{2} - \frac{ \nelecbjphalf\varepsilon^{-1/2}}{2}\left(\taueR - \taueL\right)\right]\\
  \varepsilon^{-1}\frac{\left(\rhoebR + \rhoebL\right)}{2} - \frac{\nelecbjphalf\varepsilon^{-1/2}}{2}\left(\uebR - \uebL\right)
\end{array}\right)\label{LFNumAcoustic}
\end{equation}
where the density at the interface is computed as $\nelecbjphalf = \left(\rhoebR + \rhoebL\right)/2$. 

\section*{Appendix B: Truncation error of the acoustic system}\label{ap:truncationError}

We study the truncation error of the upwind scheme for the acoustic step. We consider that the variables describe a smooth flow so that we can expand them in Taylor series the conservative variables. By using the Lax-Friedrich scheme of eq.~\eqref{LFNumAcoustic}, we obtain the following approximations in the mass equation

\begin{itemize}
  \item[*] $\frac{\tauelec^{\nPlusOneM} - \tauelec^n}{\Delta t} \approx \dtb \tauelec(t^n,x_j) + \dttb  \tauelec(t^n,x_j) \Delta t + \mathcal{O} (\Delta t^2)$
  \item[*] $[\dm  \ueb]^n_j \approx \frac{1}{\rhoeb}\dxb \ueb(t^n,x_j) + \varepsilon^{-1/2}\dxxb \tauelec(t^n,x_j) \Delta x + \mathcal{O} (\Delta x^2) + \mathcal{O} ( \varepsilon^{-1/2}\Delta x^3)$
    \item[*] $[\dmm \pelecb]^n_j \approx \frac{1}{\varepsilon\rhoeb^2} \dxxb\rhoeb(t^n,x_j) + \mathcal{O} ( \varepsilon^{-1}\Delta x^2)$
\end{itemize}

Therefore, the truncation error of the discretization is

\begin{equation}
  [\text{Trunc. Error}]_{\tauelec} = \left|\left(\varepsilon^{-1/2} - \text{CFL}\right)\dxxb \tauelec \Delta x + \frac{\Delta t}{\varepsilon\lambdaSq}\left( 1 - \frac{\rhoib}{\rhoeb}\right) - \frac{\Delta t}{\varepsilon\rhoeb^2}\dxxb \rhoeb + \mathcal{O}(\varepsilon^{-1}\Delta x^2 \Delta t) + \mathcal{O}(\Delta x^2) + \mathcal{O} ( \varepsilon^{-1/2}\Delta x^3)\right|
\end{equation}

Similarly, we obtain the following approximations for the equation for the velocity, expanded in Taylor series
\begin{itemize}
  \item[*] $\frac{\ueb^{\nPlusOneM} - \ueb^n}{\Delta t} \approx \dtb \ueb(t^n,x_j) + \dttb  \ueb(t^n,x_j) \Delta t + \mathcal{O} (\Delta t^2)$
  \item[*] $[\dm  \pelecb]^n_j \approx \frac{\varepsilon^{-1}}{\rhoeb}\dxb \rhoeb(t^n,x_j) + \varepsilon^{-1/2}\dxxb \ueb(t^n,x_j) \Delta x + \mathcal{O} (\Delta x^2) + \mathcal{O} ( \varepsilon^{-1/2}\Delta x^3)$
    \item[*] $[\dx \phib]_j^{\nPlusOneM} \approx \frac{1}{\varepsilon} \dxb\phib(t^n,x_j) + \mathcal{O} ( \varepsilon^{-1}\Delta x^2)$ 
\end{itemize}
Consequently, the truncation error of the velocity equations reads
\begin{equation}
  [\text{Trunc. Error}]_{\ueb} = \left|\left(\varepsilon^{-1/2} - \text{CFL}\right)\dxxb \ueb \Delta x + \mathcal{O} ( \varepsilon^{-1}\Delta x^2) \right|
\end{equation}

\footnotesize
\bibliographystyle{unsrt}
\bibliography{references}\label{app:bibliography}
\end{document}